\documentclass[10pt]{article}
\usepackage{algpseudocode}
\usepackage{amsmath, amssymb, amsthm}
\usepackage{float, fullpage, graphicx, multirow,parskip, subcaption, setspace}
\usepackage{comment}
\usepackage{url}
\usepackage{enumitem}
\usepackage{tikz}
\usepackage{bm, bbm}
\usetikzlibrary{shapes.geometric, positioning}
\usetikzlibrary{quotes, angles}
\usepackage{rotating}
\usepackage{hyperref}
\usepackage{natbib}
\usepackage{placeins}
\usepackage{algorithm2e}
\RestyleAlgo{ruled}

\definecolor{SkyBlue}{RGB}{14, 118, 188}
\definecolor{BrightRed}{RGB}{223,82, 78}

\hypersetup{pdfborder = {0 0 0.5 [3 3]}, colorlinks = true, linkcolor = BrightRed, citecolor = SkyBlue}

\DeclareMathOperator{\tr}{tr}

\DeclareMathOperator{\vect}{vec}

\DeclareMathOperator{\chol}{chol}
\DeclareMathOperator{\vech}{vech}

\newtheorem{myProposition}{Proposition}

\newtheorem{myRemark}{Remark}

\newcommand{\revision}{\textcolor{black}}

\usepackage{grffile}

\def\keywordname{{\bfseries \emph Keywords}}%
\def\keywords#1{\par\addvspace\medskipamount{\rightskip=0pt plus1cm
\def\and{\ifhmode\unskip\nobreak\fi\ $\cdot$
}\noindent\keywordname\enspace\ignorespaces#1\par}}

\makeatletter
\patchcmd{\@maketitle}{\LARGE}{\Large}{}{}
\makeatother

\usepackage{sectsty}
\sectionfont{\fontsize{12}{15}\selectfont}
\subsectionfont{\fontsize{12}{15}\selectfont}

\title{The Effect of the Prior and the Experimental Design on the Inference of the Precision Matrix in Gaussian Chain Graph Models}
\author{Yunyi Shen\thanks{Depts. of Statistics \& Wildlife Ecology, University of Wisconsin--Madison}\and Claudia Sol\'{i}s-Lemus\thanks{Wisconsin Institute for Discovery \& Dept. of Plant Pathology, University of Wisconsin--Madison. Correspondence to: solislemus@wisc.edu}}

\date{}

\begin{document}
\def\bY{\bm{Y}}
\def\by{\bm{y}} 

\def\bz{\bm{z}}
\def\bX{\bm{X}}
\def\bx{\bm{x}} 

\def\R{\mathbb{R}}
\def\N{\mathcal{N}}
\def\P{\mathbb{P}}
\def\E{\mathbb{E}}

\def\Xcal{\mathcal{X}}

\maketitle


\begin{abstract}
\begin{spacing}{1}
Here, we investigate whether (and how) experimental design could aid in the estimation of the precision matrix in a Gaussian chain graph model, especially the interplay between the design, the effect of the experiment and prior knowledge about the effect. Estimation of the precision matrix is a fundamental task to infer biological graphical structures like microbial networks.
We compare the marginal posterior precision of the precision matrix under four priors: flat,  conjugate Normal-Wishart, Normal-MGIG and a general independent. Under the flat and conjugate priors, the Laplace-approximated posterior precision is not a function of the design matrix rendering useless any efforts to find an optimal experimental design to infer the precision matrix. In contrast, the Normal-MGIG and general independent priors do allow for the search of optimal experimental designs, yet there is a sharp upper bound on the information that can be extracted from a given experiment. We confirm our theoretical findings via a simulation study comparing i) the KL divergence between prior and posterior and ii) the Stein's loss difference of MAPs between random and no experiment. Our findings provide practical advice for domain scientists conducting experiments to better infer the precision matrix as a representation of a biological network.
\end{spacing}
\end{abstract}

\section{Introduction}

\subsection{Background}
Networks are graphical structures that arise on many biological applications. For example, microbial networks are a graphical representation of a microbial community where nodes represent microbial taxa and edges correspond to some form of interaction \citep{matchado2021network}. Microbial networks arise in microbiome studies from soil \citep{barberan2012using} to human gut \citep{baldassano2016topological, claesson2012gut}. Other biological networks include brain networks where nodes correspond to regions in the brain (or voxels) and edges correspond to connections between brain regions \citep{rubinov2010complex}, or ecological networks which, like microbial networks, represent some community of species in the wild (e.g. food web \citet{pimm1991food}).

Many biological networks need to be estimated from data and given their broad applicability, there has been a plethora of network methods to reconstruct biological networks from a wide variety of data types: from microbiome sequencing data \citep{jovel2016characterization} to MRI images of brain regions \citep{van2013structure}. Many network methods used in microbiome studies, however, estimate a network under \textit{static conditions}. That is, researchers would estimate the human gut microbial network by sequencing reads of gut samples from different individuals \textit{without incorporating information about treatments on the individuals}. 
In the absence of any treatments or disturbers of the network structure, Gaussian graphical models are among the most widely used models to estimate such networks \citep{layeghifard2017disentangling}.

Many researchers, however, want to use information of different experimental treatments to better infer the biological network. For example, we can get samples of human gut microbiome for control patients and for patients under antibiotic treatment. The different responses (abundances of microbes) on different treatments can provide information on the microbial network structure by, for example, eliminating key players in the community like in the case of antibiotic treatments.
The conditions are no longer static (no treatment). Instead, the nodes are perturbed by experimental treatments, and these perturbations provide information on the network structure itself. We thus highlight that experimental treatments are incorporated into the model to better estimate the network structure, and not 
because there is an interest in the effect of these treatments on the responses. On the contrary, exact experimental treatment effects are nuisance parameters and our only focus is the estimation of the network structure with the help that the experimental treatments provide.

\subsection{The Gaussian chain graph model}
\label{themodel}
Standard Gaussian graphical models no longer apply when there are experimental treatments affecting the network structure. Thus, chain graph models arise as a suitable alternative for the ``network under treatment" setting \citep{lauritzen2002chain}. More formally, a Gaussian chain graph model with $k$ response nodes ($\mathbf Y_i\in \mathbb{R}^{ k}$) and $p$ predictor nodes ($\mathbf X_i\in \mathbb{R}^{p}$) is given by
\begin{equation}
\label{eqn:general_model}
    \mathbf Y_i \mid \mathbf{X}_i, \mathbf B, \mathbf \Omega \sim \mathcal{N}(\mathbf \Omega^{-1} \mathbf B^T \mathbf X_i,\mathbf \Omega^{-1})
\end{equation}
where $\mathbf{B} \in \mathbb{R}^{p \times k}$ is the matrix for the regression coefficients (e.g. treatment effects) and $\mathbf{\Omega} \in \mathbb{R}^{k \times k}$ is the precision matrix among responses (e.g. network structure). \revision{Note that the treatment effects could also be used to account for subject heterogeneity when samples are not from independent identically distributed from a Gaussian graphical model.}
In the microbial network example described above, the responses correspond to the abundances of microbes in the samples and the predictors correspond to the experimental treatments. Our parameter of interest is $\mathbf \Omega$ which represents the network among responses (the microbial network) and $\mathbf B$ represents the \textit{direct} effects of treatments on the responses (e.g. the effect of antibiotic on different microbes). As mentioned, we are not interested in $\mathbf B$ itself. The introduction of $\mathbf B$ to the model is done to facilitate inference of $\mathbf \Omega$, so in this sense, $\mathbf B$ is a nuisance parameter.
Bayesian implementations of Gaussian chain graph models further include prior distributions for $\mathbf B$ and $\mathbf \Omega$ which can  allow the inclusion of prior biological knowledge into the model. \revision{Different priors have been proposed for precision matrix such as conjugate Wishart prior, shrinkage prior like LASSO and adaptive LASSO prior \citep{glasso}, spike-and-slab priors \citep{Gan2019}, priors based on matrix decomposition and transformations like spectrum \citep{daniels1999nonconjugate}, Cholesky \citep{daniels2002bayesian}, and Givens angle \citep{kang2011bayesian} reference prior \citep{yang1994estimation}. 
Unfortunately, the fact that the mean structure involves the precision matrix in chain graph models, a Gibbs sampler is not straight-forward to implement. In multivariate linear regression, one can condition on the regression coefficients, and then sample the full conditional 
of the precision matrix using the same sampling method for Gaussian graphical model. In contrast, in a Gaussian chain graph model, the mean structure involves the precision matrix, and thus, sampling usually requires modifications. For instance,
despite following the same derivation as in the original graphical LASSO \citep{glasso}, the sampler for graphical LASSO prior on Gaussian chain graph models \citep{shen2020bayesian} is significantly different in the full conditionals. 
While a thorough comparison of multiple priors on Gaussian chain graph models is indeed an interesting research direction, it is currently beyond the scope of this paper and will focus on the priors that have already been adopted in chain graphs.}

Under the setup of Bayesian inference of Gaussian chain graph model, a rather unexplored statistical question is which experimental design on the predictors (design matrix $\mathbf X$) would improve the inference of the precision matrix ($\mathbf \Omega$). Traditional experimental design in linear models involves selecting a data matrix $\mathbf X$ such that the variance (covariance) of the estimator of the regression coefficient $\hat{\beta}$ is small which translates into selecting $\mathbf X$ such that $\mathbf \Sigma_{\mathbf X} = \frac{1}{n} \mathbf X^T \mathbf X$ is as large as possible. The design problem is then build around an optimality criterion (or measure of success) $V(\mathbf \Sigma_{\mathbf X})$ or $V(\mathbf X)$. For example, under the D-optimality setting, the design aims to maximize the determinant of the inverse $\mathbf \Sigma_{\mathbf X}$ matrix: $V_D(\mathbf \Sigma_{\mathbf X}) = \det(\mathbf \Sigma^{-1}_X)$. Other optimality criteria involve maximizing the largest eigenvalue of $\mathbf \Sigma_{\mathbf X}$ (E-optimality) or the trace of $\mathbf \Sigma_{\mathbf X}$ (A-optimality). Bayesian alternatives of the designs change slightly by the incorporation of priors. For example, the Bayesian D-optimality setting involves maximizing $V_D(\mathbf \Sigma_{\mathbf X}) = \det \left( \left( \mathbf \Sigma_{\mathbf X} + \frac{1}{n} \mathbf V_0 \right)^{-1} \right)$ where $\mathbf V_0$ is the prior covariance matrix of $\beta$.

The question of optimal experimental design is not unexplored from the biological perspective. Researchers have always been interested in experiments with specificity: the treatment is conditionally independent with all but one of the response nodes. However, it can be difficult or impossible to design an experiment where only one response node (e.g. microbe) is perturbed. Thus, in the absence of specific treatments, we can study what other alternative experimental designs can aid in the estimation of the network structure.

Here, we address the question of whether we can find an optimal Bayesian experimental design to infer $\mathbf \Omega$, our parameter of interest, on a Gaussian chain graph model. We focus on the Laplace approximation of the marginal posterior precision matrix of $\mathbf \Omega$ as our optimality criterion. We choose to focus on the marginal precision matrix of $\mathbf \Omega$ (instead of the joint precision matrix of $\mathbf B$ and $\mathbf \Omega$) given that we want to account for the nuisance parameter $\mathbf B$ without focusing on it concretely. In addition, we use the Laplace approximation of the precision matrix (instead of the actual precision matrix) given that the precision matrix tends to be intractable for most posteriors (but see Section \ref{toy}). 

We study the case of four different prior distributions: flat, conjugate Normal-Wishart, the novel Normal-Matrix Generalized Inverse Gaussian (MGIG) and a general independent prior. 
For each prior, we obtain the Laplace approximation of the marginal posterior precision matrix of $\mathbf \Omega$ to use as our optimality criterion to find the optimal experimental design $\mathbf X$.
We find, however, that the Laplace approximation of the marginal posterior precision matrix of $\mathbf \Omega$ is not a function of $\mathbf X$ for the flat nor the conjugate priors. This implies that it is difficult if not impossible to find an optimal Bayesian experimental design to aid in the estimation of $\mathbf \Omega$ for these two priors. In contrast, the  Laplace approximation of the marginal posterior precision matrix of $\mathbf \Omega$ is a function of $\mathbf X$ for the novel Normal-MGIG and for a general independent prior which allows the search for an optimal experimental design. However, we discover an information bound for both of these priors which implies that there is a theoretical limit to how much can be gained from experiments in the inference of $\mathbf \Omega$.

Our work has important repercussions for domain scientists who use experimental settings to aid in the estimation of the network structure ($\mathbf \Omega$). Under a Bayesian Gaussian chain graph model, the choice of prior is highly impactful, but even when appropriate priors are selected, there will be an information bound to information gained from experiment. Such bound depends on our prior knowledge about the (conditional) effect of the experiment itself. 


\subsection{Motivating Example: Toy data with explicit posterior precision}
\label{toy}
We show here one example that illustrates the interplay between prior knowledge and experiments. 
We simulate $k=3$ responses, $p=1$ predictor and $n=200$ samples under AR(1) model with $\sigma_{ij}=0.7^{|i-j|}$ (denoted Model 1 in the Simulations in Section \ref{sec:simulation}). We simulate two settings: 1) a null experiment (\revision{design matrix of $\mathbf{X}=0$ despite a potential experimental effect of $\mathbf{B}\ne 0$}), and 2) experiment (the predictor has an effect only on the third response), and we consider two priors each with two uncertainty levels: 1) Normal-Wishart prior with $\lambda=8$, 
$\mathbf \Phi=10^{-3} \mathbf I_3$ and two uncertainty levels of $\mathbf{B}$: i) $\mathbf \Lambda=10^{-3}$ (certain case) and ii) $\mathbf \Lambda=10^3$ (uncertain case), and 2) Normal-MGIG prior with $\lambda=4$, $\mathbf \Psi=\mathbf \Phi=10^{-3} \mathbf I_8$ and two uncertainty levels of $\mathbf{B}$: i) $\mathbf \Lambda=10^{-3}\mathbf I_3$ (certain case) and ii) $\mathbf \Lambda=10^3$ (uncertain case). We set $\mathbf B_0= \mathbf B$ (see Section \ref{priors} for more details on the priors). The design matrix is sampled from the standard Normal distribution.

For this toy example, $\mathbf \Omega$ has an explicit marginal posterior distribution, namely a Matrix Generalized Inverse Gaussian (MGIG) distribution (Equation \ref{eqn:MGIG_posterior}) with parameters $\lambda + \frac{n}{2}$, $\mathbf{\Psi}+(\mathbf{XB}_0)^T (\mathbf{XB}_0)$ and $\mathbf{\Phi} + \mathbf{Y}^T\mathbf{Y}$. 

Figure \ref{fig:toy_posterior} shows the posterior distribution on the partial correlation between the first and the second responses ($\rho_{12}$). We can observe that the experiment has an effect on the inference of the $\rho_{12}$ and that this effect differs by prior with Normal-MGIG prior displaying less variability compared to the conjugate Normal-Wishart. Empirical knowledge similar to this toy example motivated our pursuit of the theoretical interplay of prior distributions and experimental design in the inference of a precision matrix in a chain graph model.

\begin{figure}
    \centering
    \includegraphics[width = \linewidth]{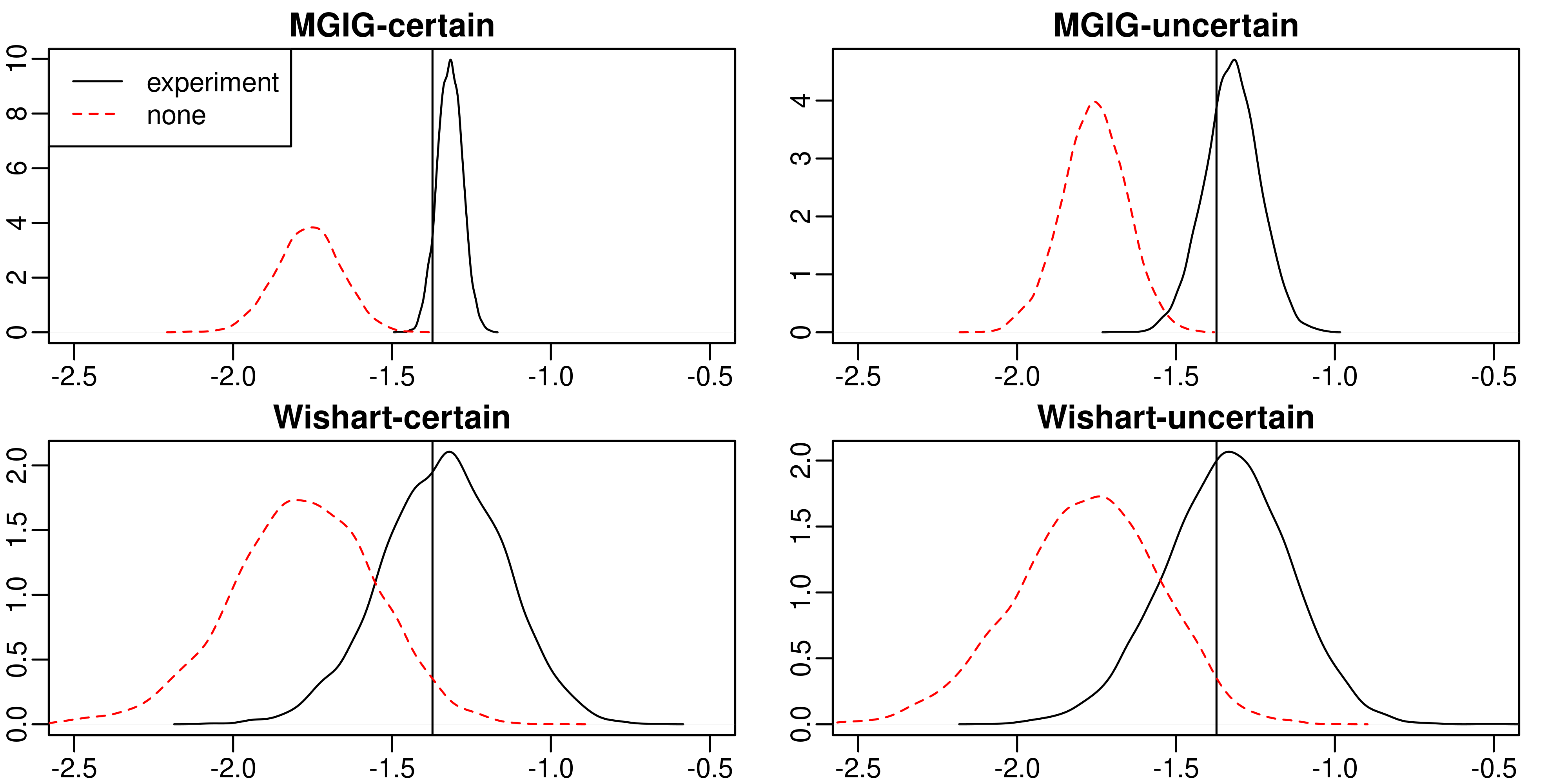}
    \caption{Posterior distribution for the partial correlation between the first and second responses in the toy simulated data with $k=3$ responses, $p=1$ predictor and $n=200$ samples under AR(1) model with and without experiment affecting only the third node, and under Normal-MGIG and Normal-Wishart prior with different uncertainty levels. \revision{Note that we are showing an entry in the precision matrix, not the covariance matrix.} }
    \label{fig:toy_posterior}
\end{figure}

\subsection{Structure of the paper}

The structure of the paper is as follows. In Section \ref{priors} we aimed to understand how prior choice could influence effectiveness of an experimental design. We began with a recap of the Laplace approximation of the posterior precision matrix, and then, we illustrate the Laplace approximation for the marginal posterior precision matrix of the parameter of interest $\mathbf \Omega$ under four priors: flat prior, Normal-Wishart conjugate prior, Normal-MGIG prior and a general  independent prior. We show that optimal experimental design is only possible under the two latter priors, but even in these cases, there is an information limit. In Section \ref{sec:simulation}, we numerically simulate posterior under several priors and different experimental design. We evaluated the result using the Kullback-Leibler divergence between prior and posterior which evaluates the information gain by conducting experiments. In addition, we use the Stein's loss of the maximum a posteriori (MAP) estimation of precision matrix which evaluates the performance of the point estimates from the experiment.

In Section \ref{microbiomedata}, we revisit a similar discussion to the one in the motivating toy dataset where we compare the posterior of partial correlations among responses under different experiments and priors for a real dataset of human gut microbiome.
Finally, in Section \ref{sec:discussion}, we conclude with some practical advice for domain scientists as well as future directions.

\section{Experimental Design under Different Priors in a Gaussian Chain Graph}
\label{priors}

Ideally, we want to use the marginal posterior precision matrix of our parameter of interest $\mathbf \Omega$ as the optimality criterion in an experimental design setting. That is, we want to find the optimal design matrix $\mathbf X$ that maximizes the posterior precision of $\mathbf \Omega$. However, this posterior precision matrix can be intractable in many cases, and thus, we will use its Laplace approximation instead.
We begin this section with a summary of the Laplace approximation of a posterior distribution. Then, we present the Laplace approximation of the marginal posterior precision matrix of $\mathbf \Omega$ under the four priors under study.

\subsection{Laplace Approximation of the Posterior Precision Matrix}

For a log concave posterior distribution $p(\theta | Y)$ for a random variable $Y$ and parameter of interest $\theta$, 
%
the Laplace approximation of the posterior precision is the negative of the Hessian of the log posterior $-\nabla^2_{\theta}\log p(\hat{\theta}|Y)$ which can be partitioned as the sum of the Hessian of the log prior and the Hessian of the log likelihood 
$\nabla^2_{\theta}\log p(\hat{\theta}|Y) =\nabla^2_{\theta}\log p(\hat{\theta})+\nabla^2_{\theta}\log p(Y|\theta)$ \revision{near the maximum \textit{a posteriori} (MAP) estimator. However, we often do not have a closed form expression for the MAP and we simply know than they are close to the true parameter. Thus, we make the additional approximation of using the true parameter to approximate the MAP. We note that while the MAP can depend on the experimental design $\mathbf{X}$, this approximation should not have a considerable impact on the results given that the MAP should be close to the true parameter when the prior has enough support around this true value.}

For the Gaussian chain graph model  \eqref{eqn:general_model}, the Hessian of the log likelihood has the following form (derivation can be found in Appendix \ref{appendix-fi}):
\begin{equation*}
    \begin{aligned}
    \frac{\partial^2 \ell}{\partial \vect(\mathbf{B})\partial \vect(\mathbf{B})^T}&=-\mathbf{\Omega}^{-1}\otimes \mathbf{X}^T\mathbf{X}\\
    \frac{\partial^2 \ell}{\partial \vech(\mathbf{\Omega})\partial \vech(\mathbf{\Omega})^T}&=-\mathbf{D}_k^T\left(\frac{n}{2} \mathbf{\Omega}^{-1}\otimes \mathbf{\Omega}^{-1}+\mathbf{\Omega}^{-1}\otimes \mathbf{\Omega}^{-1}\mathbf{B}^T\mathbf{X}^T\mathbf{X}\mathbf{B}\mathbf{\Omega}^{-1}\right)\mathbf{D}_k\\
    \frac{\partial^2 \ell}{\partial \vect(\mathbf{B})\partial \vech(\mathbf{\Omega})^T}&=\mathbf{D}_k^T\left(\mathbf{\Omega}^{-1}\otimes (\mathbf{\Omega}^{-1}\mathbf{B}^T\mathbf{X}^T\mathbf{X}) \right).
    \end{aligned}
\end{equation*}
where $\mathbf D_k$ is the duplication matrix \citep{minka2000old,magnus2019matrix}, a permutation matrix such that $\mathbf D_k\vect{\mathbf\Omega}= \vech(\mathbf\Omega)$ where $\vech(\mathbf\Omega)$ denotes the vectorization of the unique parameters in $\mathbf\Omega$ (upper triangular part in our case) given that $\mathbf \Omega$ is symmetric, so there are fewer free parameters.

Because the Gaussian chain graph model is an exponential family, the Hessian of the log likelihood is also the negative of the Fisher information matrix:
\begin{equation}
\label{eqn:FI}
    I(\mathbf \Omega, \mathbf B)=\left[
    \begin{matrix}
    \mathbf{D}_k^T\left(\frac{n}{2} \mathbf{\Omega}^{-1}\otimes \mathbf{\Omega}^{-1}+\mathbf{\Omega}^{-1}\otimes \mathbf{\Omega}^{-1}\mathbf{B}^T\mathbf{X}^T\mathbf{X}\mathbf{B}\mathbf{\Omega}^{-1}\right)\mathbf{D}_k & 
    -\mathbf{D}_k^T\left(\mathbf{\Omega}^{-1}\otimes (\mathbf{\Omega}^{-1}\mathbf{B}^T\mathbf{X}^T\mathbf{X}) \right)\\
    -(\mathbf{D}_k^T\left(\mathbf{\Omega}^{-1}\otimes (\mathbf{\Omega}^{-1}\mathbf{B}^T\mathbf{X}^T\mathbf{X}) \right))^T & \mathbf\Omega^{-1} \otimes  \mathbf{X}^T \mathbf{X}
    \end{matrix}
    \right].
\end{equation}

Having the Hessian of the log likelihood (the negative Fisher information in \eqref{eqn:FI}), we only need the Hessian of the log prior for each of the priors under study to obtain the negative Hessian of the log posterior, and thus, the Laplace approximation of the posterior precision matrix of $\mathbf B$ and $\mathbf \Omega$ which can be written in matrix form:
\begin{equation}
\left[
\begin{matrix}
\mathbf A & \mathbf G\\
\mathbf C & \mathbf D\\
\end{matrix}
\right].
\label{matrixform}
\end{equation}
Given that this matrix corresponds to the Laplace-approximated \textit{posterior precision}, its inverse corresponds to the Laplace-approximated \textit{posterior covariance} matrix of $\mathbf B$ and $\mathbf \Omega$:
\begin{equation*}
\left[
\begin{matrix}
\mathbf A & \mathbf G\\
\mathbf C & \mathbf D\\
\end{matrix}
\right]^{-1}
=
\left[
\begin{matrix}
(\mathbf A-\mathbf G\mathbf D^{-1}\mathbf C)^{-1} & -\mathbf A^{-1}\mathbf G(\mathbf D-\mathbf C\mathbf A^{-1}\mathbf G)^{-1}\\
-\mathbf D^{-1}\mathbf C(\mathbf A-\mathbf G\mathbf D^{-1}\mathbf C)^{-1} & (\mathbf D-\mathbf C\mathbf A^{-1}\mathbf G)^{-1}\\
\end{matrix}
\right].
\end{equation*}
The block matrix $(\mathbf A-\mathbf G\mathbf D^{-1}\mathbf C)^{-1}$ is denoted the Schur complement of $\mathbf A$ \citep{prasolov1994problems} and it corresponds to the Laplace approximation of the \textit{marginal posterior covariance} matrix of $\mathbf \Omega$. Its inverse $(\mathbf A-\mathbf G\mathbf D^{-1}\mathbf C)$ is then the Laplace approximation of the \textit{marginal posterior precision} matrix of $\mathbf \Omega$, our optimality criterion to address the question of optimal Bayesian experimental design for each of the priors. As mentioned before, we focus on the marginal precision matrix of $\mathbf \Omega$ because we want to account for $\mathbf B$, but only as a nuisance parameter.

Our goal for the remaining of this section is to obtain the matrix $(\mathbf A-\mathbf G\mathbf D^{-1}\mathbf C)$ for each of the four priors under study with $\mathbf A, \mathbf G, \mathbf C, \mathbf D$ coming from the negative of the Hessian of the log posterior.
This matrix $(\mathbf A-\mathbf G\mathbf D^{-1}\mathbf C)$ is our optimality criterion in the Bayesian experimental design and will allow us to find the optimal design matrix $\mathbf X$ for each prior.

Finally, we note that the Laplace approximation is only valid for posterior densities that are log concave. This is trivially true for the case of the flat prior because the likelihood model is log concave. We prove log concavity of the other three priors under study in the Appendix.

\subsection{Flat Prior}

As mentioned, the flat prior is trivially log concave because the likelihood model (Gaussian chain graph model) is log concave. Thus, we can obtain the Laplace approximation of the posterior precision matrix of $\mathbf B$ and $\mathbf \Omega$ as the negative Hessian of the log likelihood (the Fisher information in  \eqref{eqn:FI}) which can be written as the matrix in  \eqref{matrixform} with
\begin{align*}
 \mathbf A &= \mathbf{D}_k^T\left(\frac{n}{2} \mathbf{\Omega}^{-1}\otimes \mathbf{\Omega}^{-1}+\mathbf{\Omega}^{-1}\otimes \mathbf{\Omega}^{-1}\mathbf{B}^T\mathbf{X}^T\mathbf{X}\mathbf{B}\mathbf{\Omega}^{-1}\right)\mathbf{D}_k, \\
 \mathbf G &=\mathbf C^T=-\mathbf{D}_k^T\left(\mathbf{\Omega}^{-1}\otimes (\mathbf{\Omega}^{-1}\mathbf{B}^T\mathbf{X}^T\mathbf{X}) \right),\\
 \mathbf D &=\mathbf\Omega^{-1} \otimes  \mathbf{X}^T \mathbf{X}.
\end{align*}

Then, the Laplace approximation of the marginal posterior precision matrix of $\mathbf \Omega$ is given by the inverse of the Schur complement of $\mathbf A$: $(\mathbf A-\mathbf G\mathbf D^{-1}\mathbf C)$, typically denoted as $\mathbf \Omega | I(\mathbf \Omega, \mathbf B)$:

\begin{align*}
\mathbf \Omega | I(\mathbf \Omega, \mathbf B)=&\mathbf{D}_k^T\left(\frac{n}{2} \mathbf{\Omega}^{-1}\otimes \mathbf{\Omega}^{-1}+\mathbf{\Omega}^{-1}\otimes \mathbf{\Omega}^{-1}\mathbf{B}^T\mathbf{X}^T\mathbf{X}\mathbf{B}\mathbf{\Omega}^{-1}\right)\mathbf{D}_k\\
&-\mathbf{D}_k^T\left(\mathbf{\Omega}^{-1}\otimes (\mathbf{\Omega}^{-1}\mathbf{B}^T\mathbf{X}^T\mathbf{X}) \right)\left[\mathbf{\Omega}\otimes(\mathbf{X}^T \mathbf{X})^{-1}\right]\left[(\mathbf{D}_k^T\left(\mathbf{\Omega}^{-1}\otimes \mathbf{\Omega}^{-1}\mathbf{B}^T\mathbf{X}^T\mathbf{X})\right) \right]^T\\
=&\frac{n}{2}\mathbf{D}_k^T\left( \mathbf{\Omega}^{-1}\otimes \mathbf{\Omega}^{-1}\right)\mathbf{D}_k.
\end{align*}

The next step in Bayesian experimental design would be to find a design matrix $\mathbf X$ that maximizes our optimality criterion $\mathbf \Omega | I(\mathbf \Omega, \mathbf B)$. We note, however, that $\mathbf \Omega | I(\mathbf \Omega, \mathbf B)$ is not a function of $\mathbf X$, and thus, optimal experimental design cannot be performed with this optimality criterion.
We also note that under the setting of $\mathbf B$ completely known, the Fisher information matrix becomes the upper left block in  \eqref{eqn:FI} (matrix $\mathbf A$) which indeed contains $\mathbf X$, and thus, experimental design is possible for the known $\mathbf B$ case. This remark hints at the possibility that prior knowledge from $\mathbf B$ could help infer $\mathbf \Omega$ as will be confirmed in the next sections.

\subsection{Standard Conjugate Prior: Normal-Wishart}

The standard conjugate prior for a Gaussian chain graph model is the Normal-Wishart family: 
\begin{equation}
    \begin{aligned}
    \mathbf \Omega|\lambda, \mathbf \Phi &\sim W_k(\lambda, \mathbf \Phi^{-1})\\
    \vect (\mathbf B)|\mathbf \Omega,\mathbf B_0,\mathbf \Lambda &\sim N(\vect(\mathbf B_0\mathbf \Omega),\mathbf{\Omega}\otimes \mathbf \Lambda)
    \end{aligned}
\end{equation}
where $\mathbf \Phi\in \mathbb{R}^{k\times k}$ is positive definite, $\lambda$ is a scalar, $\mathbf{B}_0 \in \mathbb{R}^{p \times k}$, and $\mathbf \Lambda\in \mathbb{R}^{p\times p}$ represents the uncertainty on $\mathbf{B}$.
Then, the posterior distribution is given by 
\begin{equation*}
    \begin{aligned}
    \mathbf \Omega| \mathbf Y,\mathbf X, \lambda, \mathbf \Phi &\sim W_k(\lambda+n,\hat{\mathbf \Phi}^{-1})\\
    \vect (\mathbf B)|\mathbf \Omega,\mathbf Y,\mathbf X, \mathbf B_0, \mathbf \Lambda &\sim N(\vect((\mathbf \Lambda^{-1}+\mathbf X^T\mathbf X)^{-1}(\mathbf{B}_0^T \mathbf \Lambda^{-1}+\mathbf Y^T \mathbf X) \mathbf \Omega),\mathbf{\Omega}\otimes(\mathbf \Lambda^{-1}+\mathbf X^T\mathbf X)^{-1})
    \end{aligned}
\end{equation*}
where 
$\hat{\mathbf \Phi}=\mathbf \Phi+\mathbf{Y}^T\mathbf{Y}+\mathbf{B}_0^T\mathbf{\Lambda}^{-1}\mathbf{B}_0-(\mathbf{B}_0^T\mathbf{\Lambda}^{-1}+\mathbf{Y}^T\mathbf{X})(\mathbf{X}^T\mathbf{X}+\mathbf{\Lambda}^{-1})^{-1}(\mathbf{B}_0^T\mathbf{\Lambda}^{-1}+\mathbf{Y}^T\mathbf{X})^T$. We show that this matrix is positive definite in Appendix \ref{sec:posdef}.

To use the Laplace approximation for the posterior precision, the posterior distribution needs to be log concave. The Normal-Wishart posterior is log concave under $\frac{1}{2}(\lambda-k-p-1)\ge k/2$ (Appendix \ref{logNW}), so we can use the Laplace approximation to approximate the posterior precision as the negative Hessian of the log posterior which, in turn, is the sum of the Hessian of the log likelihood  \eqref{eqn:FI} and the Hessian of the log prior which is given by:
\begin{equation}
\begin{aligned}
\frac{\partial^2}{\partial\vect{(\mathbf B)}\partial\vect{(\mathbf B)}^T} \log p(\mathbf{\Omega,B}) &=-\mathbf \Omega^{-1}\otimes \mathbf \Lambda^{-1}\\
\frac{\partial^2}{\partial\vech{(\mathbf \Omega)} \partial\vech{(\mathbf \Omega)}^T} \log p(\mathbf{\Omega,B})&=-\mathbf D_k^T\left(\mathbf \Omega^{-1}\otimes \left(\frac{1}{2}(\lambda-k-p-1) \mathbf \Omega^{-1} + \mathbf \Omega^{-1}(\mathbf B^T \mathbf \Lambda^{-1}\mathbf B)\mathbf \Omega^{-1} \right) \right)\mathbf D_k\\
\frac{\partial^2}{\partial\vech{(\mathbf \Omega)}\partial\vect{(\mathbf B)}^T} \log p(\mathbf{\Omega,B})&=\mathbf D_k^T\left(\mathbf \Omega^{-1}\otimes \left( \mathbf \Omega^{-1}\mathbf B^T \mathbf \Lambda^{-1} \right)\right).\\
\end{aligned}
\label{eqn:hess_wish}
\end{equation}

The Hessian of the log posterior can then be written as the matrix in  \eqref{matrixform} with
\begin{align*}
    \mathbf A
    =& \mathbf D_k^T\left(\mathbf \Omega^{-1}\otimes \left(\alpha \mathbf \Omega^{-1} + \mathbf \Omega^{-1}(\mathbf B^T \mathbf \Lambda^{-1}\mathbf B)\mathbf \Omega^{-1} \right) \right)\mathbf D_k \\
    &+ \mathbf{D}_k^T\left(\frac{n}{2} \mathbf{\Omega}^{-1}\otimes \mathbf{\Omega}^{-1}+\mathbf{\Omega}^{-1}\otimes \mathbf{\Omega}^{-1}\mathbf{B}^T\mathbf{X}^T\mathbf{X}\mathbf{B}\mathbf{\Omega}^{-1}\right)\mathbf{D}_k \\
    =& \mathbf D_k^T\left[\mathbf \Omega\otimes \left( (\frac{n}{2}+\alpha )\mathbf \Omega^{-1} + \mathbf \Omega^{-1}(\mathbf B^T\mathbf X^T\mathbf{XB}+\mathbf B^T\Lambda^{-1}B)\mathbf \Omega^{-1}  \right)\right]\mathbf D_k,\\
    \mathbf G
    =& \mathbf C^T
    = -\mathbf D_k^T\left(\mathbf \Omega^{-1}\otimes \left( \mathbf \Omega^{-1}\mathbf B^T \mathbf \Lambda^{-1} \right)\right)-\mathbf{D}_k^T\left(\mathbf{\Omega}^{-1}\otimes (\mathbf{\Omega}^{-1}\mathbf{B}^T\mathbf{X}^T\mathbf{X}) \right)\\
    =& -\mathbf D_k^T\left(\mathbf \Omega^{-1}\otimes \left( \mathbf \Omega^{-1}\mathbf B^T (\mathbf{X}^T \mathbf{X}+\mathbf \Lambda^{-1}) \right)\right),\\
    \mathbf D
    =& \mathbf \Omega^{-1}\otimes {(\mathbf{X}^T \mathbf{X} + \mathbf \Lambda^{-1})}.
\end{align*}

Then, the Laplace approximation of the marginal posterior precision matrix of $\mathbf \Omega$ is given by the inverse of the Schur complement of $\mathbf A$: $(\mathbf A-\mathbf G\mathbf D^{-1}\mathbf C)$, typically denoted as $\mathbf \Omega | I(\mathbf \Omega, \mathbf B)$:
\begin{equation}
\label{eqn:Wishart_postprecision}
\begin{aligned}
\mathbf{\Omega}|I(\mathbf \Omega,\mathbf{B})=&\mathbf D_k^T\left[\mathbf \Omega^{-1}\otimes \left( \left(\frac{n}{2}+\alpha \right)\mathbf \Omega^{-1} + \mathbf \Omega^{-1}(\mathbf B^T\mathbf X^T\mathbf{XB}+\mathbf B^T \mathbf \Lambda^{-1} \mathbf B)\mathbf \Omega^{-1}  \right)\right]\mathbf D_k\\
&-\mathbf{D}_k^T\left(\mathbf{\Omega}^{-1}\otimes (\mathbf{\Omega}^{-1}\mathbf{B}^T(\mathbf{X}^T\mathbf{X}+\mathbf \Lambda^{-1})) \right) \\
& \left[\mathbf{\Omega}\otimes(\mathbf{X}^T \mathbf{X}+\mathbf \Lambda^{-1})^{-1}\right]\left[(\mathbf{D}_k^T\left(\mathbf{\Omega}^{-1}\otimes \mathbf{\Omega}^{-1}\mathbf{B}^T(\mathbf{X}^T\mathbf{X}+\mathbf \Lambda^{-1}))\right) \right]^T\\
=&\left(\frac{n}{2}+\alpha \right)\mathbf D_k^T\left[\mathbf \Omega^{-1}\otimes  \mathbf \Omega^{-1}\right]\mathbf D_k
\end{aligned}
\end{equation}
which is again not a function of design $\mathbf X$ (see Appendix \ref{sec:approx-conj} for details on the algebraic simplification). This implies again that we cannot find an optimal experimental design for the conjugate prior with the optimality criterion of the Laplace approximation of the marginal posterior precision of $\mathbf \Omega$.

In addition, if we use the Normal-Wishart conjugate prior, our prior knowledge is actually on the marginal regression coefficient $\tilde{\mathbf{B}} = \mathbf B \mathbf\Omega^{-1}$, not on $\mathbf B$.
That is, the Normal-Wishart prior does not identify parameters $\mathbf\Omega$ and $\mathbf B$ separately, rather it only has information about linear combinations of them. This is also evident by observing that if we take $\mathbf{\Lambda}\to 0$ (the uncertainty on $\mathbf{B}$), $\mathbf B$ is still not fully known. Instead, we would only know $\mathbf B_0\mathbf \Omega^{-1}=\tilde{\mathbf{B}}_0$ where $\mathbf{\Omega}$ is still random. Thus, when the uncertainty of $\mathbf B$ goes to zero, the conjugate Normal-Wishart prior does not reduce to the known $\mathbf B$ case which, in addition to not allowing optimal experimental design to infer $\mathbf \Omega$, make the conjugate prior a suboptimal prior for the Gaussian chain graph model when our focus is the estimation of $\mathbf \Omega$.

\subsection{Normal-Matrix Generalized Inverse Gaussian Prior}

The drawbacks on the Normal-Wishart conjugate prior for the estimation of $\mathbf \Omega$ motivate us to search for other prior alternatives.
We consider the
Matrix Generalized Inverse Gaussian distribution (MGIG) \citep{barndorff1982exponential,fazayeli2016matrix} for $\mathbf \Omega$ to define the Normal-MGIG prior
which is not conjugate for $\mathbf B$, but yields a MGIG posterior for $\mathbf\Omega$ that can be sampled via importance sampling \citep{fazayeli2016matrix}. 

The Normal-MGIG prior is given by
\begin{equation}
\begin{aligned}
\mathbf{\Omega}|\lambda, \mathbf \Psi, \mathbf \Phi & \sim MGIG(\lambda,\mathbf \Psi,\mathbf \Phi)\\
\vect(\mathbf{B})|\mathbf \Omega, \mathbf B_0, \mathbf \Lambda  &\sim N(\vect(\mathbf{B}_0),\mathbf{\Omega}\otimes \mathbf \Lambda)
\end{aligned}
\end{equation}
where $\mathbf \Psi$, $\mathbf \Phi\in \mathbb R^{k\times k}$ are positive definite while $\lambda$ is a scalar. $\mathbf{B}_0\in \mathbb R^{p\times k}$ is the mean of $\mathbf B$ and $\mathbf \Lambda\in \mathbb R^{p\times p}$ is the uncertainty on $\mathbf{B}$.
Then, the posterior distribution is proportional to 
\begin{equation}
\label{eqn:prior_good}
\begin{aligned}
p(\mathbf{\Omega},\mathbf{B}|\mathbf{Y},\mathbf{X},\theta)\propto&|\mathbf{\Omega}|^{\frac{n}{2}}\exp\left( \tr(\mathbf{Y}^T\mathbf{X}\mathbf{B})-\frac{1}{2} \tr(\mathbf{Y}^T\mathbf{Y}\mathbf{\Omega})- \frac{1}{2} \tr(\mathbf{B}^T\mathbf{X}^T\mathbf{X}\mathbf{B}\mathbf{\Omega}^{-1})\right)\\
&\times |\mathbf{\Omega}|^{-\frac{p}{2}}\exp\left( -\frac{1}{2}\tr([\mathbf{B}-\mathbf{B}_0]^T\mathbf{\Lambda}^{-1}[\mathbf{B}-\mathbf{B}_0]\mathbf{\Omega}^{-1}) ) \right)\\
&\times |\mathbf{\Omega}|^{\lambda-\frac{k+1}{2}}\exp(-\frac{1}{2} \tr[\mathbf \Psi\mathbf{\Omega}^{-1}]-\frac{1}{2}\tr[\mathbf \Phi\mathbf{\Omega}])\\
\propto& |\mathbf{\Omega}|^{-\frac{p}{2}}\exp\left( -\frac{1}{2} \tr(-2(\mathbf{\Omega}^{-1}\mathbf{B}_0^T\mathbf{\Lambda}^{-1}+\mathbf{Y}^T\mathbf{X})\mathbf{B}+\mathbf{B}^T(\mathbf{\Lambda}^{-1}+\mathbf{X}^T\mathbf{X})\mathbf{B}\mathbf{\Omega}^{-1}) ) \right)\\
&\times \exp\left(-\frac{1}{2} \tr\left( (\mathbf{\Omega}^{-1}\mathbf{B}_0^T\mathbf{\Lambda}^{-1}+\mathbf{Y}^T\mathbf{X})(\mathbf{X}^T\mathbf{X}+\mathbf{\Lambda}^{-1})^{-1}(\mathbf{\Omega}^{-1}\mathbf{B}_0^T\mathbf{\Lambda}^{-1}+\mathbf{Y}^T\mathbf{X})^T\mathbf{\Omega}\right)\right)\\
&\times |\mathbf{\Omega}|^{\lambda+\frac{n}{2}-\frac{k+1}{2}}\exp(-\frac{1}{2} \tr((\mathbf \Phi+\mathbf{Y}^T\mathbf{Y}-\mathbf{Y}^T\mathbf{X}(\mathbf{X}^T\mathbf{X}+\mathbf{\Lambda}^{-1})^{-1}\mathbf{X}^T\mathbf{Y})\mathbf{\Omega}))\\
&\times\exp(-\frac{1}{2} \tr((\mathbf \Psi+\mathbf{B}_0^T\mathbf{\Lambda}^{-1}\mathbf{B}_0-\mathbf{B}_0^T\mathbf{\Lambda}^{-1}(\mathbf{X}^T\mathbf{X}+\mathbf{\Lambda}^{-1})^{-1}\mathbf{\Lambda}^{-1}\mathbf{B}_0)\mathbf{\Omega}^{-1}))
\end{aligned}
\end{equation}
where we use $\theta$ to denote all hyper-parameters in the prior and thus, we get
\begin{equation*}
    \begin{aligned}
    \mathbf{\Omega}|\mathbf{Y},\mathbf{X},\theta&\sim MGIG(\lambda+\frac{n}{2},\hat{\mathbf \Psi},\hat{\mathbf \Phi}) \\
    \vect(\mathbf{B})|\mathbf{\Omega},\mathbf{Y},\mathbf{X},\theta&\sim N\left( \mathbf{\Omega}\otimes (\mathbf{X}^T\mathbf{X}+\mathbf{\Lambda}^{-1})^{-1}\vect(\mathbf{X}^T\mathbf{Y}+\mathbf{\Lambda}^{-1}\mathbf{B}_0\mathbf{\Omega}^{-1}), \mathbf{\Omega}\otimes(\mathbf{X}^T\mathbf{X}+\mathbf{\Lambda}^{-1})^{-1} \right)
    \end{aligned}
\end{equation*}
where $\hat{\mathbf \Psi}=\mathbf \Psi+\mathbf{B}_0^T\mathbf{\Lambda}^{-1}\mathbf{B}_0-\mathbf{B}_0^T\mathbf{\Lambda}^{-1}(\mathbf{X}^T\mathbf{X}+\mathbf{\Lambda}^{-1})^{-1}\mathbf{\Lambda}^{-1}\mathbf{B}_0$ and $\hat{\mathbf \Phi}=\mathbf \Phi+\mathbf{Y}^T\mathbf{Y}-\mathbf{Y}^T\mathbf{X}(\mathbf{X}^T\mathbf{X}+\mathbf{\Lambda}^{-1})^{-1}\mathbf{X}^T\mathbf{Y}$. 
We show that these matrices are positive definite in Appendix \ref{sec:posdef}.

To the best of our knowledge, the Normal-MGIG prior has not been used for the Gaussian chain graph model, and thus, we prove here some of its properties in the Appendix \ref{propNMGIG}: we show that the MGIG prior is conjugate for the case of known $\mathbf B$ (Proposition \ref{conjB}), that it is log concave under certain conditions (Proposition \ref{logconcave}), that it is unimodal for the case of unknown $\mathbf B$ (Proposition \ref{unimod}), and that its limiting case is indeed the case of known $\mathbf B$ (Remark \ref{NMGIG-limit}).

The Normal-MGIG posterior is log concave under $\lambda-\frac{k+p+1}{2}\ge \frac{p}{2}$ (Proposition \ref{logconcave} in the Appendix), so we can use the Laplace approximation to approximate the posterior precision as the negative Hessian of the log posterior which, in turn, is the sum of the Hessian of the log likelihood \eqref{eqn:FI} and the Hessian of the log prior which has the following form:
\begin{equation}
    \begin{aligned}
    \frac{\partial^2}{\partial\vect(\mathbf B)\partial\vect(\mathbf B)^T}\log p(\mathbf \Omega, \mathbf B)=&-\mathbf \Omega^{-1}\otimes \mathbf \Lambda^{-1}\\
    \frac{\partial^2}{\partial\vech(\mathbf \Omega)\partial\vech(\mathbf \Omega)^T}\log p(\mathbf \Omega, \mathbf B) =& \mathbf D_k^T\left(\left(\lambda-\frac{k+p+1}{2}\right) \mathbf \Omega^{-1}\otimes \mathbf \Omega^{-1}\right)\mathbf D_k\\
     &+\mathbf D_k^T\left(\mathbf \Omega^{-1}\otimes \left(\mathbf \Omega^{-1}[(\mathbf B-\mathbf B_0)^T\mathbf \Lambda^{-1}(\mathbf B-\mathbf B_0)+\mathbf \Psi]\mathbf \Omega^{-1}\right)\right)\mathbf D_k\\
    \frac{\partial^2}{\partial\vect(\mathbf B)\partial\vech(\mathbf \Omega)^T}\log p(\mathbf \Omega, \mathbf B)=&-\mathbf D_k^T\left(\mathbf \Omega^{-1}\otimes (\mathbf \Omega^{-1}(\mathbf B-\mathbf B_0)^T\mathbf \Lambda^{-1}) \right).
    \end{aligned}
    \label{eqn:prior_good_Hessian}
\end{equation}
The Hessian of the log posterior can then be written as the matrix in  \eqref{matrixform} with
\begin{align*}
    \mathbf A=& \mathbf D_k^T\left(\alpha \mathbf \Omega^{-1}\otimes \mathbf \Omega^{-1}\right)\mathbf D_k
     +\mathbf D_k^T\left(\mathbf \Omega^{-1}\otimes \left(\mathbf \Omega^{-1}[(\mathbf B-\mathbf B_0)^T\mathbf \Lambda^{-1}(\mathbf B-\mathbf B_0)+\mathbf \Psi]\mathbf \Omega^{-1}\right)\right)\mathbf D_k \\
    &+ \mathbf{D}_k^T\left(\frac{n}{2} \mathbf{\Omega}^{-1}\otimes \mathbf{\Omega}^{-1}+\mathbf{\Omega}^{-1}\otimes \mathbf{\Omega}^{-1}\mathbf{B}^T\mathbf{X}^T\mathbf{X}\mathbf{B}\mathbf{\Omega}^{-1}\right)\mathbf{D}_k \\
    =&\mathbf D_k^T\left[\mathbf \Omega^{-1}\otimes \left( (\frac{n}{2}+\alpha )\mathbf \Omega^{-1} + \mathbf \Omega^{-1}(\mathbf B^T\mathbf X^T\mathbf{XB}+(\mathbf B-\mathbf B_0)^T \mathbf \Lambda^{-1}(\mathbf B-\mathbf B_0)+\mathbf \Psi)\mathbf \Omega^{-1}  \right)\right]\mathbf{D}_k, \\
    \mathbf G=&\mathbf C^T
    = -\mathbf D_k^T\left(\mathbf \Omega^{-1}\otimes (\mathbf \Omega^{-1}(\mathbf B-\mathbf B_0)^T\mathbf \Lambda^{-1}) \right) -\mathbf{D}_k^T\left(\mathbf{\Omega}^{-1}\otimes (\mathbf{\Omega}^{-1}\mathbf{B}^T\mathbf{X}^T\mathbf{X}) \right) \\
    =& -\mathbf{D}_k^T\left(\mathbf{\Omega}^{-1}\otimes (\mathbf{\Omega}^{-1}(\mathbf{B}^T\mathbf{X}^T\mathbf{X}+(\mathbf B-\mathbf B_0)^T\mathbf \Lambda^{-1})) \right),\\
    \mathbf D
    =& \mathbf \Omega^{-1}\otimes {(\mathbf{X}^T \mathbf{X} + \mathbf \Lambda^{-1} )}.
\end{align*}

Then, the Laplace approximation of the marginal posterior precision matrix of $\mathbf \Omega$ is given by the inverse of the Schur complement of $\mathbf A$: $(\mathbf A-\mathbf G\mathbf D^{-1}\mathbf C)$, typically denoted as $\mathbf \Omega | I(\mathbf \Omega, \mathbf B)$:
\begin{equation}
\begin{aligned}
\mathbf{\Omega}|I(\mathbf \Omega,\mathbf{B})=&\mathbf D_k^T\left[\mathbf \Omega^{-1}\otimes \left( (\frac{n}{2}+\alpha )\mathbf \Omega^{-1} + \mathbf \Omega^{-1}(\mathbf B^T\mathbf X^T\mathbf{XB}+(\mathbf B-\mathbf B_0)^T \mathbf \Lambda^{-1}(\mathbf B-\mathbf B_0)+\mathbf \Psi)\mathbf \Omega^{-1}  \right)\right]\mathbf{D}_k\\
&-\mathbf{D}_k^T\left(\mathbf{\Omega}^{-1}\otimes (\mathbf{\Omega}^{-1}(\mathbf{B}^T\mathbf{X}^T\mathbf{X}+(\mathbf B-\mathbf B_0)^T\mathbf \Lambda^{-1})) \right)\left[\mathbf{\Omega}\otimes(\mathbf{X}^T \mathbf{X}+\mathbf \Lambda^{-1})^{-1}\right]\\
&\left[(\mathbf{D}_k^T\left(\mathbf{\Omega}^{-1}\otimes \mathbf{\Omega}^{-1}(\mathbf{B}^T\mathbf{X}^T\mathbf{X}+(\mathbf B-\mathbf B_0)^T\mathbf \Lambda^{-1}))\right) \right]^T\\
=&\mathbf D_k^T\left[\mathbf \Omega^{-1}\otimes\left( (\frac{n}{2}+\alpha)\mathbf \Omega^{-1}+\mathbf \Omega^{-1} \mathbf B_0^T(\mathbf \Lambda^{-1} - \mathbf \Lambda^{-1}(\mathbf X^T\mathbf X+\mathbf \Lambda^{-1})^{-1}\mathbf \Lambda^{-1})\mathbf B_0 \mathbf \Omega^{-1} \right) \right]\mathbf D_k\\
&+\mathbf D_k^T\left[\mathbf{\Omega}^{-1}\otimes(\mathbf \Omega^{-1}\mathbf \Psi\mathbf \Omega^{-1})\right]\mathbf D_k
\end{aligned}
\label{eqn:MGIG_postprecision}
\end{equation}
which is a function of design $\mathbf X$ (see Appendix \ref{appNMGIG} for details on the algebraic simplification). 
Note that the only term that involves the design matrix $\mathbf{X}$ is $\mathbf \Omega^{-1}\otimes\mathbf \Omega^{-1}\mathbf B_0^T(\mathbf \Lambda^{-1} - \mathbf \Lambda^{-1}(\mathbf X^T\mathbf X+\mathbf \Lambda^{-1})^{-1}\mathbf \Lambda^{-1})\mathbf B_0\mathbf \Omega^{-1}$
which is $\mathbf 0$ when $\mathbf X=\mathbf 0$ (i.e. no experiment), so we can understand this term as the information gain due to experiment. We can then consider the optimal design based on $( \mathbf \Lambda^{-1}(\mathbf X^T\mathbf X+\mathbf \Lambda^{-1})^{-1}\mathbf \Lambda^{-1})$ which is the only term that we have control over. For example, for a  D-optimal design, we want to maximize
$|\mathbf \Lambda^{-1}(\mathbf X^T \mathbf X+\mathbf \Lambda^{-1})^{-1}\mathbf \Lambda^{-1}|=\frac{1}{|\mathbf \Lambda|^2|\mathbf X^T\mathbf X+\mathbf \Lambda^{-1}|}$
which can be achieved by maximizing $|\mathbf X^T\mathbf X+\mathbf \Lambda^{-1}|$ which coincides with the usual Bayesian D-optimal for marginal regression coefficient \citep{chaloner1995design}. 

\subsubsection{Information bound under the Normal-MGIG prior}
\label{info-bound-NMGIG}

We now know that the experimental design has an effect on the estimation of $\mathbf \Omega$ under a Normal-MGIG prior by influencing its approximate posterior precision. However, it turns out that there is a bound on the information we can gain from the experiment ($\mathbf X$).

Recall that we denote the term $\mathbf \Omega^{-1}\otimes\mathbf \Omega^{-1}\mathbf B_0^T(\mathbf \Lambda^{-1} - \mathbf \Lambda^{-1}(\mathbf X^T\mathbf X+\mathbf \Lambda^{-1})^{-1}\mathbf \Lambda^{-1})\mathbf B_0\mathbf \Omega^{-1}$ the information gain from experiments.
First, we observe that the inequality $(\mathbf \Lambda^{-1} - \mathbf \Lambda^{-1}(\mathbf X^T\mathbf X+\mathbf \Lambda^{-1})^{-1}\mathbf \Lambda^{-1})\le \mathbf \Lambda^{-1}$
is true because the difference ($\mathbf \Lambda^{-1}(\mathbf X^T\mathbf X+\mathbf \Lambda^{-1})^{-1}\mathbf \Lambda^{-1}$) is a positive semi-definite matrix (as it takes a quadratic form with a positive definite matrix). 

By multiplying the inequality by $\mathbf \Omega^{-1} \mathbf B_0^T$ and $\mathbf{B}_0 \mathbf \Omega^{-1}$ left and right respectively, we get
\begin{align*}
    \label{eqn:information_gain}
    \mathbf \Omega^{-1} \mathbf B_0^T(\mathbf \Lambda^{-1} - \mathbf \Lambda^{-1}(\mathbf X^T\mathbf X+\mathbf \Lambda^{-1})^{-1}\mathbf \Lambda^{-1})\mathbf B_0 \mathbf \Omega^{-1} &\le \mathbf \Omega^{-1} \mathbf{B}_0^T \mathbf \Lambda^{-1}\mathbf{B}_0 \mathbf \Omega^{-1}
\end{align*}
where the term on the left is precisely the information that we can gain from non-zero experiments ($\mathbf X\ne \mathbf 0$) and this term is bounded by $\mathbf \Omega^{-1} \mathbf{B}_0^T \mathbf \Lambda^{-1}\mathbf{B}_0 \mathbf \Omega^{-1}$.
This means that when we try to find the optimal experimental design $\mathbf X$ that maximizes the Laplace approximation of the marginal posterior precision matrix of $\mathbf \Omega$, the only term that depends on $\mathbf X$ is bounded, and thus, there is a limit to how much can be gained by an optimal experimental design.

Next, we observe that this bound is sharp as the equality is achieved when  $\mathbf X^T \mathbf X\to\infty$. In addition, this inequality provides the intuition that the information gain due to the experiment is bounded by the product of the marginal effect of the experiment ($\mathbf B_0 \mathbf \Omega^{-1}$) and the prior certainty on the experiment's conditional effect ($\mathbf \Lambda^{-1}$). Thus, if our prior on the effect of the experiment is no effect on any nodes ($\mathbf B_0=\mathbf 0$), we will gain no information from the experiment.
We can re-write the bound in terms of the marginal regression coefficient $\tilde{\mathbf{B}}_0 = \mathbf B_0 \mathbf \Omega^{-1}$: $\mathbf \Omega^{-1} \mathbf B_0^T(\mathbf \Lambda^{-1} - \mathbf \Lambda^{-1}(\mathbf X^T\mathbf X+\mathbf \Lambda^{-1})^{-1}\mathbf \Lambda^{-1})\mathbf B_0 \mathbf \Omega^{-1}\le \tilde{\mathbf{B}}_0^T \mathbf \Lambda^{-1}\tilde{\mathbf{B}}_0$.

In conclusion, given the information bound, the experimental design $\mathbf X$ is not as important as the prior knowledge on the experiment's effect. To increase the information bound, there are two directions: either having a large prior marginal effect ($\tilde{\mathbf{B}}_0$) or a small uncertainty on the conditional regression effects ($\mathbf \Lambda$). That is, the most helpful experiments are the ones with large marginal effects on nodes with well-known conditional effects. More on practical advice for domain scientists in the Discussion (Section \ref{sec:discussion}).

\subsection{General Independent Prior}

We now investigate whether the information bound is specific to the Normal-MGIG case or whether it exists in the general case when we have a prior distribution on $\mathbf{B}$ independent of $\mathbf \Omega$ assuming that the prior distribution is log concave.

Let $\log p(\mathbf \Omega, \mathbf B)=f(\mathbf B)+g(\mathbf \Omega)$ be the log prior density where  $f(\mathbf{B})$ has a Hessian given by $-\mathbf \Lambda^{-1}\in \mathbb{R}^{kp\times kp}$ 
and $g(\mathbf{\Omega})$ has a Hessian (with respect to unique parameters of $\mathbf{\Omega}$) given by $-\mathbf \Psi\in \mathbb{R}^{\frac{k(k+1)}{2}\times \frac{k(k+1)}{2}}$. Then, the negative Hessian of the log posterior is the sum of the negative Hessian of the log prior and the negative Hessian from log likelihood \eqref{eqn:FI}:
\begin{equation*}
\label{eqn:Info_indep}
    \left[
    \begin{matrix}
    \mathbf{D}_k^T\left(\frac{n}{2} \mathbf{\Omega}^{-1}\otimes \mathbf{\Omega}^{-1}+\mathbf{\Omega}^{-1}\otimes \mathbf{\Omega}^{-1}\mathbf{B}^T\mathbf{X}^T\mathbf{X}\mathbf{B}\mathbf{\Omega}^{-1}\right)\mathbf{D}_k +\mathbf \Psi& 
    -\mathbf{D}_k^T\left(\mathbf{\Omega}^{-1}\otimes (\mathbf{\Omega}^{-1}\mathbf{B}^T\mathbf{X}^T\mathbf{X}) \right)\\
    -(\mathbf{D}_k^T\left(\mathbf{\Omega}^{-1}\otimes (\mathbf{\Omega}^{-1}\mathbf{B}^T\mathbf{X}^T\mathbf{X}) \right))^T & \mathbf\Omega^{-1} \otimes  \mathbf{X}^T \mathbf{X}+\mathbf \Lambda^{-1}
    \end{matrix}
    \right]
\end{equation*}
which we can write as the matrix in \eqref{matrixform} with
\begin{align*}
    \mathbf A &= \mathbf{D}_k^T\left(\frac{n}{2} \mathbf{\Omega}^{-1}\otimes \mathbf{\Omega}^{-1}+\mathbf{\Omega}^{-1}\otimes \mathbf{\Omega}^{-1}\mathbf{B}^T\mathbf{X}^T\mathbf{X}\mathbf{B}\mathbf{\Omega}^{-1}\right)\mathbf{D}_k+\mathbf \Psi, \\
    \mathbf G &= \mathbf C^T = -\mathbf{D}_k^T\left(\mathbf{\Omega}^{-1}\otimes (\mathbf{\Omega}^{-1}\mathbf{B}^T\mathbf{X}^T\mathbf{X}) \right), \\
    \mathbf D &= \mathbf\Omega^{-1} \otimes  \mathbf{X}^T \mathbf{X}+\mathbf \Lambda^{-1}.
\end{align*}

Then, the Laplace approximation of the marginal posterior precision matrix of $\mathbf \Omega$ is given by the inverse of the Schur complement of $\mathbf A$: $(\mathbf A-\mathbf G\mathbf D^{-1}\mathbf C)$, typically denoted as $\mathbf \Omega | I(\mathbf \Omega, \mathbf B)$:
\begin{equation*}
    \begin{aligned}
    \mathbf{\Omega}|I(\mathbf \Omega,\mathbf{B}) =& \mathbf{D}_k^T\left(\frac{n}{2} \mathbf{\Omega}^{-1}\otimes \mathbf{\Omega}^{-1}+\mathbf{\Omega}^{-1}\otimes \mathbf{\Omega}^{-1}\mathbf{B}^T\mathbf{X}^T\mathbf{X}\mathbf{B}\mathbf{\Omega}^{-1}\right)\mathbf{D}_k+\mathbf \Psi\\
    &-\mathbf{D}_k^T\left(\mathbf{\Omega}^{-1}\otimes (\mathbf{\Omega}^{-1}\mathbf{B}^T\mathbf{X}^T\mathbf{X}) \right)\left(\mathbf\Omega^{-1} \otimes  \mathbf{X}^T \mathbf{X}+\mathbf \Lambda^{-1}\right)^{-1}\left(\mathbf{\Omega}^{-1}\otimes (\mathbf{X}^T\mathbf{X}\mathbf{B}\mathbf{\Omega}^{-1}) \right)\mathbf{D}_k.
    \end{aligned}
\end{equation*}

To simplify the notation, we take $\mathbf{E}=\left(\mathbf{\Omega}^{-1}\otimes (\mathbf{\Omega}^{-1}\mathbf{B}^T\mathbf{X}^T\mathbf{X}) \right)$ and $\mathbf F=\mathbf\Omega^{-1} \otimes  \mathbf{X}^T \mathbf{X}$. 
It is simple to check that $\mathbf{EF}^{-1}=\mathbf{I}_{k}\otimes \mathbf \Omega^{-1}\mathbf B^T, \mathbf{F}^{-1}\mathbf{E}^T=\mathbf{I}_{k}\otimes \mathbf {B\Omega}^{-1}$ and $\mathbf{EF}^{-1}\mathbf{E}^{T} = \mathbf{\Omega}^{-1}\otimes \mathbf{\Omega}^{-1}\mathbf{B}^T\mathbf{X}^T\mathbf{X}\mathbf{B}\mathbf{\Omega}^{-1}$.
The Laplace approximation of the marginal posterior precision matrix of $\mathbf \Omega$ ($\mathbf{\Omega}|I(\mathbf \Omega,\mathbf{B})$) then becomes 
\begin{equation*}
    \begin{aligned}
    \mathbf{\Omega}|I(\mathbf \Omega,\mathbf{B}) =& \mathbf{D}_k^T\left(\frac{n}{2} \mathbf{\Omega}^{-1}\otimes \mathbf{\Omega}^{-1}+\mathbf{\Omega}^{-1}\otimes \mathbf{\Omega}^{-1}\mathbf{B}^T\mathbf{X}^T\mathbf{X}\mathbf{B}\mathbf{\Omega}^{-1}\right)\mathbf{D}_k+\mathbf \Psi\\
    &-\mathbf{D}_k^T\left(\mathbf{\Omega}^{-1}\otimes (\mathbf{\Omega}^{-1}\mathbf{B}^T\mathbf{X}^T\mathbf{X}) \right)\left(\mathbf\Omega^{-1} \otimes  \mathbf{X}^T \mathbf{X}+\mathbf \Lambda^{-1}\right)^{-1}\left(\mathbf{\Omega}^{-1}\otimes (\mathbf{X}^T\mathbf{X}\mathbf{B}\mathbf{\Omega}^{-1}) \right)\mathbf{D}_k\\
    =&\mathbf{D}_k^T\left(\frac{n}{2} \mathbf{\Omega}^{-1}\otimes \mathbf{\Omega}^{-1}+\mathbf{EF}^{-1}\mathbf E^T\right)\mathbf{D}_k+\mathbf \Psi -\mathbf{D}_k^T\mathbf{E} \left(\mathbf F+\mathbf \Lambda^{-1}\right)^{-1}\mathbf E ^T \mathbf{D}_k\\
    =&\mathbf{D}_k^T\left(\frac{n}{2} \mathbf{\Omega}^{-1}\otimes \mathbf{\Omega}^{-1}+\mathbf E(\mathbf F^{-1}-(\mathbf F+\mathbf \Lambda^{-1})^{-1})\mathbf E^T\right)\mathbf{D}_k+\mathbf \Psi.
    \end{aligned}
\end{equation*}

Note that the only term that involves $\mathbf X$ is $\mathbf E(\mathbf F^{-1}-(\mathbf F+\mathbf \Lambda^{-1})^{-1})\mathbf E^T$ which we can re-write by taking the Cholesky decomposition of $\mathbf \Lambda^{-1}=\mathbf{LL}^T$ and by \eqref{eqn:inverse_for_positive}
(with $\mathbf A=\mathbf{F}$, $\mathbf{P}=\mathbf I_{kp}$, $\mathbf{U}=\mathbf V^T=\mathbf{L}$) as 
\begin{equation}
\label{eqn:bound_indep}
\begin{aligned}
\mathbf E(\mathbf F^{-1}-(\mathbf F+\mathbf \Lambda^{-1})^{-1})\mathbf E^T&=\mathbf E(\mathbf F^{-1}-(\mathbf F+\mathbf{LL}^T)^{-1})\mathbf E^T\\
&=\mathbf E\mathbf F^{-1}\mathbf L(\mathbf I_{kp}+\mathbf L^T\mathbf F^{-1}\mathbf L)^{-1}\mathbf L^T\mathbf F^{-1}\mathbf E^T.
\end{aligned}
\end{equation}
Then, we have the following information bound (derivation in Appendix \ref{appGeneral}) 

\begin{equation*}
\begin{aligned}
\mathbf E(\mathbf F^{-1}-(\mathbf F+\mathbf \Lambda^{-1})^{-1})\mathbf E^T&=\mathbf E\mathbf F^{-1}\mathbf L(\mathbf I_{kp}+\mathbf L^T\mathbf F^{-1}\mathbf L)^{-1}\mathbf L^T(\mathbf E\mathbf F^{-1})^{T}\\
&\le \mathbf E\mathbf F^{-1}\mathbf \Lambda^{-1} (\mathbf E\mathbf F^{-1})^{T}\\
&=\left(\mathbf{I}_{k}\otimes \mathbf (\mathbf B\mathbf\Omega^{-1})^T\right)\mathbf \Lambda^{-1} \left(\mathbf{I}_{k}\otimes \mathbf (\mathbf B\mathbf\Omega^{-1})^T\right)^T.
\end{aligned}
\end{equation*}

This bound is similar to the Normal-MGIG case (Section \ref{info-bound-NMGIG}) in that it depends on the marginal correlation coefficients $\mathbf{B\Omega}^{-1}$ and is also not a function of $\mathbf X$. We observe that this bound is also sharp as one can achieve equality when $\mathbf{X}^T\mathbf{X}$ tends to infinite. In addition, this bound could guide the choice of experiments when the goal is the estimation of the precision matrix. Again, one should chose experiments that have large marginal effects 
and high certainty in the prior on the conditional effects. More on practical advice for domain scientists in the Discussion (Section \ref{sec:discussion}).

\section{Simulation Study}
\label{sec:simulation}

\revision{
In order for an experimental design to help infer network structure, such experiment should at least provide better results than not doing any experiment (a null experiment corresponding to $\mathbf{X}=0$). In addition, we evaluate a standard type of experiments denoted ``specific'' in which there is only effect to one of the nodes.}
For each experimental design (null, random, specific), we calculate the KL divergence between the prior and the posterior which represents the information gain for each experiment. For evaluation of the point estimation, we compare the performance of the different priors with the Stein's loss \citep{dey1985estimation} between
the maximum a posteriori (MAP) and the true value of $\mathbf \Omega$.

We simulate data under the 6 covariance structures in \citet{glasso} with $k=50$ responses and $p=50$ specific predictors ($\mathbf B=\mathbf I_{50}$). Each simulation is repeated 100 times.

\begin{itemize}
    \item Model 1: AR(1) model with $\sigma_{ij}=0.7^{|i-j|}$
    \item Model 2: AR(2) model with $\omega_{ii}=1$, $\omega_{i-1,i}=\omega_{i,i-1}=0.5$, $\omega_{i-2,i}=\omega_{i,i-2}=0.25$ for $i=1,\dots,k$
    \item Model 3: Block model with $\sigma_{ii}= 1$  for $i=1,\dots,k$, $\sigma_{ij}= 0.5$ for $1\le i\ne j\le k/2$, $\sigma_{ij}=0.5$ for $k/2 + 1\le i\ne j\le 10$ and $\sigma_{ij=0}$ otherwise.
    \item Model 4: Star model with every node connected to the first node, with $\omega_{ii}=1$, $\omega_{1,i}=\omega_{i,1}= 0.1$ for $i=1,\dots,k$, and $\omega_{ij}= 0$ otherwise.
    \item Model 5: Circle model with $\omega_{ii}= 2$, $\omega_{i-1,i}=\omega_{i,i-1}= 1$ for $i=1,\dots,k$, and $\omega_{1,j}=\omega_{j,1}= 0.9$ for $j=1,\dots,k$.
    \item Model 6: Full model with $\omega_{ii}= 2$ and $\omega_{ij}= 1$ for $i\ne j \in \{1,\dots,k\}$.
\end{itemize}

We take sample sizes ranging from 200 to 2200 increased by 50. 

For each covariance structure, we test two priors, each with two levels of uncertainty, \revision{and with and without bias in the prior of $\mathbf{B}$}. Namely, 
\begin{enumerate}
    \item Normal-Wishart prior with $\lambda=2k+2$, 
    $\mathbf \Phi=10^{-3} \mathbf I_k$ and two uncertainty levels of $\mathbf{B}$: i) $\mathbf \Lambda=10^{-3} \mathbf I_p$ (certain case) and ii) $\mathbf \Lambda=10^3 \mathbf I_p$ (uncertain case). 
    \item Normal-MGIG prior with $\lambda=k+1$, $\mathbf \Psi=\mathbf \Phi=10^{-3} \mathbf I_k$ and two uncertainty levels of $\mathbf{B}$: i) $\mathbf \Lambda=10^{-3}\mathbf I_{p}$ (certain case) and ii) $\mathbf \Lambda=10^3 \mathbf I_p$ (uncertain case).
\end{enumerate}
\revision{The biased case for the prior of $\mathbf{B}$ is given by the biased prior mean of $\mathbf B_0 = \mathbf B + \epsilon$ for $\epsilon \sim N(0,1)$. Note that we also test a case with a smaller bias and present these results in the Appendix.}
\revision{In addition, we test for midpoint uncertainty where $\mathbf \Lambda = 10 \mathbf I_{p}$ and $\mathbf \Lambda = 0.1 \mathbf I_{p}$ in the Appendix.} 

\subsection{Information Gain by Experiments}
\label{other_measures}

We evaluated how much the posterior distribution changes compared to the prior under different experimental settings. Intuitively, if the prior and the posterior are similar, then no much information is gained by the data or the experiment. \revision{However, simply evaluating the difference between prior and posterior is not enough since sample size surely has an effect. Thus, to better evaluate the effect of the design rather than the effect of the sample size, we need to choose a baseline experiment. We take a null design (design with $\mathbf{X}=0$) as such baseline because one would expect it to perform poorly. 
That is, we consider an experiment to provide useful information when it is better than a null experiment in the same settings, and thus, calculate the difference between the given design (random or specific) and the null design as our measure of performance shown in figures below. 
}

For each simulation at a given sample size, we calculate 1) the KL divergence between the prior and the posterior for a given experiment (random or specific), and 2) the KL divergence between the prior and posterior under the null experiment. The information gain due to experiment is then evaluated by taking difference of log KL divergence between the design (random or specific) and the null design.  \revision{Figure \ref{fig:star_kl} shows the results when simulating data under the covariance structure of Model 4 (star model). Points above 0 indicate more information gained by conducting a non-null experiment. We expect random experiments to do better in this case because the \textit{marginal} effect $\mathbf{B\Omega^{-1}}$ can be large. The results corresponding to all other covariance structures (with similar conclusions) can be found in the Appendix~\ref{app:other_graphs}.}

\begin{figure}[htp]
    \centering
    \includegraphics[width = 0.8\linewidth]{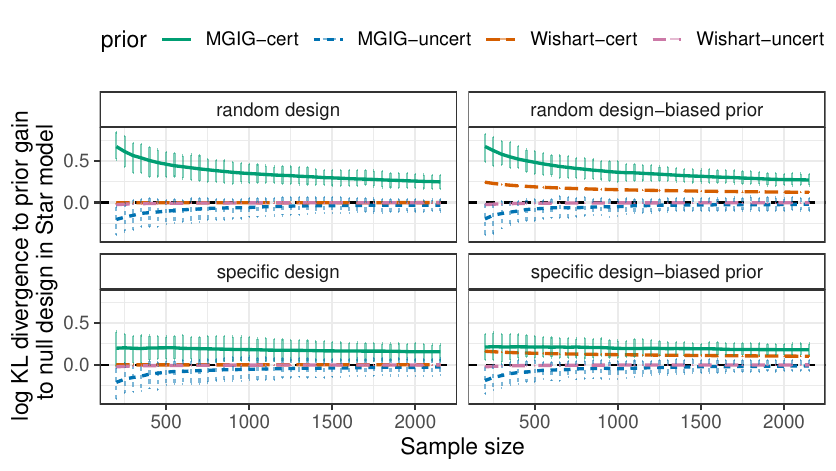}
    \caption{Difference in log KL divergence between prior and posterior comparing random experiment ($\mathbf X \ne \mathbf 0$) and specific experiment (diagonal $\mathbf X$) vs null experiment ($\mathbf X=\mathbf 0$) under a star model with 50 responses and 50 predictors, with and without bias on the prior of $\mathbf{B}$. Lines are averages over 100 repeats while error bars are 0.975 and 0.025 quantiles. The only case in which the experiment gains information compared to the null experiment is when we have certain and unbiased prior knowledge on $\mathbf B$ (Normal-MGIG certain case in green). \revision{With a biased prior, MGIG still has better information gain while certain Wishart also leads to better information gains compared to a null design.}}
    \label{fig:star_kl}
\end{figure}

Except for the case of Normal-MGIG prior with $\mathbf \Lambda=10^{-3} \mathbf I_p$ (certain case), the difference in information gain will eventually reach 0 for all other prior cases as the sample size increases, meaning that there will no longer be any information gain from doing an experiment (random nor specific) compared to no experiment at all.
While the Normal-MGIG prior with $\mathbf \Lambda=10^{-3} \mathbf I_p$ (certain case) stays at a distance from 0, this distance does not depend much on the sample size. \revision{With a biased prior, MGIG still shows better information gain while certain Wishart also leads to better information gains compared to a null design. More results on other types of biases can be found in the Appendix~\ref{app:less_biased}.}

\subsection{Performance on Point Estimation of $\mathbf \Omega$}

While KL divergence evaluates the information gain of the experiment, information gain does not imply better point estimates. To compare the performance of point estimation of $\mathbf \Omega$, we use the difference in Stein's loss of the experiment (random or specific) and the null design. The results are shown in Figure \ref{fig:simu}. Points below 0 indicate more accurate point estimates under an experimental design rather than under a null experiment.

We have a similar observation as in the information gain section. That is, except for the case of Normal-MGIG prior with $\mathbf \Lambda=10^{-3} \mathbf I_p$ (certain case), the Stein's loss ratio will eventually reach 0 for all other prior cases regardless of the experimental design (random or specific). While the Normal-MGIG prior with $\mathbf \Lambda=10^{-3} \mathbf I_p$ (certain case) stays at a distance from 0, this distance does not depend much on the sample size. 
\revision{While biased prior such as certain MGIG and certain Wishart can undermine our ability to get good point estimates from experiments. smaller biases do not have such a strong negative effect (see Appendix~\ref{app:less_biased}).}

\begin{figure}[htp]
    \centering
    \includegraphics[width = 0.8\linewidth]{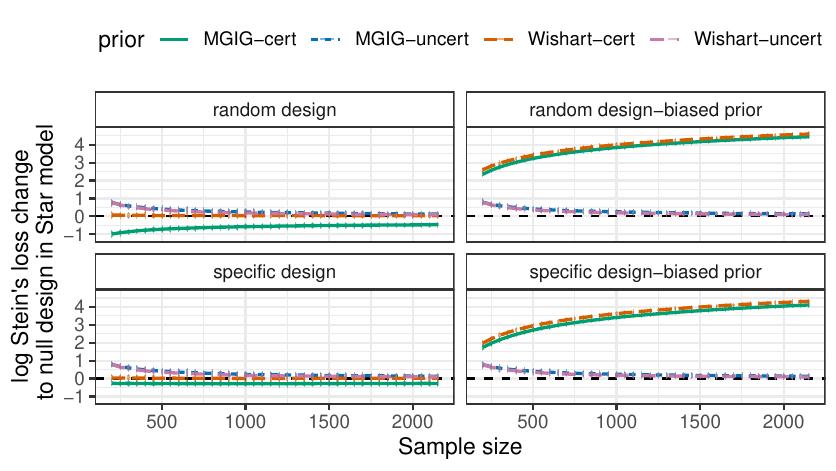}
    \caption{Difference in log Stein's loss of random experiment ($\mathbf X \ne \mathbf 0$) and specific experiment (diagonal $\mathbf X$) vs null experiment ($\mathbf X=\mathbf 0$) under Star models with 50 responses and 50 predictors, with and without biases on the prior of $\mathbf{B}$. Lines are averages over 100 repeats while error bars are 0.975 and 0.025 quantiles. The only case in which the random design mostly has lower Stein's loss is when we have certain and unbiased prior knowledge on $\mathbf B$ (Normal-MGIG certain case in green). All other unbiased prior cases eventually reach the zero line (no difference in MAP performance of $\mathbf \Omega$ compared to null experiment). \revision{With a biased prior, certain MGIG and certain Wishart lead to less accurate point estimates.}}
    \label{fig:simu}
\end{figure}

\section{Human Gut Microbiome Data}
\label{microbiomedata}

We revisit a similar comparison to the one in the motivating toy example (Section \ref{toy}) of the posterior distribution of partial correlations among responses under different priors and experiments. 
We use data from \citet{claesson2012gut} which collected fecal microbiota composition from 178 elderly subjects, together with subjects' residence type (in the community, day-hospital, rehabilitation or in long-term residential care) and diet (data at \citet{gut_dataset}) with the goal of understanding the interactions between microbes and environment via partial correlations.

We use the MG-RAST server \citep{meyer2008metagenomics} for profiling with an e-value of 5, 60\% identity, alignment length of 15 bp, and minimal abundance of 10 reads. Unclassified hits are not included in the analysis. Genus with more than 0.5\% relative abundance in more than 50 samples is selected as the focal genus and all other genus serve as the reference group. This yield 13 responses and 11 predictors (i.e. $p=11,k=14$)
We then fit a Gaussian chain graph model to the data.
Since we cannot design an experiment on this already collected data, to compare with a hypothetical null experiment, we draw a simulated sample from a Gaussian chain graph model whose regression coefficients and precision matrix are the MLE from the original data and a null experiment by definition has $\mathbf{X} = \mathbf{0}$.

Here, we focus on the partial correlation between \textit{Bacteroides}, one of the largest genera in gut and \textit{Clostridium}, a group known to be pathological. Figure \ref{fig:gut} shows the posterior distribution of this partial correlation under the two different priors (Normal-Wishart and Normal-MGIG) with three different uncertainty levels on $\mathbf B$: $\mathbf{\Lambda}=10^{-1}\mathbf{I}_p$, $10^0 \mathbf{I}_p$, $10^{-1}\mathbf{I}_p$ (shown as columns: 0.1, 1, 10) with $\mathbf{\Psi}= 0.01\mathbf{I}_p$ and $\mathbf{\Phi}=0.01\mathbf{I}_p$. 
Just as in the toy example (and in agreement with the expectation of our theory in Section \ref{priors}), only the Normal-MGIG prior results in lower uncertainty when an experiment is performed.

\begin{figure}
    \centering
    \includegraphics[width = \linewidth]{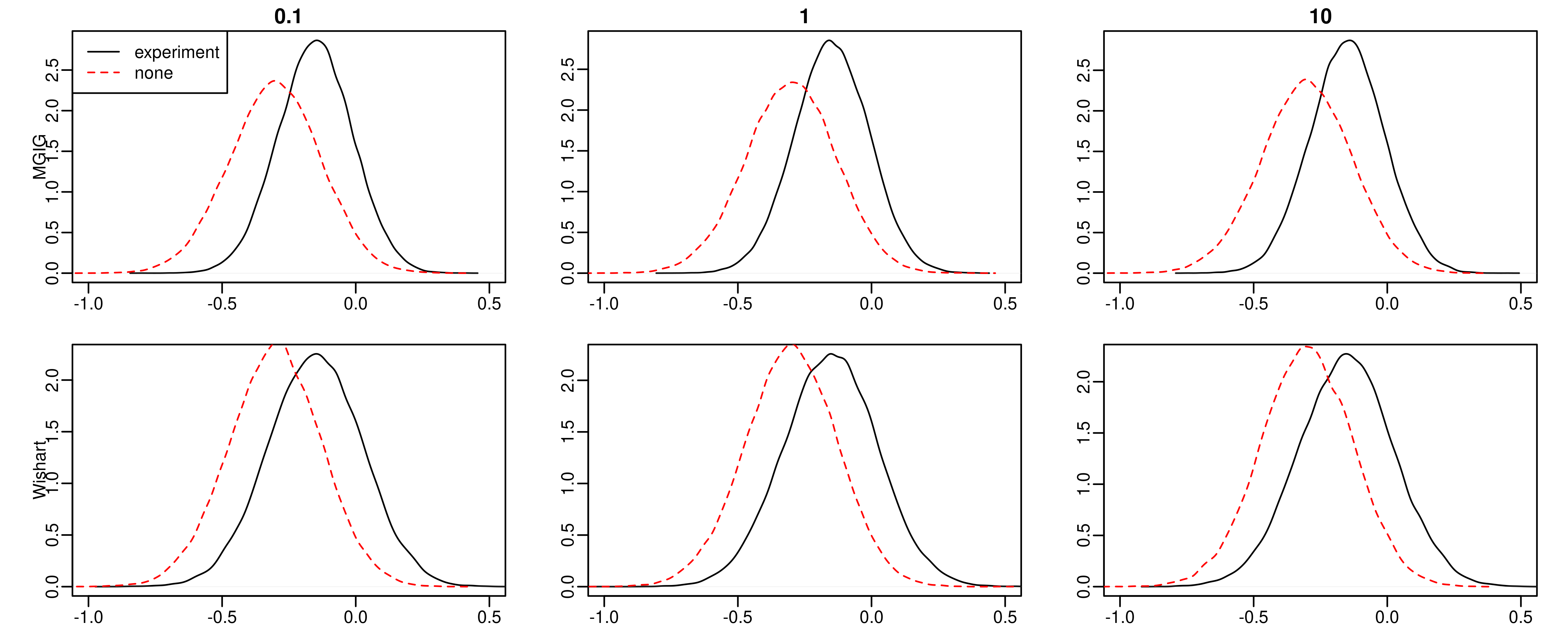}
    \caption{Posterior distribution for the partial correlation between \textit{Bacteroides}, one of the largest genera in gut and \textit{Clostridium}, a group known to be pathological. Rows correspond to the two priors: Normal-MGIG (top) and Normal-Wishart (bottom). Columns correspond to the certainty level on the prior of $\mathbf B$ ($\mathbf{\Lambda}=10^{-1}\mathbf{I}_p$, $10^0 \mathbf{I}_p$, $10^{-1}\mathbf{I}_p$; shown as 0.1, 1, 10 respectively).  We observe that only under the MGIG prior, the experiment reduces uncertainty in the posterior. }
    \label{fig:gut}
\end{figure}

\section{Discussion}
\label{sec:discussion}

Chain graph models are relevant in genomic, microbiome and ecological applications because they encode conditional dependence among responses and predictors and focus on the estimation of the precision matrix, an important parameter to understand interactions within microbial and ecological communities.
Here, we evaluated the effect of prior knowledge on conducting experiments to better estimate the precision matrix in a chain graph model. Using the Laplace approximation of the marginal posterior precision matrix of $\mathbf \Omega$ as the optimality criterion on experimental design settings, we proved theoretically that without prior knowledge that identifies $\mathbf B$ and $\mathbf \Omega$ separately (instead of $\mathbf{B \Omega}^{-1}$ combined), experiments provide no gain in knowledge for the estimation of $\mathbf \Omega$. That is, the Laplace approximation of the marginal posterior precision matrix of $\mathbf \Omega$ is not a function of $\mathbf X$.
We also showed a bound on the information gain under the Normal-MGIG prior which generalizes to the case of any independent priors. Our findings are highly relevant for domain scientists who aim to design optimal experimental designs to infer the precision matrix.

We further verified our theoretical conclusions using numerical simulations where we showed that without certain prior knowledge on $\mathbf{B}$, experiments provide nearly no information gain and there is not an increase performance on the estimation of $\mathbf \Omega$ either.
Furthermore, it is not enough for an experiment to be specific, the prior knowledge about this specificity is also needed (more examples below in Practical advice for domain scientists). 

\paragraph{Connections to multicollinearity.}
Chain graph models have a dependence property that is similar to multicollinearity in classical regression. Take the conditional distribution of the $q$th response node in sample $\mathbf Y_i \in \mathbb{R}^k$ with the design $\mathbf X_i \in \mathbb{R}^p$:
\begin{equation}
\label{eqn:cond_dist_chain}
	\begin{aligned}
	\left[Y_{qi}|\mathbf X_i= \mathbf x_i,\mathbf Y_{-q,i}= \mathbf y_{-q,i}\right]= \frac{1}{\omega_{qq}}\sum_{j=1}^p\beta_{jq}x_{ji}-\frac{1}{\omega_{qq}}\sum_{l\ne q} \omega_{ql} y_{li}+\epsilon_{qi}
	\end{aligned}
\end{equation}
where $\epsilon_{qi}\sim \mathcal N(0,1/\omega_{qq})$, $\beta_{jq}$ is the $(j,q)$ entry of the $\mathbf B$ matrix, $\omega_{qq}$ and $\omega_{ql}$ are the $(q,q)$ and $(q,l)$ entries in the $\mathbf \Omega$ matrix, and $\mathbf Y_{-q,i}$ corresponds to the vector of responses for sample $i$ without the $q$th response. Multicollinearity in this model arises given that the correlation between $\mathbf Y_l$ and $\mathbf X_j$ is 0 only if when the $(j,l)$ entry of $\mathbf{B\Omega}^{-1}$ is 0, which requires $\mathbf B_{j\cdot} \mathbf \Omega^{-1}_{\cdot l}=0$, which is difficult to hold for all $l,j$. Thus, in practice, we are most likely to have some intrinsic multicollinearity in chain graph models.
In univarite settings, in principle, we could design experiments to avoid problems caused by multicollinearity. When such experiments are hard to conduct, an alternative approach is to have informative priors on some of the parameters. For instance, if two predictors are collinear, the sum of the two respective regression coefficients can be easily identified, but not the individual ones. However, one individual coefficient can be identified if we have informative prior on the other. 
Take this intuition to chain graph where we have multicollinearty between $\omega$'s and $\beta$'s. In this case, prior knowledge on the regression coefficients $\mathbf{B}$ might actually help the estimation of $\mathbf{\Omega}$ under certain experimental conditions. 
Thus, our work to explore the interplay between experimental design and prior knowledge in chain graph models is also justified by the classical regression setting that has routinely used prior knowledge to infer parameters on cases when multicollinearity arises.

\paragraph{Future directions.} 
As mentioned, the difficulty of designing experiments for the estimation of the precision matrix on a chain graph model is similar to the multicollinearity problem in univariate regression. In both cases, a prior can help identify one set of parameters to better infer another set of parameters.
A natural question is whether this is also true in a general Gibbs measure with two-body interaction. For instance, 
the auto-logistic model (Ising model) can been used to infer networks and existing work has discussed the experimental design of this model when the effect of the treatment is completely known \citep{jiang2019active}. 
One interesting question is whether 
in order to design experiments effectively to infer the network among responses with this model, we need prior knowledge on the effect of the treatment. One difficulty when answering this question is the intractable normalizing constant or partition function in a general Gibbs measure. Some approximations of the partition function proposed by \citet{Wainwright2006} might be helpful to connect it with what we already know for the Gaussian case.

\paragraph{Practical advice for domain scientists.}
As we presented in both theoretical and simulation studies, for experiments to aid in the estimation of precision matrix under a chain graph model, the experiment should have 
large marginal effects, prior knowledge on the conditional effect of predictors on responses and high certainty on those conditional effects. For instance, if an experimenter wants to understand a microbial community, she could try different candidate experiments on the community to identify the treatment that alter the community the most (i.e. that has large marginal effects). Then, the experimenter should culture several of the species to evaluate the effect of those candidate treatments (i.e. gain prior knowledge of $\mathbf B$). By focusing on a some single species, the experimenter will (ideally) have high certainty on the conditional effects of some of the experiment.

Similarly, an experiment where the experimenter knocks out one gene and evaluates the reaction of another gene in order to infer the interaction between two genes 
is useful 
because it can affect two target genes (marginal effect) while we know (by assumption) it is specific to one of the genes (good prior knowledge of its conditional effect). While there is a keen interest in experiments that are specific (e.g. gene knockout), our theory shows that specificity (that is, the row in $\mathbf B$ has only one non-zero entry) is not necessary for the experiment to be useful in the inference of the precision matrix ($\mathbf \Omega$). However, specificity is helpful given that it is easier to obtain prior knowledge on specificity (we are certain about some entries in $\mathbf{B}$ being zero) than trying to obtain prior knowledge on multiple non-zero entries of $\mathbf{B}$.

For any given experiment, the experimenter has control over the design ($\mathbf X$) and the prior of the regression coefficients ($\mathbf B$). Our findings show that without prior knowledge of $\mathbf B$ on some single predictors, experiments produce zero information gain on the estimation of $\mathbf \Omega$. Our work draws attention to the importance of thorough analysis of priors and experimental design for domain scientists who aim to infer biological network structures from controlled experimental data.

\section*{Declarations}
\subsection*{Funding and/or Conflicts of interests/Competing interests}
This material is based upon work support by the National Institute of Food and Agriculture, United States Department of Agriculture, Hatch project 1023699.
This work was also supported by the Department of Energy [DE-SC0021016 to CSL].
There are no conflict of interests or competing interests to declare.

\bibliography{references.bib}

\newpage

\appendix

\renewcommand{\theequation}{A\arabic{equation}}
\renewcommand{\thesection}{A\arabic{section}}  
\renewcommand{\thefigure}{A\arabic{figure}}  
\renewcommand{\thesubsection}{A\arabic{subsection}} 
\renewcommand{\thesubsubsection}{A\arabic{subsubsection}} 
\setcounter{equation}{0}
\setcounter{figure}{0}
\setcounter{section}{0}
\setcounter{subsection}{0}
\setcounter{subsubsection}{0}

\section{Hessian of the likelihood function: Fisher information matrix}
\label{appendix-fi}

Consider the log likelihood of the Gaussian chain graph model $\mathbf Y \sim N(\mathbf{\Sigma}\mathbf{B^T X^T},\mathbf\Sigma)$:
\[
\ell=\frac{n}{2}\log(2\pi|\mathbf \Omega|)+\tr(\mathbf Y^T\mathbf X\mathbf B)-\frac{1}{2}\tr(\mathbf Y^T\mathbf Y \mathbf \Omega )-\frac{1}{2}\tr(\mathbf X^T\mathbf X \mathbf B \mathbf \Omega^{-1} \mathbf B^T ).
\]

Since $\mathbf \Omega$ is symmetric, there are fewer free parameters and we have a constraint. Following \citet{minka2000old,magnus2019matrix}, we use the duplication
matrix $\mathbf D_k$, a permutation matrix such that $\mathbf D_k\vect{\mathbf\Omega}=\vech(\mathbf\Omega)$ where $\vech(\mathbf\Omega)$ denote the vectorization of unique parameters in $\mathbf\Omega$ (upper triangular part in our case).

The first term has a Hessian of $-\frac{n}{2}\mathbf{D}_k\mathbf{D}_k^T(\mathbf \Omega^{-1}\otimes \mathbf\Omega^{-1})\mathbf{D}_k\mathbf{D}_k^T$ (Equation 121 in \citet{minka2000old}).

Let $l=\tr(\mathbf Y^T\mathbf X\mathbf B)-\frac{1}{2}\tr(\mathbf Y^T\mathbf Y \mathbf \Omega )-\frac{1}{2}\tr(\mathbf X^T\mathbf X \mathbf B \mathbf \Omega^{-1} \mathbf B^T )$. Following \citet{minka2000old}, the differential is
$dl=\tr(\mathbf{Y}^T\mathbf{X}d\mathbf{B})-\frac{1}{2}tr(\mathbf{Y}^T\mathbf{Y}d\mathbf{\Omega})+\frac{1}{2}\tr(\mathbf{\Omega}^{-1}\mathbf{B}^T\mathbf{X}^T\mathbf{X}\mathbf{B}\mathbf{\Omega}^{-1}d\mathbf{\Omega})-\tr(\mathbf{\Omega}^{-1}\mathbf{B}^T\mathbf{X}^T\mathbf{X}d\mathbf{B})$
and the second order differential is 
$d^2 l=-\tr(d\mathbf{\Omega}\mathbf{\Omega}^{-1}\mathbf{B}^T\mathbf{X}^T\mathbf{X}\mathbf{B}\mathbf{\Omega}^{-1}d\mathbf{\Omega}\mathbf{\Omega}^{-1}) +\tr(\mathbf{\Omega}^{-1}\mathbf{B}^T\mathbf{X}^T\mathbf{X}d\mathbf{B}\mathbf{\Omega}^{-1}d\mathbf{\Omega}) -\tr(\mathbf{X}^T\mathbf{X}d\mathbf{B}\mathbf{\Omega}^{-1}d\mathbf{B}^T)$.

Thus, the full Hessian is given by 
\begin{equation*}
    \begin{aligned}
    \frac{\partial^2 \ell}{\partial \vect(\mathbf{B})\partial \vect(\mathbf{B})^T}&=-\mathbf{\Omega}^{-1}\otimes \mathbf{X}^T\mathbf{X}\\
    \frac{\partial^2 \ell}{\partial \vech(\mathbf{\Omega})\partial \vech(\mathbf{\Omega})^T}&=-\mathbf{D}_k^T\left(\frac{n}{2} \mathbf{\Omega}^{-1}\otimes \mathbf{\Omega}^{-1}+\mathbf{\Omega}^{-1}\otimes \mathbf{\Omega}^{-1}\mathbf{B}^T\mathbf{X}^T\mathbf{X}\mathbf{B}\mathbf{\Omega}^{-1}\right)\mathbf{D}_k\\
    \frac{\partial^2 \ell}{\partial \vect(\mathbf{B})\partial \vech(\mathbf{\Omega})^T}&=\mathbf{D}_k^T\left(\mathbf{\Omega}^{-1}\otimes (\mathbf{\Omega}^{-1}\mathbf{B}^T\mathbf{X}^T\mathbf{X}) \right).
    \end{aligned}
\end{equation*}

\section{Log concavity of Normal-Wishart conjugate prior}
\label{logNW}

The Hessian of the log Normal-Wishart prior is given by  \eqref{eqn:hess_wish}.
We observe similarities between the Hessian of the log prior and the negative Fisher information matrix  \eqref{eqn:FI} which is negative definite.
Namely, the first and third partial derivatives in  \eqref{eqn:hess_wish} coincide with the lower right diagonal block and the off-diagonal block respectively of the negative Fisher information matrix \eqref{eqn:FI} of a sampling model in which we set the design matrix to be $\mathbf X^T=\chol(\mathbf \Lambda^{-1})$ with sample size $k$.

We also observe that the second partial derivative in  \eqref{eqn:hess_wish} has the extra term of  $(\frac{1}{2}(\lambda-k-p-1)- k/2)\mathbf \Omega^{-1}\otimes\mathbf \Omega^{-1}$ when compared to the upper left diagonal block of the negative Fisher information matrix  \eqref{eqn:FI} of a sampling model in which we set the design matrix to be $\mathbf X^T=\chol(\mathbf \Lambda^{-1})$ with sample size $k$. Since $\mathbf \Omega$ positive definite, then $(\frac{1}{2}(\lambda-k-p-1)- k/2)\mathbf \Omega^{-1}\otimes\mathbf \Omega^{-1}$ is positive semi-definite when $\frac{1}{2}(\lambda-k-p-1)\ge k/2$, and thus, the Hessian would be negative definite.

\section{Log concavity, conjugacy and unimodality of Normal-MGIG prior}
\label{propNMGIG}

To the best of our knowledge, the Normal-MGIG prior has not been used for the Gaussian chain graph model, and thus, we prove here some of its properties.
We show that the MGIG prior is conjugate for the case of known $\mathbf B$ (Proposition \ref{conjB}), that it is log concave under certain conditions (Proposition \ref{logconcave}) which is needed for the Laplace approximation, {that it is unimodal for the case of unknown $\mathbf B$ (Proposition \ref{unimod})}, and that its limiting case is indeed the case of known $\mathbf B$ (Remark \ref{NMGIG-limit}).


\begin{myProposition}
\label{conjB}
Under the setting of known $\mathbf B$, the Matrix Generalized Inverse Gaussian (MGIG) distribution is a conjugate prior for $\mathbf \Omega$ with density:
\begin{equation*}
    p(\mathbf{\Omega}|\lambda,\mathbf \Psi,\mathbf \Phi)\propto |\mathbf{\Omega}|^{\lambda-\frac{1}{2}(k+1)}\exp(-\frac{1}{2}\tr(\mathbf \Psi \mathbf{\Omega}^{-1})-\frac{1}{2}\tr(\mathbf \Phi\mathbf{\Omega}))1_{\mathbf{\Omega}+}
\end{equation*}

where $\mathbf{\Psi,\Phi}\in \mathbb{R}^{k\times k}$ are positive definite and $\lambda$ is a scalar. 
\end{myProposition}

\begin{proof}
Denote $\mu=\mathbf{XB}$ and let $\theta$ represent the hyper-parameters. Consider the posterior distribution:
\begin{equation}
    \begin{aligned}
p(\mathbf{\Omega}|\mathbf{Y}, \mu,\theta)\propto& |\mathbf{\Omega}|^{n/2}\exp(\tr(\mathbf{Y}^T\mu)-\frac{1}{2}\tr(\mathbf{Y}^T\mathbf{Y}\mathbf{\Omega})-\frac{1}{2}\tr(\mu^T\mu\mathbf{\Omega}^{-1}))\\
&\times |\mathbf{\Omega}|^{\lambda-\frac{k+1}{2}}\exp(-\frac{1}{2}\tr(\mathbf \Psi \mathbf{\Omega}^{-1})-\frac{1}{2}\tr(\mathbf \Phi\mathbf{\Omega}))1_{\mathbf{\Omega}+}\\
\propto& |\mathbf{\Omega}|^{\lambda+\frac{n}{2}-\frac{k+1}{2}}\exp(-\frac{1}{2}\tr[(\mathbf \Psi+\mu^T\mu)\mathbf{\Omega}^{-1}]-\frac{1}{2}\tr[(\mathbf \Phi+\mathbf{Y}^T\mathbf{Y})\mathbf{\Omega}])
\end{aligned}
\label{eqn:MGIG_posterior}
\end{equation}

which is a MGIG distribution with parameters $\lambda+\frac{n}{2}$, $\mathbf \Psi+\mu^T\mu$ and $\mathbf \Phi+\mathbf{Y}^T\mathbf{Y}$.
\end{proof}




\begin{myProposition}[Log-Concavity]
\label{logconcave}
The (Normal-)MGIG prior is log concave under both settings: known and unknown $\mathbf B$ if $\lambda-\frac{2k+1}{2}\ge 0$.
\end{myProposition}

\begin{proof}
We observe similarities between the Hessian of the log prior and the negative Fisher information matrix \eqref{eqn:FI} which is negative definite.

For the case of known $\mathbf B$, we observe that the log density of the MGIG prior has a similar form to the log density of the model Normal distribution
\begin{align}
\log p(\mathbf Y | \mathbf{X,\Omega,B}) = \frac{n}{2} \log |\mathbf{\Omega}|-\frac{1}{2} \tr(\mathbf{Y}^T\mathbf{Y}\mathbf{\Omega})- \frac{1}{2}\tr(\mathbf{B}^T\mathbf{X}^T\mathbf{X}\mathbf{B}\mathbf{\Omega}^{-1}) + \tr(\mathbf{Y}^T\mathbf{X}\mathbf{B}) + C.
\label{eqn:normal}
\end{align}

That is, if we set $\mathbf B=\mathbf I_k$ and $\mathbf X^T=\chol(\mathbf \Psi)$ then we have $\mathbf \Psi=\mathbf B^T\mathbf X^T \mathbf{XB}$ so that the log density of MGIG can be re-written as 
\begin{align}
\log p(\mathbf \Omega | \lambda, \mathbf{\Psi,\Phi})=\frac{k}{2}\log(|\mathbf \Omega|)-\frac{1}{2}\tr(\mathbf{\Phi\Omega})-\frac{1}{2}\tr(\mathbf{B}^T\mathbf X^T\mathbf{XB\Omega}^{-1})+\left(\lambda-\frac{2k+1}{2}\right)\log(|\mathbf \Omega|)+C.
\label{eqn:hessian-mgig}
\end{align}

Note that the first and third terms in  \eqref{eqn:hessian-mgig} coincide with the first and third terms in the Normal model  \eqref{eqn:normal}, and thus, the Hessian of these terms will coincide with
the negative Fisher information matrix of the Normal model \eqref{eqn:FI}. Note that we ignore the second term in  \eqref{eqn:hessian-mgig} because it is of first order and thus, it will not appear in the second derivative. Thus, the log MGIG prior would be concave as long as the last term in  \eqref{eqn:hessian-mgig} is concave, and this happens when $\lambda-\frac{2k+1}{2}\ge 0$.


For the case of unknown $\mathbf B$, the Hessian of the log Normal-MGIG prior has the following form  \eqref{eqn:prior_good_Hessian}.
We observe that the first partial derivative in  \eqref{eqn:prior_good_Hessian} coincides with the lower right diagonal block in the negative Fisher information of the Normal Chain graph model \eqref{eqn:FI}
with a sampling model in which we set the design matrix to be $\mathbf X^T=\chol(\mathbf \Lambda^{-1})$ with sample size $p$.
Next, we observe that the third partial derivative in  \eqref{eqn:prior_good_Hessian} coincides with the off-diagonal block in the negative Fisher information of the Normal Chain graph model  \eqref{eqn:FI} with the same sampling model already described ($\mathbf X^T=\chol(\mathbf \Lambda^{-1})$) and a regression coefficient matrix given by $\mathbf B-\mathbf B_0$.
For the missing block (second partial derivative in  \eqref{eqn:prior_good_Hessian}), there is an extra term of $ \left( \lambda - p +\frac{k+1}{2}\right) \mathbf \Omega^{-1} \otimes \mathbf \Omega^{-1} + \mathbf \Omega^{-1} \otimes \mathbf \Omega^{-1} \mathbf \Psi \mathbf \Omega$.
%
Since $\mathbf \Omega$ is positive definite, then we simply need to show that the terms $\mathbf \Omega^{-1}\mathbf \Psi \mathbf \Omega^{-1}+(\lambda-p-\frac{k+1}{2})\mathbf \Omega^{-1}$ are positive definite. The term $\mathbf \Omega^{-1}\mathbf \Psi \mathbf \Omega^{-1}$ is positive definite because it is a quadratic form of a positive definite matrix $\mathbf \Psi$. For the second term, when $\lambda-\frac{k+p+1}{2}\ge \frac{p}{2}$, $\mathbf{\Omega}^{-1}$ has a non-negative coefficient so that $(\lambda-p-\frac{k+1}{2})\mathbf \Omega^{-1}$ is also positive semi-definite. Then, $\mathbf \Omega^{-1}\otimes((\lambda-p-\frac{k+1}{2})\mathbf \Omega^{-1} + \mathbf \Omega^{-1}\mathbf \Psi \mathbf \Omega^{-1})$ is itself positive definite when $\lambda-\frac{k+p+1}{2}\ge \frac{p}{2}$, and thus the prior is indeed log concave. 
\end{proof}

For the case of known $\mathbf{B}$, it is already known that the MGIG prior is unimodal \citep{fazayeli2016matrix}. Next, in Proposition \ref{unimod}, we show that the Normal-MGIG prior is unimodal for the case of unknown $\mathbf B$.

\begin{myProposition}[Unimodality]
\label{unimod}
For unknown $\mathbf B$, the Normal-MGIG prior and corresponding posterior are unimodal.
\end{myProposition}

\begin{proof}
Unimodality of prior follows by taking derivative of log density of the prior and set it to 0 to show that it has a unique solution. 
The solution of $\frac{\partial }{\partial\mathbf B}p(\mathbf{\Omega,B})=0$ is $\hat{\mathbf{B}}=\mathbf B_0 $. Then, we take the partial derivative with respect to $\mathbf \Omega$ of the log prior and plug in the solution $\hat{\mathbf{B}}$:
\begin{equation*}
    \begin{aligned}
        \frac{\partial}{\partial \mathbf \Omega} \log p(\mathbf \Omega, \hat{\mathbf{B}})&=\beta \mathbf \Omega^{-1}-\frac{1}{2}\mathbf \Phi + \frac{1}{2}\mathbf \Omega^{-1}\left[\mathbf \Psi+(\hat{\mathbf B}-\mathbf B_0)^T\mathbf\Lambda^{-1}(\hat{\mathbf B}-\mathbf B_0)\right]\mathbf \Omega^{-1}
    \end{aligned}
\end{equation*}
where $\beta = \lambda-\frac{k+1}{2}-\frac{p}{2}$.


By setting the derivative to 0 and multiplying by $\mathbf \Omega$ on both sides (left and right), 
we get the equation:
$-2\beta \mathbf \Omega + \mathbf{\Omega}\mathbf \Phi \mathbf \Omega-\mathbf \Psi =0$
which is a special case of continuous-time algebraic Riccati equation (CARE) \citep{boyd1991linear,anderson2007optimal}. Since matrices $\mathbf \Phi$ and $\mathbf{\Psi}$ are positive definite, this equation has the exact form in the proof of unimodality of MGIG distribution from \citet{fazayeli2016matrix}, and thus, we can conclude that the prior has unique solution and is unimodal.

Unimodality of posterior follows the same steps: we take partial the derivative of the posterior \eqref{eqn:prior_good} with respect to $\mathbf B$ and set it to 0 and let $\hat{\mathbf{B}}=(\mathbf X^T \mathbf X+\mathbf \Lambda^{-1})^{-1}(\mathbf X^T \mathbf Y+\mathbf \Lambda^{-1}\mathbf B_0 \mathbf \Omega^{-1}) \mathbf \Omega$ be that solution.
We then take the partial derivative of the log posterior  \eqref{eqn:prior_good} with respect to $\mathbf \Omega$ and we plug in the solution $\hat{\mathbf B}$. Let $\alpha = \lambda+\frac{n}{2}-\frac{k+1}{2}-\frac{p}{2}$, so that we get
\begin{equation*}
    \begin{aligned}
        \frac{\partial}{\partial \mathbf \Omega} \log p(\mathbf \Omega, \hat{\mathbf{B}})&=\alpha \mathbf \Omega^{-1}-\frac{1}{2}(\mathbf{Y^TY}+\mathbf{\Phi}) + \frac{1}{2}\mathbf \Omega^{-1}\left[\mathbf \Psi+\hat{\mathbf B}^T\mathbf X^T\mathbf X\hat{\mathbf B}+(\hat{\mathbf{B}}-\mathbf B_0)^T\mathbf \Lambda^{-1}(\hat{\mathbf{B}}-\mathbf B_0)\right]\mathbf\Omega^{-1}\\
        &=\alpha \mathbf \Omega^{-1} -\frac{1}{2}(\mathbf Y^T\mathbf Y+\mathbf \Phi) + \frac{1}{2}\mathbf \Omega^{-1}\left[\mathbf \Psi + \mathbf B_0^T(\mathbf \Lambda^{-1}-\mathbf \Lambda^{-1}(\mathbf X^T\mathbf X+\mathbf\Lambda^{-1})^{-1}\mathbf\Lambda^{-1})\mathbf B_0\right]\mathbf\Omega^{-1}
    \end{aligned}
\end{equation*}

Again, by setting the derivative to 0 and multiplying by $\mathbf \Omega$ on both sides (left and right), we get the equation:
$-2\alpha \mathbf \Omega+ \mathbf \Omega(\mathbf Y^T \mathbf Y+\mathbf \Phi)\mathbf \Omega-\left[\mathbf \Psi + \mathbf B_0^T(\mathbf \Lambda^{-1}-\mathbf \Lambda^{-1}(\mathbf X^T\mathbf X+\mathbf \Lambda^{-1})^{-1}\mathbf \Lambda^{-1})\mathbf B_0\right]=0$
which is again a special case of continuous time algebraic Riccati equation (CARE) \citep{boyd1991linear,anderson2007optimal}. Since $\mathbf Y^T\mathbf Y+\mathbf\Phi$ and $[\mathbf\Psi + \mathbf B_0^T(\mathbf\Lambda^{-1}-\mathbf\Lambda^{-1}(\mathbf X^T\mathbf X+\mathbf\Lambda^{-1})^{-1}\mathbf\Lambda^{-1})\mathbf B_0]$ are positive definite (see Proposition \ref{lemma:positive} in the Appendix), then following the proof of \citet{fazayeli2016matrix}, we get that the posterior is also unimodal. 
\end{proof}

Lastly, unlike the Normal-Wishart conjugate prior, the Normal-MGIG does reach its limiting case of known $\mathbf B$ when the uncertainty of $\mathbf B$ goes to zero as highlighted in Remark \ref{NMGIG-limit}.

\begin{myRemark}
\label{NMGIG-limit}
Since $\mathbf \Lambda$ represents the uncertainty on $\mathbf B$, when we take the limit of $\mathbf \Lambda\to \mathbf 0$, $\mathbf B$ is fully known (as $\mathbf{B}_0$) in the Normal-MGIG prior and it reduces to the known $\mathbf B$ case.
\end{myRemark}

\section{Positive definiteness of $\hat{\mathbf \Phi}$ in the Normal-Wishart prior and of $\hat{\mathbf \Psi}$ in the Normal-MGIG prior}
\label{sec:posdef}

\begin{myProposition}
\label{lemma:positive}
Let $\hat{\mathbf \Psi}=\mathbf \Psi+\mathbf{B}_0^T\mathbf{\Lambda}^{-1}\mathbf{B}_0-\mathbf{B}_0^T\mathbf{\Lambda}^{-1}(\mathbf{X}^T\mathbf{X}+\mathbf{\Lambda}^{-1})^{-1}\mathbf{\Lambda}^{-1}\mathbf{B}_0$ and $\hat{\mathbf \Phi}=\mathbf \Phi+\mathbf{Y}^T\mathbf{Y}-\mathbf{Y}^T\mathbf{X}(\mathbf{X}^T\mathbf{X}+\mathbf{\Lambda}^{-1})^{-1}\mathbf{X}^T\mathbf{Y}$. Let $\mathbf \Lambda$, $\mathbf \Psi$ and $\mathbf \Phi$ be positive definite. If $\mathbf X= \mathbf 0$ or if $\mathbf X^T \mathbf X$ is invertible, then both $\hat{\mathbf \Psi}$ and $\hat{\mathbf \Phi}$ are positive definite. 
\end{myProposition}

\begin{proof}
For $\mathbf X=\mathbf 0$, the matrices are reduced to $\hat{\mathbf \Psi}=\mathbf \Psi$ and $\hat{\mathbf \Phi}=\mathbf \Phi+\mathbf{Y}^T\mathbf{Y}$, and thus are trivially positive definite.

For $\mathbf X^T \mathbf X$ invertible we can write the last two terms of $\hat{\mathbf{\Psi}}$ as $\mathbf{B}^T_0\mathbf{\Lambda^{-1}}(\mathbf{\Lambda}-(\mathbf{X}^T\mathbf{X}+\mathbf{\Lambda}^{-1})^{-1})\mathbf{\Lambda}^{-1}\mathbf{B}_0$ thus it suffices to show $\mathbf{\Lambda}-(\mathbf{X}^T\mathbf{X}+\mathbf{\Lambda}^{-1})^{-1}$ is positive definite. We can also write the last two terms of $\hat{\mathbf{\Phi}}$ as $\mathbf{Y}^T(\mathbf{I}_n-\mathbf{X}(\mathbf{X}^T\mathbf{X}+\mathbf{\Lambda}^{-1})^{-1}\mathbf{X}^T)\mathbf{Y}$ thus it suffices to show $\mathbf{I}_n-\mathbf{X}(\mathbf{X}^T\mathbf{X}+\mathbf{\Lambda}^{-1})^{-1}\mathbf{X}^T$ is positive definite. We use the fact that $\mathbf{\Lambda}^{-1}$ is symmetric. 

Assuming proper invertibility, we have 
\begin{equation}
(\mathbf{UPV} + \mathbf A)^{-1}=\mathbf A^{-1}-\mathbf A^{-1}\mathbf{UP}(\mathbf I + \mathbf{VA}^{-1}\mathbf{UP})^{-1}\mathbf{VA}^{-1}
\label{eqn:inverse_for_positive}
\end{equation}

(Equation 24 in \citet{henderson1981deriving}) which allows us to decompose $\mathbf{\Lambda}-(\mathbf{X}^T\mathbf{X}+\mathbf{\Lambda}^{-1})^{-1}$ by taking $\mathbf A=\mathbf \Lambda^{-1}$, $\mathbf U= \mathbf X^T$, $\mathbf V=\mathbf X$ and $\mathbf P=\mathbf I_n$
\begin{align*}
\mathbf{\Lambda}-(\mathbf{X}^T\mathbf{X}+\mathbf{\Lambda}^{-1})^{-1}&=\mathbf \Lambda-\mathbf \Lambda+\mathbf \Lambda \mathbf X^T(\mathbf I_n+\mathbf X \mathbf \Lambda \mathbf X^T)^{-1} \mathbf X \mathbf \Lambda=\mathbf \Lambda \mathbf X^T(\mathbf I_n+\mathbf X \mathbf \Lambda \mathbf X^T)^{-1} \mathbf X \mathbf \Lambda\\
&= (\mathbf X\mathbf \Lambda)^T(\mathbf I_n+\mathbf X \mathbf \Lambda \mathbf X^T)^{-1} \mathbf X \mathbf \Lambda
\end{align*}
which is positive definite because it is a quadratic form of a positive definite matrix $\mathbf I_n+\mathbf X \mathbf \Lambda \mathbf X^T$.

For $\mathbf I_n-\mathbf{X}(\mathbf{X}^T\mathbf{X}+\mathbf{\Lambda}^{-1})^{-1}\mathbf{X}^T$, we take the Cholesky decomposition of $\mathbf \Lambda^{-1}=\mathbf{LL}^T$ and we take $\mathbf A = \mathbf X ^T \mathbf X$, $\mathbf U=\mathbf L$, $\mathbf P= \mathbf I_p$ and $\mathbf V=\mathbf L^T$.
\begin{align*}
(\mathbf X ^T \mathbf X + \mathbf \Lambda ^{-1})^{-1} &= (\mathbf X^T \mathbf X)^{-1} - (\mathbf X^T \mathbf X)^{-1} \mathbf L (\mathbf I_p + \mathbf L^T (\mathbf X ^T \mathbf X)^{-1} \mathbf L)^{-1} \mathbf L^T(\mathbf X ^T \mathbf X)^{-1} \\
\mathbf I_p-\mathbf{X}(\mathbf{X}^T\mathbf{X}+\mathbf{\Lambda}^{-1})^{-1}\mathbf{X}^T&= \mathbf I_p-\mathbf{X}(\mathbf{X}^T\mathbf{X})^{-1}\mathbf{X}^T+\mathbf{X}[(\mathbf{X}^T\mathbf{X})^{-1} \mathbf L(\mathbf I_p+ \mathbf L^T(\mathbf{X}^T\mathbf{X})^{-1} \mathbf L)^{-1} \mathbf L^T(\mathbf{X}^T\mathbf{X})^{-1}]\mathbf{X}^T.
\end{align*}

The last term is positive definite because it is a quadratic form of a positive definite matrix $\mathbf I_p+ \mathbf L^T(\mathbf{X}^T\mathbf{X})^{-1} \mathbf L$. It remains to show that $\mathbf I_p-\mathbf{X}(\mathbf{X}^T\mathbf{X})^{-1}\mathbf{X}^T$ is positive definite. We observe that $\mathbf{X}(\mathbf{X}^T\mathbf{X})^{-1}\mathbf{X}^T$ is a projection matrix and denote it as $\mathbf Q$. We then have $\mathbf Q^T\mathbf Q=\mathbf Q$, and thus $\mathbf I_p-\mathbf Q=\mathbf I_p-\mathbf Q^T \mathbf Q$. For any vector $a$, the quadratic form is $a^T(\mathbf I_p-\mathbf Q)a=a^Ta-(\mathbf Q a)^T(\mathbf Q a)=||a||^2-||\mathbf Q a||^2$. Since $\mathbf Q$ is a projection to a subspace, we have that $||\mathbf Qa||\le ||a||$, thus we have that $\mathbf I_p-\mathbf Q$ is positive definite. 
\end{proof}

\section{Laplace approximation of the marginal posterior precision under the conjugate prior}
\label{sec:approx-conj}

This section provides the algebraic details of the simplification in  \eqref{eqn:Wishart_postprecision}. We simplify the second term in  \eqref{eqn:Wishart_postprecision} and it can be seen that it is almost the same as the first term, with a difference of $\left(\frac{n}{2}+\alpha \right)\mathbf D_k^T\left[\mathbf \Omega^{-1}\otimes  \mathbf \Omega^{-1}\right]\mathbf D_k$
\begin{equation*}
    \begin{aligned}
    &\mathbf{D}_k^T\left(\mathbf{\Omega}^{-1}\otimes (\mathbf{\Omega}^{-1}\mathbf{B}^T(\mathbf{X}^T\mathbf{X}+\mathbf \Lambda^{-1})) \right)\left[\mathbf{\Omega}\otimes(\mathbf{X}^T \mathbf{X}+\mathbf \Lambda^{-1})^{-1}\right]\left[(\mathbf{D}_k^T\left(\mathbf{\Omega}^{-1}\otimes \mathbf{\Omega}^{-1}\mathbf{B}^T(\mathbf{X}^T\mathbf{X}+\mathbf \Lambda^{-1}))\right) \right]^T\\
    =&\mathbf{D}_k^T\left(\mathbf{\Omega}^{-1}\otimes (\mathbf{\Omega}^{-1}\mathbf{B}^T(\mathbf{X}^T\mathbf{X}+\mathbf \Lambda^{-1})) \right)\left[\mathbf{\Omega}\otimes(\mathbf{X}^T \mathbf{X}+\mathbf \Lambda^{-1})^{-1}\right]\left[(\left(\mathbf{\Omega}^{-1}\otimes (\mathbf{X}^T\mathbf{X}+\mathbf \Lambda^{-1})\mathbf{B}\mathbf{\Omega}^{-1})\right) \right]^T\mathbf{D}_k\\
    =&\mathbf{D}_k^T\left(\mathbf{\Omega}^{-1}\otimes (\mathbf{\Omega}^{-1}\mathbf{B}^T(\mathbf{X}^T\mathbf{X}+\mathbf \Lambda^{-1})) \right)\left[\mathbf{I}\otimes\mathbf{B}\mathbf{\Omega}^{-1}\right]\mathbf{D}_k\\
    =&\mathbf{D}_k^T\left(\mathbf{\Omega}^{-1}\otimes (\mathbf{\Omega}^{-1}\mathbf{B}^T(\mathbf{X}^T\mathbf{X}+\mathbf \Lambda^{-1})\mathbf{B}\mathbf{\Omega}^{-1}) \right)\mathbf D_k\\
    =&\mathbf{D}_k^T\left(\mathbf{\Omega}^{-1}\otimes (\mathbf{\Omega}^{-1}\mathbf{B}^T\mathbf{X}^T\mathbf{X}\mathbf{B}\mathbf{\Omega}^{-1}+\mathbf{\Omega}^{-1}\mathbf{B}^T\mathbf \Lambda^{-1}\mathbf{B}\mathbf{\Omega}^{-1}) \right)\mathbf D_k\\
    \end{aligned}
\end{equation*}

\section{Laplace approximation of the marginal posterior precision under the Normal-MGIG prior}
\label{appNMGIG}

This section provides the algebraic details of the simplification in  \eqref{eqn:MGIG_postprecision}. We focus on the second term to convert to a form that is similar to the first term: 
\begin{equation}
\label{eqn:alge_MGIG}
\begin{aligned}
&\mathbf{D}_k^T\left(\mathbf{\Omega}^{-1}\otimes (\mathbf{\Omega}^{-1}(\mathbf{B}^T\mathbf{X}^T\mathbf{X}+(\mathbf B-\mathbf B_0)^T\mathbf \Lambda^{-1})) \right)\left[\mathbf{\Omega}\otimes(\mathbf{X}^T \mathbf{X}+\mathbf \Lambda^{-1})^{-1}\right]\\
&\left[\mathbf{D}_k^T\left(\mathbf{\Omega}^{-1}\otimes (\mathbf{\Omega}^{-1}(\mathbf{B}^T\mathbf{X}^T\mathbf{X}+(\mathbf B-\mathbf B_0)^T\mathbf \Lambda^{-1}))\right) \right]^T
\end{aligned}
\end{equation}

The transpose on the term below can be simplified as follows:
\begin{equation*}
\begin{aligned}
&\left[\mathbf{D}_k^T\left(\mathbf{\Omega}^{-1}\otimes (\mathbf{\Omega}^{-1}(\mathbf{B}^T\mathbf{X}^T\mathbf{X}+(\mathbf B-\mathbf B_0)^T\mathbf \Lambda^{-1}))\right) \right]^T = \left(\mathbf{\Omega}^{-1}\otimes (\mathbf{\Omega}^{-1}(\mathbf{B}^T\mathbf{X}^T\mathbf{X}+(\mathbf B-\mathbf B_0)^T\mathbf \Lambda^{-1}))\right)^T \mathbf D_k\\
&= \left(\mathbf{\Omega}^{-1}\otimes (\mathbf{\Omega}^{-1}(\mathbf{B}^T\mathbf{X}^T\mathbf{X}+(\mathbf B-\mathbf B_0)^T\mathbf \Lambda^{-1}))^T \right) \mathbf D_k = \left(\mathbf{\Omega}^{-1}\otimes (\mathbf{B}^T\mathbf{X}^T\mathbf{X}+(\mathbf B-\mathbf B_0)^T\mathbf \Lambda^{-1})^T \mathbf{\Omega}^{-1} \right) \mathbf D_k \\
&= \left(\mathbf{\Omega}^{-1}\otimes (\mathbf{X}^T\mathbf{X}\mathbf{B}+\mathbf \Lambda^{-1}(\mathbf B-\mathbf B_0)) \mathbf{\Omega}^{-1} \right) \mathbf D_k.
\end{aligned}
\end{equation*}
which combined with the other terms in  \eqref{eqn:alge_MGIG} becomes
\begin{equation}
\label{eqn:alge_MGIG3}
\begin{aligned}
&\mathbf{D}_k^T\left(\mathbf{\Omega}^{-1}\otimes (\mathbf{\Omega}^{-1}(\mathbf{B}^T\mathbf{X}^T\mathbf{X}+(\mathbf B-\mathbf B_0)^T\mathbf \Lambda^{-1})) \right)\left[\mathbf{\Omega}\otimes(\mathbf{X}^T \mathbf{X}+\mathbf \Lambda^{-1})^{-1}\right]\\
&\left(\mathbf{\Omega}^{-1}\otimes (\mathbf{X}^T\mathbf{X}\mathbf{B}+\mathbf \Lambda^{-1}(\mathbf B-\mathbf B_0)) \mathbf{\Omega}^{-1} \right) \mathbf D_k
\end{aligned}
\end{equation}

By the mixed-product property of Kronecker product $(\mathbf{M}_1 \otimes \mathbf{M}_2)(\mathbf{M}_3 \otimes \mathbf{M}_4) = (\mathbf{M}_1\mathbf{M}_3) \otimes (\mathbf{M}_2 \mathbf{M}_4)$, we multiply the two last terms in  \eqref{eqn:alge_MGIG3} and obtain the first line in  \eqref{eqn:alge_MGIG4}. We then repeat the same mixed-product property to go from the first equation to the second equation and re-group terms: 
\begin{equation}
\label{eqn:alge_MGIG4}
\begin{aligned}
&\mathbf D_k^T[\left(\mathbf{\Omega}^{-1}\otimes (\mathbf{\Omega}^{-1}(\mathbf{B}^T\mathbf{X}^T\mathbf{X}+(\mathbf B-\mathbf B_0)^T\mathbf \Lambda^{-1})) \right)
\left[ \mathbf I\otimes (\mathbf{X}^T \mathbf{X}+\mathbf \Lambda^{-1})^{-1}(\mathbf X^T\mathbf X\mathbf B +\mathbf \Lambda^{-1}(\mathbf B-\mathbf B_0))\mathbf \Omega^{-1} \right]]\mathbf D_k\\
=&\mathbf D_k^T[\mathbf{\Omega}^{-1}\otimes \left((\mathbf{\Omega}^{-1}(\mathbf{B}^T\mathbf{X}^T\mathbf{X}+(\mathbf B-\mathbf B_0)^T\mathbf \Lambda^{-1}))((\mathbf{X}^T \mathbf{X}+\mathbf \Lambda^{-1})^{-1}(\mathbf X^T\mathbf X\mathbf B +\mathbf \Lambda^{-1}(\mathbf B-\mathbf B_0))\mathbf \Omega^{-1} )\right)]\mathbf D_k \\
=&\mathbf D_k^T[\mathbf{\Omega}^{-1}\otimes \left((\mathbf{\Omega}^{-1}(\mathbf{B}^T(\mathbf{X}^T\mathbf{X}+\mathbf \Lambda^{-1})-\mathbf B_0^T\mathbf \Lambda^{-1}))((\mathbf{X}^T \mathbf{X}+\mathbf \Lambda^{-1})^{-1}((\mathbf X^T\mathbf X+\mathbf \Lambda^{-1})\mathbf B -\mathbf \Lambda^{-1}\mathbf B_0)\mathbf \Omega^{-1} ) \right)]\mathbf D_k \\
=&\mathbf D_k^T[\mathbf{\Omega}^{-1}\otimes \left(\mathbf{\Omega}^{-1}(\mathbf{B}^T(\mathbf{X}^T\mathbf{X}+\mathbf \Lambda^{-1})-\mathbf B_0^T\mathbf \Lambda^{-1})(\mathbf{X}^T \mathbf{X}+\mathbf \Lambda^{-1})^{-1}((\mathbf X^T\mathbf X+\mathbf \Lambda^{-1})\mathbf B -\mathbf \Lambda^{-1}\mathbf B_0)\mathbf \Omega^{-1}  \right)]\mathbf D_k 
\end{aligned}
\end{equation}

The second factor in the Kronecker product in the middle is a quadratic form on $((\mathbf X^T\mathbf X+\mathbf \Lambda^{-1})\mathbf B -\mathbf \Lambda^{-1}\mathbf B_0)$, so we can expand it and re-group terms:
\begin{equation*}
\begin{aligned}
\mathbf D_k^T&[\mathbf{\Omega}^{-1}\otimes \left( \mathbf{\Omega}^{-1}\left[\mathbf B^T(\mathbf X^T\mathbf X+\mathbf \Lambda^{-1})\mathbf B-\mathbf B^T\mathbf \Lambda^{-1}\mathbf B_0-\mathbf B_0^T\mathbf \Lambda^{-1}\mathbf B+\mathbf B_0^T\mathbf \Lambda^{-1}(\mathbf X^T\mathbf X+\mathbf \Lambda^{-1})^{-1}\mathbf \Lambda^{-1}\mathbf B_0\right]\mathbf \Omega^{-1} \right)] \mathbf D_k \\
=&\mathbf D_k^T[\left(\mathbf{\Omega}^{-1}\otimes \mathbf{\Omega}^{-1}\left[\mathbf B^T\mathbf X^T\mathbf X\mathbf B+(\mathbf B-\mathbf B_0)^T\mathbf \Lambda^{-1}(\mathbf B-\mathbf B_0)^T+\mathbf B_0^T\mathbf \Lambda^{-1}(\mathbf X^T\mathbf X+\mathbf \Lambda^{-1})^{-1}\mathbf \Lambda^{-1}\mathbf B_0\right]\mathbf \Omega^{-1} \right)]\mathbf D_k \\
=&\mathbf D_k^T[\left(\mathbf{\Omega}^{-1}\otimes \mathbf{\Omega}^{-1}\left[\mathbf B^T\mathbf X^T\mathbf X\mathbf B+(\mathbf B-\mathbf B_0)^T\mathbf \Lambda^{-1}(\mathbf B-\mathbf B_0)^T\right]\mathbf \Omega^{-1} \right)]\mathbf D_k\\
&+\mathbf D_k^T[\left(\mathbf{\Omega}^{-1}\otimes \mathbf{\Omega}^{-1}\left[\mathbf B_0^T\mathbf \Lambda^{-1}(\mathbf X^T\mathbf X+\mathbf \Lambda^{-1})^{-1}\mathbf \Lambda^{-1}\mathbf B_0\right]\mathbf \Omega^{-1} \right)]\mathbf D_k
\end{aligned}
\end{equation*}

We can now look at the first term in  \eqref{eqn:MGIG_postprecision} and note that the part 
$$\mathbf D_k^T[\left(\mathbf{\Omega}^{-1}\otimes \mathbf{\Omega}^{-1}[\mathbf B^T\mathbf X^T\mathbf X\mathbf B+(\mathbf B-\mathbf B_0)^T\mathbf \Lambda^{-1}(\mathbf B-\mathbf B_0)^T]\mathbf \Omega^{-1} \right)]\mathbf D_k$$
cancels out, so the remaining of  \eqref{eqn:MGIG_postprecision} becomes:
\begin{align*}
    &\mathbf D_k^T\left[\mathbf \Omega^{-1}\otimes\left( \mathbf \Omega^{-1} \mathbf B_0^T\mathbf \Lambda^{-1}\mathbf B_0 \mathbf \Omega^{-1} + \mathbf \Omega^{-1}\mathbf \Psi\mathbf \Omega^{-1}-\mathbf \Omega^{-1} \mathbf B_0^T( \mathbf \Lambda^{-1}(\mathbf X^T\mathbf X+\mathbf \Lambda^{-1})^{-1}\mathbf \Lambda^{-1})\mathbf B_0 \mathbf \Omega^{-1}  \right) \right]\mathbf D_k\\
    &+\mathbf D_k^T\left[\mathbf \Omega^{-1}\otimes (\frac{n}{2}+\alpha)\mathbf \Omega^{-1}\right]\\
    =&\mathbf D_k^T\left[\mathbf \Omega^{-1}\otimes\left( (\frac{n}{2}+\alpha)\mathbf \Omega^{-1}+\mathbf \Omega^{-1} \mathbf B_0^T(\mathbf \Lambda^{-1} - \mathbf \Lambda^{-1}(\mathbf X^T\mathbf X+\mathbf \Lambda^{-1})^{-1}\mathbf \Lambda^{-1})\mathbf B_0 \mathbf \Omega^{-1} + \mathbf \Omega^{-1}\mathbf \Psi\mathbf \Omega^{-1} \right) \right]\mathbf D_k
\end{align*}

\section{Laplace approximation of the marginal posterior precision under the general independent prior}
\label{appGeneral}


Here, we show the bound inequality in  \eqref{eqn:bound_indep}: $\mathbf{EF}^{-1} \mathbf L(\mathbf I_{kp}+\mathbf L^T\mathbf F^{-1}\mathbf L)^{-1}\mathbf L^T(\mathbf{EF}^{-1})^{T}\le \mathbf{EF}^{-1}\mathbf\Lambda^{-1} (\mathbf{EF}^{-1})^T$.
%
Consider the difference $\mathbf{EF}^{-1}\mathbf\Lambda^{-1} (\mathbf{EF}^{-1})^T-
\mathbf{EF}^{-1} \mathbf L(\mathbf I_{kp}+\mathbf L^T\mathbf F^{-1}\mathbf L)^{-1}\mathbf L^T(\mathbf{EF}^{-1})^{T}$ which is equal to 

$\mathbf{EF}^{-1}\left(\mathbf\Lambda^{-1}-\mathbf L(\mathbf I_{kp}+\mathbf L^T\mathbf F^{-1}\mathbf L)^{-1}\mathbf L^T\right) (\mathbf{EF}^{-1})^T$.
 
Given that this is a quadratic form, it suffices to show that $\mathbf\Lambda^{-1}-\mathbf L(\mathbf I_{kp}+\mathbf L^T\mathbf F^{-1}\mathbf L)^{-1}\mathbf L^T$ is positive semi-definite. Again by  \eqref{eqn:inverse_for_positive} \citep{henderson1981deriving}, since $\mathbf L^T\mathbf F^{-1}\mathbf L$ is positive definite, we take the Cholesky decomposition of this matrix as $\mathbf Q\mathbf Q^T$. Recall that $\mathbf L \mathbf L^T=\mathbf \Lambda^{-1}$ by definition, thus 
\begin{equation*}
    \begin{aligned}
    \mathbf\Lambda^{-1}-\mathbf L(\mathbf I_{kp}+\mathbf L^T\mathbf F^{-1}\mathbf L)^{-1}\mathbf L^T&=\mathbf L \mathbf L^T-\mathbf L(\mathbf I_{kp}+\mathbf L^T\mathbf F^{-1}\mathbf L)^{-1}\mathbf L^T =\mathbf L(\mathbf I^{-1}_{kp}-(\mathbf I_{kp}+\mathbf L^T\mathbf F^{-1}\mathbf L)^{-1})\mathbf L^T\\
    &=\mathbf L(\mathbf I^{-1}_{kp}-(\mathbf I_{kp}+\mathbf Q\mathbf Q^T)^{-1})\mathbf L^T =\mathbf L\mathbf Q(\mathbf I_{kp}+\mathbf Q^T\mathbf Q)^{-1}\mathbf Q^T\mathbf L^T
    \end{aligned}
\end{equation*}
which is positive (semi-)definite since it is a quadratic form of a positive definite matrix. The last equality is per  \eqref{eqn:inverse_for_positive}, taking $\mathbf{A}=\mathbf I_{kp}$ and $\mathbf{U}=\mathbf{V}^T=\mathbf Q^T$.


\section{\revision{Simulations on other covariance structures}}
\label{app:other_graphs}

\revision{In this Appendix, we provide the results on the simulation of the remaining five covariance structures not shown in the main text.}

\subsection*{KL divergence compared to null design}
\begin{figure}[H]
    \centering
    \includegraphics[width = 0.8\linewidth]{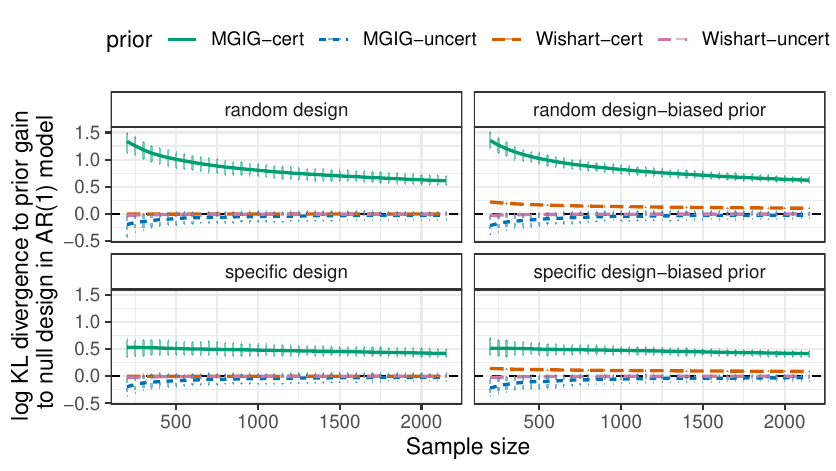}
    \caption{Difference in log KL divergence between prior and posterior comparing random experiment ($\mathbf X \ne \mathbf 0$) and specific experiment (diagonal $\mathbf X$) vs null experiment ($\mathbf X=\mathbf 0$) under AR1 models with 50 responses and 50 predictors, with and without biases on the prior of $\mathbf{B}$. Lines are averages over 100 repeats while error bars are 0.975 and 0.025 quantiles. }
    \label{fig:ar1_kl}
\end{figure}

\begin{figure}[H]
    \centering
    \includegraphics[width = 0.8\linewidth]{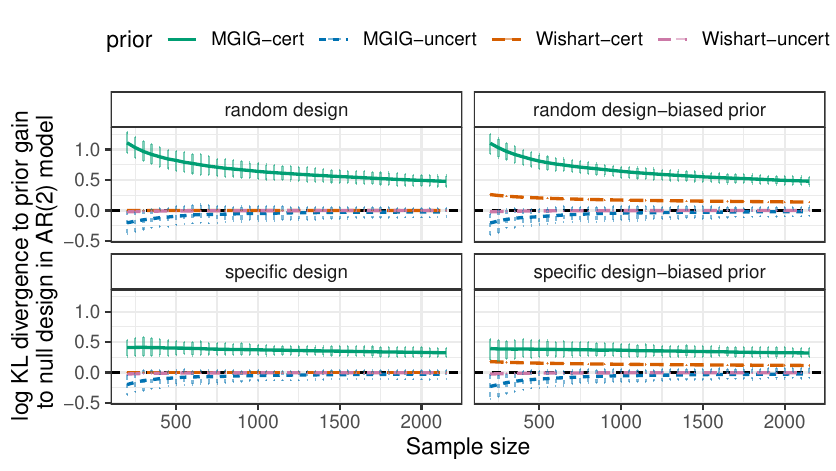}
    \caption{Difference in log KL divergence between prior and posterior comparing random experiment ($\mathbf X \ne \mathbf 0$) and specific experiment (diagonal $\mathbf X$) vs null experiment ($\mathbf X=\mathbf 0$) under AR2 models with 50 responses and 50 predictors, with and without biases on the prior of $\mathbf{B}$. Lines are averages over 100 repeats while error bars are 0.975 and 0.025 quantiles. }
    \label{fig:ar2_kl}
\end{figure}

\begin{figure}[H]
    \centering
    \includegraphics[width = 0.8\linewidth]{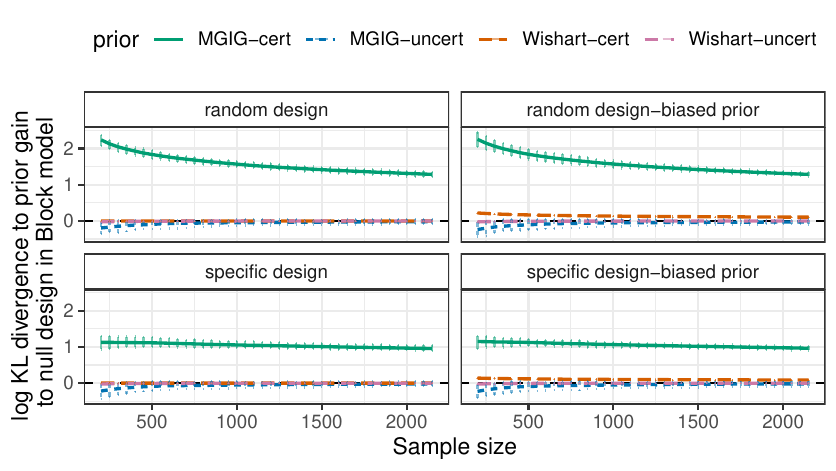}
    \caption{Difference in log KL divergence between prior and posterior comparing random experiment ($\mathbf X \ne \mathbf 0$) and specific experiment (diagonal $\mathbf X$) vs null experiment ($\mathbf X=\mathbf 0$) under block models with 50 responses and 50 predictors, with and without biases on the prior of $\mathbf{B}$. Lines are averages over 100 repeats while error bars are 0.975 and 0.025 quantiles. }
    \label{fig:block_kl}
\end{figure}

\begin{figure}[H]
    \centering
    \includegraphics[width = 0.8\linewidth]{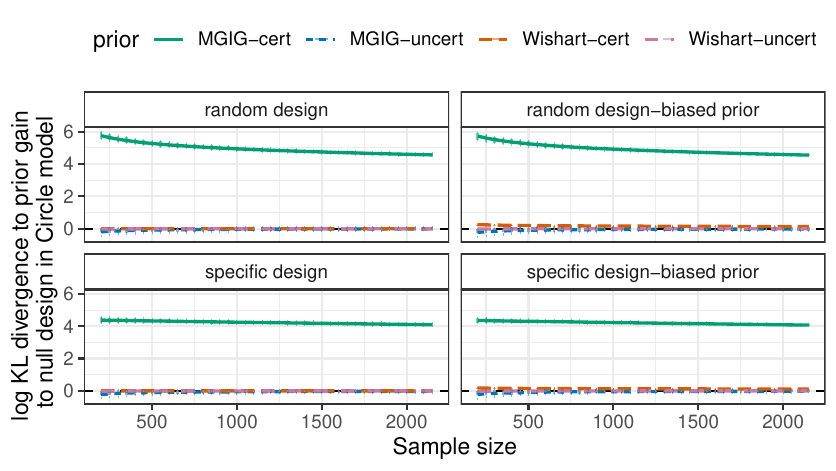}
    \caption{Difference in log KL divergence between prior and posterior comparing random experiment ($\mathbf X \ne \mathbf 0$) and specific experiment (diagonal $\mathbf X$) vs null experiment ($\mathbf X=\mathbf 0$) under circle models with 50 responses and 50 predictors, with and without biases on the prior of $\mathbf{B}$. Lines are averages over 100 repeats while error bars are 0.975 and 0.025 quantiles. }
    \label{fig:circle_kl}
\end{figure}

\begin{figure}[H]
    \centering
    \includegraphics[width = 0.8\linewidth]{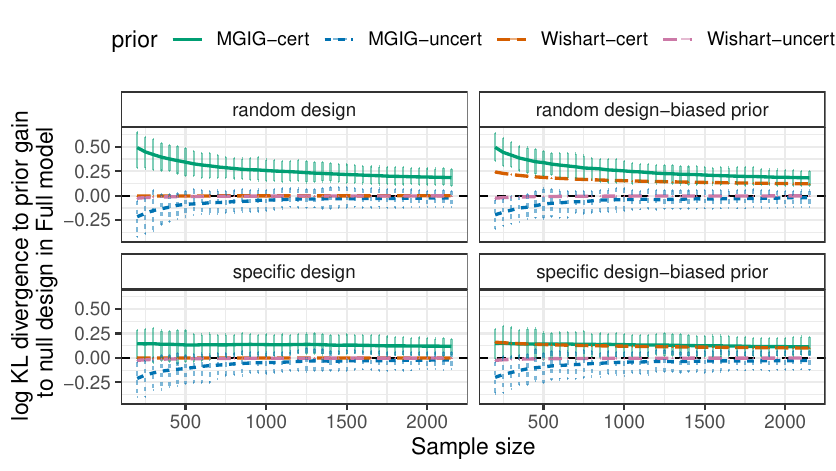}
    \caption{Difference in log KL divergence between prior and posterior comparing random experiment ($\mathbf X \ne \mathbf 0$) and specific experiment (diagonal $\mathbf X$) vs null experiment ($\mathbf X=\mathbf 0$) under Full models with 50 responses and 50 predictors, with and without biases on the prior of $\mathbf{B}$. Lines are averages over 100 repeats while error bars are 0.975 and 0.025 quantiles. }
    \label{fig:full_kl}
\end{figure}

\subsection*{Stein's loss compared to null design}
\revision{This set of experiments suggests that a good prior on $\mathbf{B}$ makes experiments more helpful to estimate $\Omega$, but with biased prior, this is not the case as positive Stein's loss change means a worse point estimate compared to null design.} 

\begin{figure}[H]
    \centering
    \includegraphics[width = 0.8\linewidth]{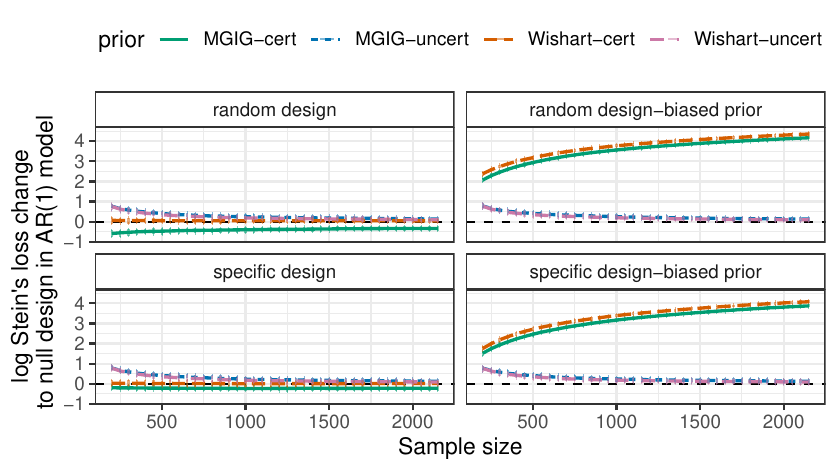}
    \caption{Difference in log Stein's loss of random experiment ($\mathbf X \ne \mathbf 0$) and specific experiment (diagonal $\mathbf X$) vs null experiment ($\mathbf X=\mathbf 0$) under AR1 models with 50 responses and 50 predictors, with and without biases on the prior of $\mathbf{B}$. Lines are averages over 100 repeats while error bars are 0.975 and 0.025 quantiles. }
    \label{fig:ar1stein}
\end{figure}

\begin{figure}[H]
    \centering
    \includegraphics[width = 0.8\linewidth]{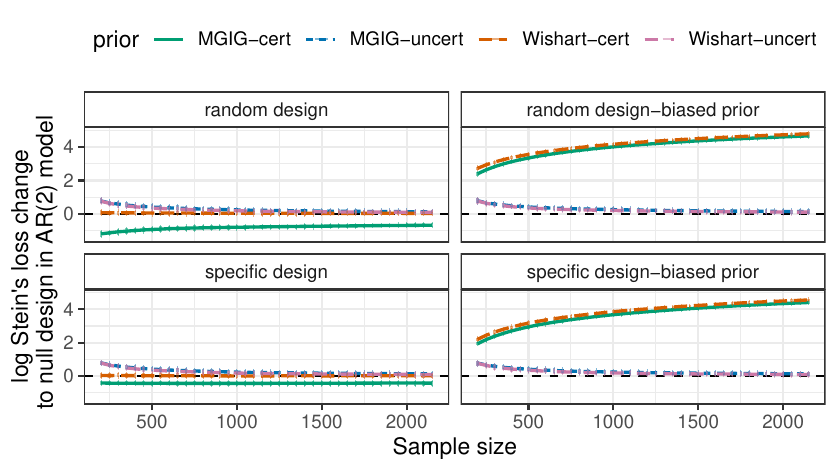}
    \caption{Difference in log Stein's loss of random experiment ($\mathbf X \ne \mathbf 0$) and specific experiment (diagonal $\mathbf X$) vs null experiment ($\mathbf X=\mathbf 0$) under AR2 models with 50 responses and 50 predictors, with and without biases on the prior of $\mathbf{B}$. Lines are averages over 100 repeats while error bars are 0.975 and 0.025 quantiles. }
    \label{fig:ar2stein}
\end{figure}

\begin{figure}[H]
    \centering
    \includegraphics[width = 0.8\linewidth]{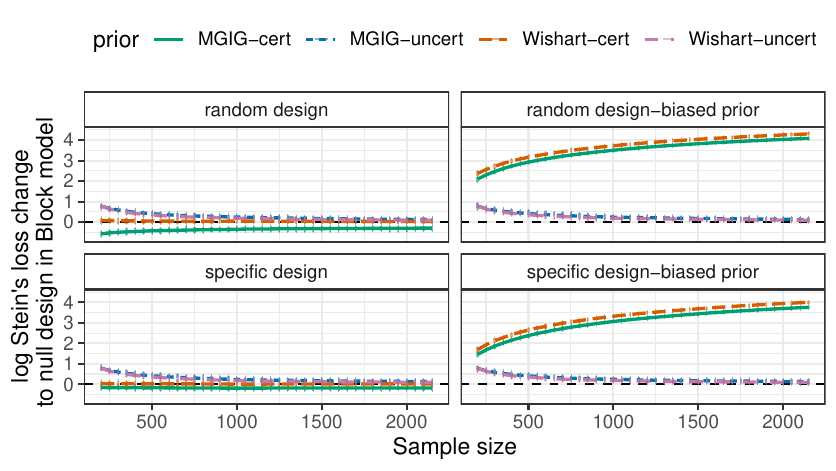}
    \caption{Difference in log Stein's loss of random experiment ($\mathbf X \ne \mathbf 0$) and specific experiment (diagonal $\mathbf X$) vs null experiment ($\mathbf X=\mathbf 0$) under Block models with 50 responses and 50 predictors, with and without biases on the prior of $\mathbf{B}$. Lines are averages over 100 repeats while error bars are 0.975 and 0.025 quantiles. }
    \label{fig:Blockstein}
\end{figure}

\begin{figure}[H]
    \centering
    \includegraphics[width = 0.8\linewidth]{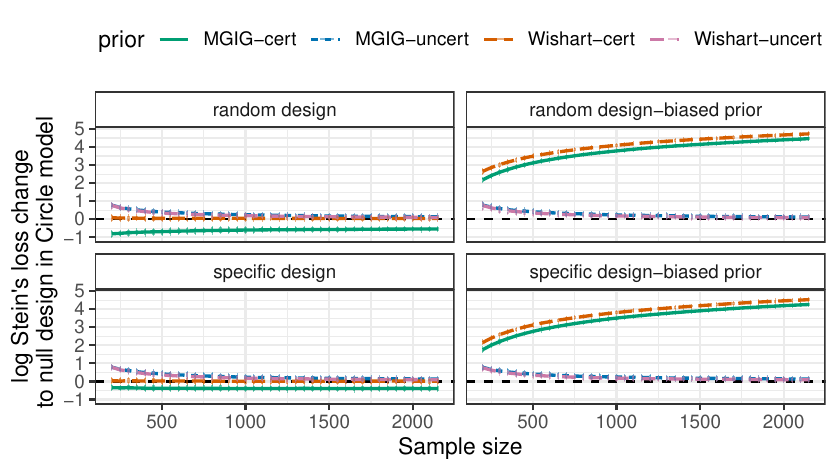}
    \caption{Difference in log Stein's loss of random experiment ($\mathbf X \ne \mathbf 0$) and specific experiment (diagonal $\mathbf X$) vs null experiment ($\mathbf X=\mathbf 0$) under Circle models with 50 responses and 50 predictors, with and without biases on the prior of $\mathbf{B}$. Lines are averages over 100 repeats while error bars are 0.975 and 0.025 quantiles. }
    \label{fig:circlestein}
\end{figure}

\begin{figure}[H]
    \centering
    \includegraphics[width = 0.8\linewidth]{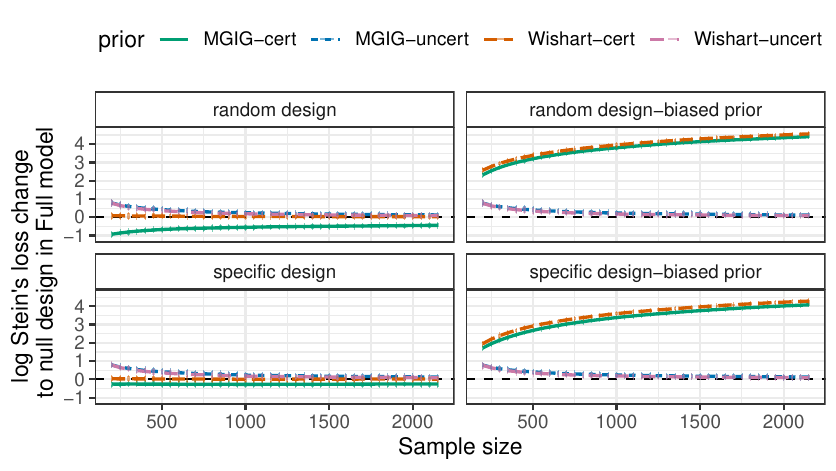}
    \caption{Difference in log Stein's loss of random experiment ($\mathbf X \ne \mathbf 0$) and specific experiment (diagonal $\mathbf X$) vs null experiment ($\mathbf X=\mathbf 0$) under Full models with 50 responses and 50 predictors, with and without biases on the prior of $\mathbf{B}$. Lines are averages over 100 repeats while error bars are 0.975 and 0.025 quantiles. }
    \label{fig:fullstein}
\end{figure}

\section{\revision{Simulations with smaller bias in the prior for $\mathbf{B}$}}
\label{app:less_biased}


\revision{Here, we explore how the degree of bias in the prior of $\mathbf B$ (the effect of experiments) can influence the inference of the network $\mathbf \Omega$. We choose the same setting as the simulations in the main text (Section \ref{sec:simulation}), except that here the prior mean of $\mathbf B$ is $\mathbf B_0= \mathbf B+\epsilon$ where $\epsilon\sim N(0,0.1)$. In general, we observe the same results as without bias in the prior and having certain prior on $\mathbf{B}$ makes experiment useful in inferring $\mathbf{\Omega}$.}

\subsection{KL divergence compared to null experiment}
\begin{figure}[H]
    \centering
    \includegraphics[width = 0.8\linewidth]{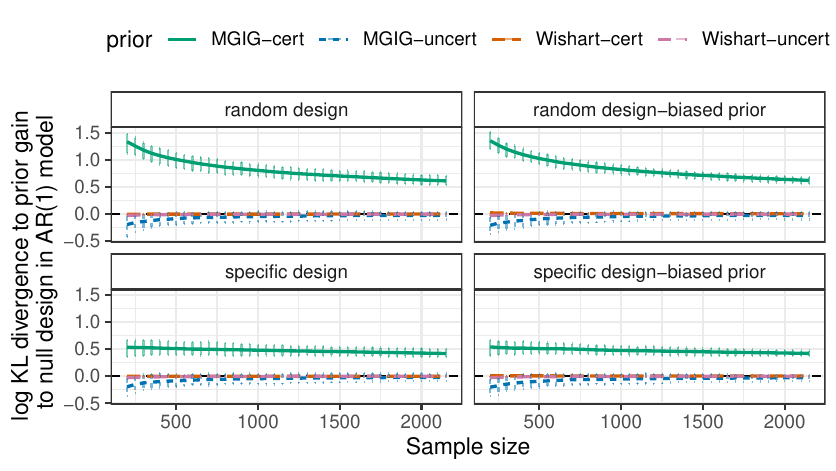}
    \caption{Difference in log KL divergence between prior and posterior comparing random experiment ($\mathbf X \ne \mathbf 0$) and specific experiment (diagonal $\mathbf X$) vs null experiment ($\mathbf X=\mathbf 0$) under AR1 models with 50 responses and 50 predictors, with and without biases on the prior of $\mathbf{B}$. Lines are averages over 100 repeats while error bars are 0.975 and 0.025 quantiles. }
    \label{fig:ar1_kl2}
\end{figure}

\begin{figure}[H]
    \centering
    \includegraphics[width = 0.8\linewidth]{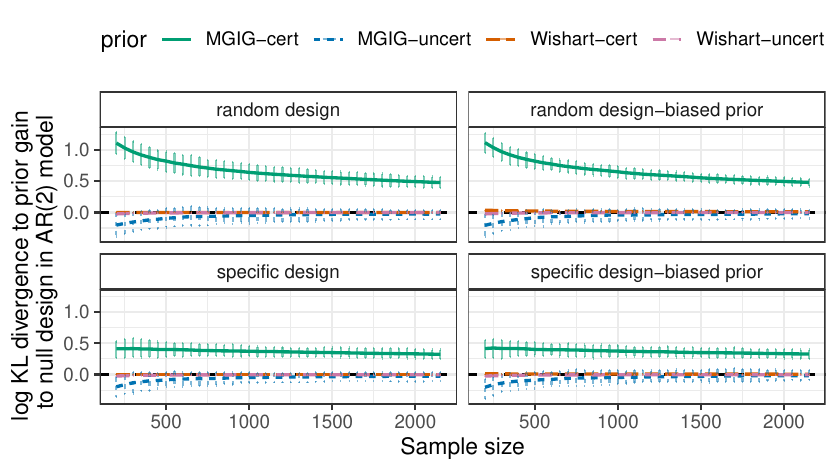}
    \caption{Difference in log KL divergence between prior and posterior comparing random experiment ($\mathbf X \ne \mathbf 0$) and specific experiment (diagonal $\mathbf X$) vs null experiment ($\mathbf X=\mathbf 0$) under AR2 models with 50 responses and 50 predictors, with and without biases on the prior of $\mathbf{B}$. Lines are averages over 100 repeats while error bars are 0.975 and 0.025 quantiles. }
    \label{fig:ar2_kl}
\end{figure}

\begin{figure}[H]
    \centering
    \includegraphics[width = 0.8\linewidth]{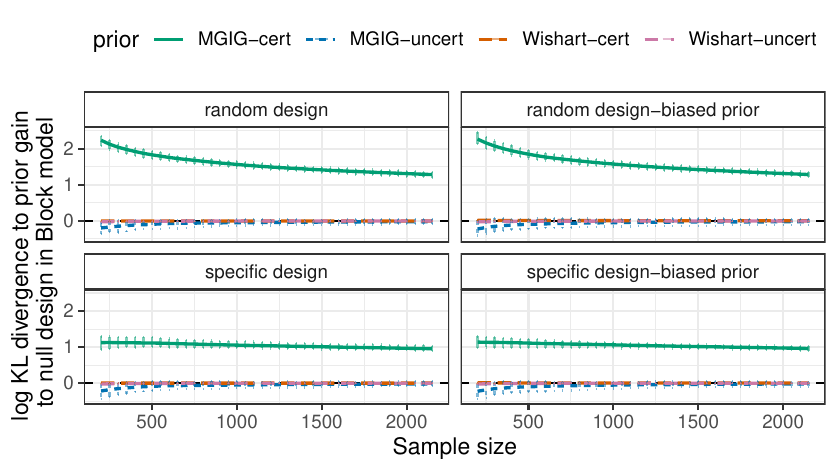}
    \caption{Difference in log KL divergence between prior and posterior comparing random experiment ($\mathbf X \ne \mathbf 0$) and specific experiment (diagonal $\mathbf X$) vs null experiment ($\mathbf X=\mathbf 0$) under block models with 50 responses and 50 predictors, with and without biases on the prior of $\mathbf{B}$. Lines are averages over 100 repeats while error bars are 0.975 and 0.025 quantiles. }
    \label{fig:block_kl}
\end{figure}

\begin{figure}[H]
    \centering
    \includegraphics[width = 0.8\linewidth]{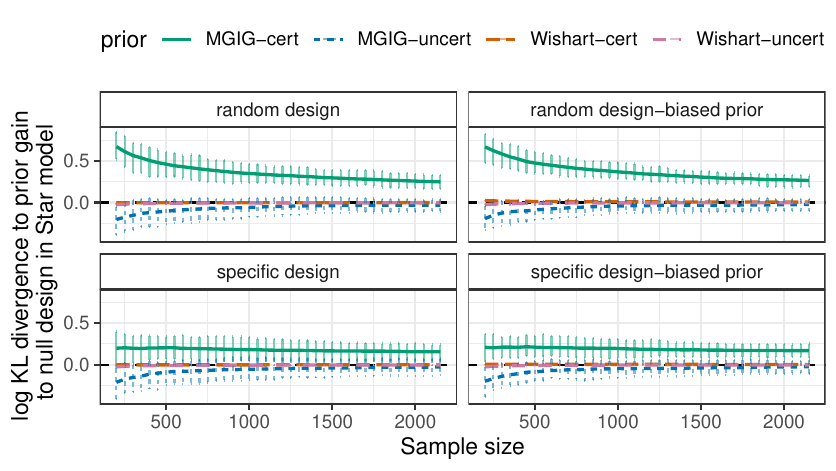}
    \caption{Difference in log KL divergence between prior and posterior comparing random experiment ($\mathbf X \ne \mathbf 0$) and specific experiment (diagonal $\mathbf X$) vs null experiment ($\mathbf X=\mathbf 0$) under Star models with 50 responses and 50 predictors, with and without biases on the prior of $\mathbf{B}$. Lines are averages over 100 repeats while error bars are 0.975 and 0.025 quantiles. }
    \label{fig:block_kl}
\end{figure}

\begin{figure}[H]
    \centering
    \includegraphics[width = 0.8\linewidth]{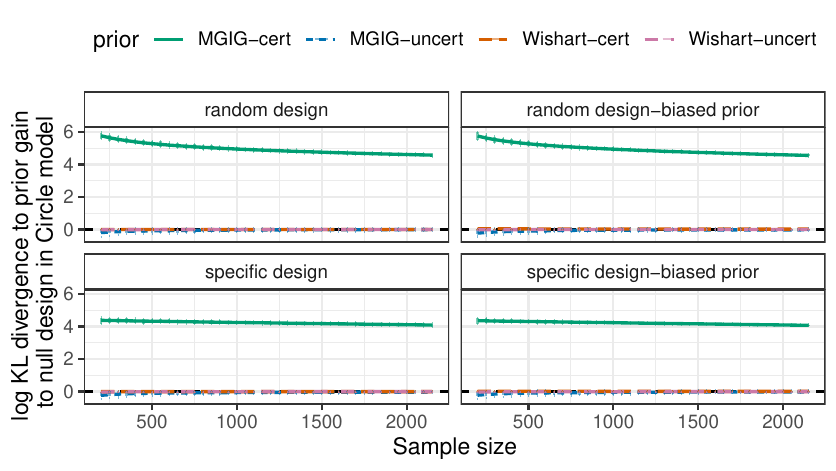}
    \caption{Difference in log KL divergence between prior and posterior comparing random experiment ($\mathbf X \ne \mathbf 0$) and specific experiment (diagonal $\mathbf X$) vs null experiment ($\mathbf X=\mathbf 0$) under circle models with 50 responses and 50 predictors, with and without biases on the prior of $\mathbf{B}$. Lines are averages over 100 repeats while error bars are 0.975 and 0.025 quantiles. }
    \label{fig:circle_kl}
\end{figure}

\begin{figure}[H]
    \centering
    \includegraphics[width = 0.8\linewidth]{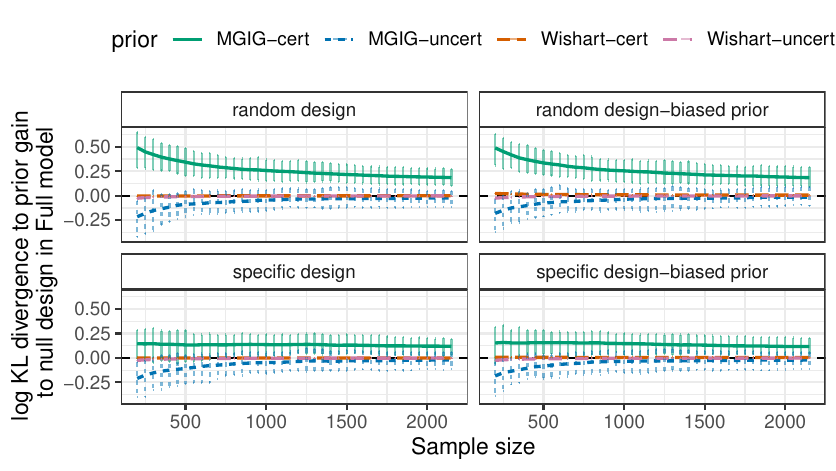}
    \caption{Difference in log KL divergence between prior and posterior comparing random experiment ($\mathbf X \ne \mathbf 0$) and specific experiment (diagonal $\mathbf X$) vs null experiment ($\mathbf X=\mathbf 0$) under Full models with 50 responses and 50 predictors, with and without biases on the prior of $\mathbf{B}$. Lines are averages over 100 repeats while error bars are 0.975 and 0.025 quantiles. }
    \label{fig:full_kl}
\end{figure}

\subsection{Stein's loss compared to null experiment}
\revision{This set of experiments suggests that good prior on $\mathbf{B}$ make experiments helpful in the point estimate of $\Omega$ in other models, but with biased prior this is not the case as positive Stein's loss change means a worse point estimate compared to null design.} 

\begin{figure}[H]
    \centering
    \includegraphics[width = 0.8\linewidth]{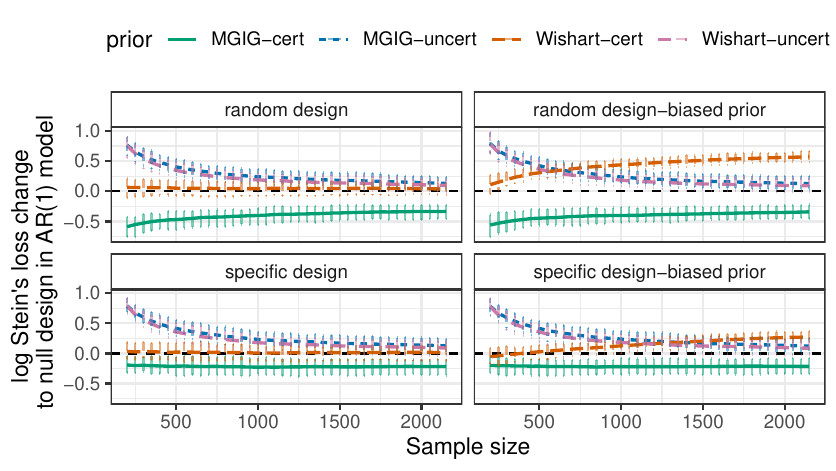}
    \caption{Difference in log Stein's loss of random experiment ($\mathbf X \ne \mathbf 0$) and specific experiment (diagonal $\mathbf X$) vs null experiment ($\mathbf X=\mathbf 0$) under AR1 models with 50 responses and 50 predictors, with and without biases on the prior of $\mathbf{B}$. Lines are averages over 100 repeats while error bars are 0.975 and 0.025 quantiles. }
    \label{fig:ar1stein}
\end{figure}

\begin{figure}[H]
    \centering
    \includegraphics[width = 0.8\linewidth]{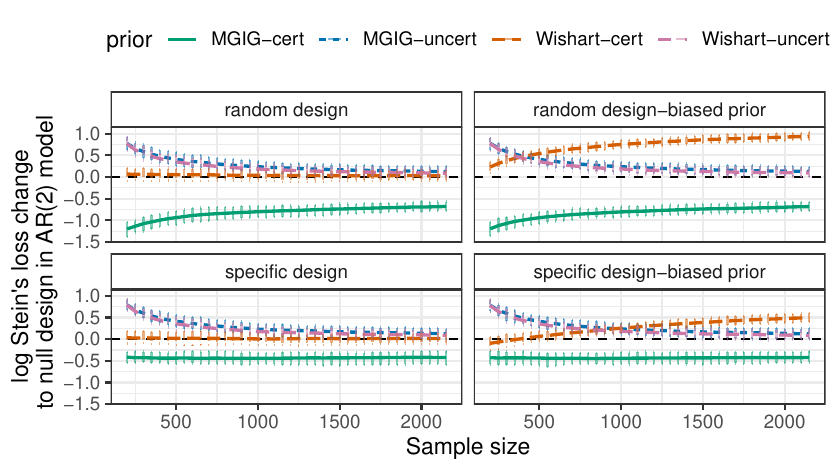}
    \caption{Difference in log Stein's loss of random experiment ($\mathbf X \ne \mathbf 0$) and specific experiment (diagonal $\mathbf X$) vs null experiment ($\mathbf X=\mathbf 0$) under AR2 models with 50 responses and 50 predictors, with and without biases on the prior of $\mathbf{B}$. Lines are averages over 100 repeats while error bars are 0.975 and 0.025 quantiles. }
    \label{fig:ar2stein}
\end{figure}

\begin{figure}[H]
    \centering
    \includegraphics[width = 0.8\linewidth]{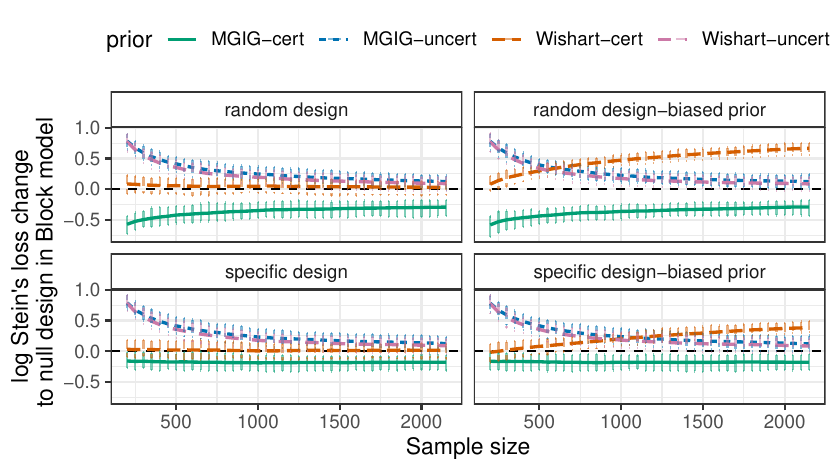}
    \caption{Difference in log Stein's loss of random experiment ($\mathbf X \ne \mathbf 0$) and specific experiment (diagonal $\mathbf X$) vs null experiment ($\mathbf X=\mathbf 0$) under Block models with 50 responses and 50 predictors, with and without biases on the prior of $\mathbf{B}$. Lines are averages over 100 repeats while error bars are 0.975 and 0.025 quantiles. }
    \label{fig:Blockstein}
\end{figure}

\begin{figure}[H]
    \centering
    \includegraphics[width = 0.8\linewidth]{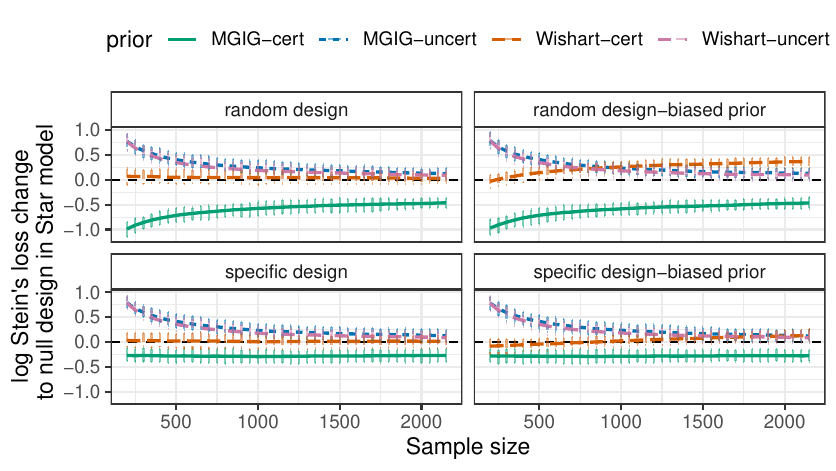}
    \caption{Difference in log Stein's loss of random experiment ($\mathbf X \ne \mathbf 0$) and specific experiment (diagonal $\mathbf X$) vs null experiment ($\mathbf X=\mathbf 0$) under Star models with 50 responses and 50 predictors, with and without biases on the prior of $\mathbf{B}$. Lines are averages over 100 repeats while error bars are 0.975 and 0.025 quantiles. }
    \label{fig:Blockstein}
\end{figure}

\begin{figure}[H]
    \centering
    \includegraphics[width = 0.8\linewidth]{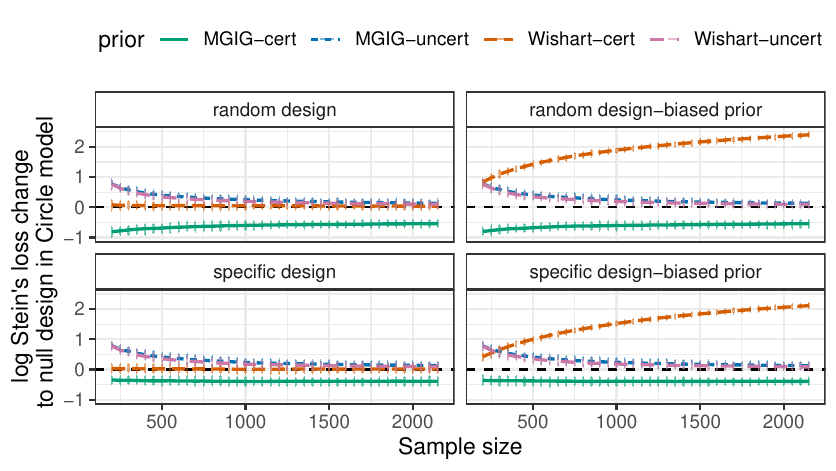}
    \caption{Difference in log Stein's loss of random experiment ($\mathbf X \ne \mathbf 0$) and specific experiment (diagonal $\mathbf X$) vs null experiment ($\mathbf X=\mathbf 0$) under Circle models with 50 responses and 50 predictors, with and without biases on the prior of $\mathbf{B}$. Lines are averages over 100 repeats while error bars are 0.975 and 0.025 quantiles. }
    \label{fig:circlestein}
\end{figure}

\begin{figure}[H]
    \centering
    \includegraphics[width = 0.8\linewidth]{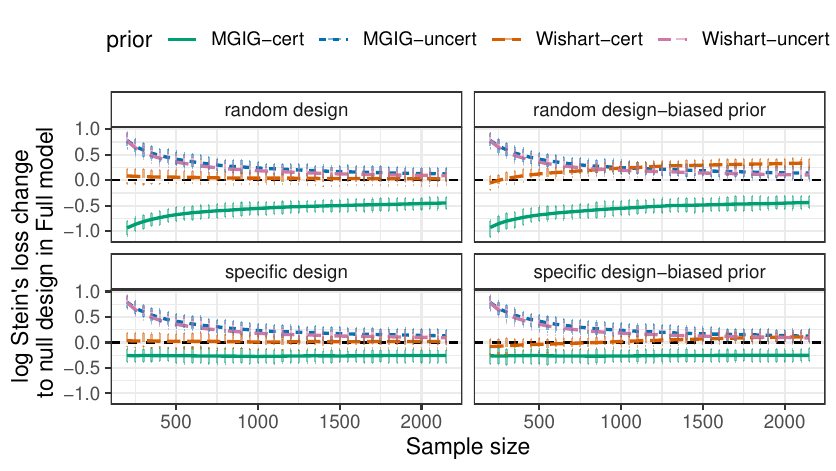}
    \caption{Difference in log Stein's loss of random experiment ($\mathbf X \ne \mathbf 0$) and specific experiment (diagonal $\mathbf X$) vs null experiment ($\mathbf X=\mathbf 0$) under Full models with 50 responses and 50 predictors, with and without biases on the prior of $\mathbf{B}$. Lines are averages over 100 repeats while error bars are 0.975 and 0.025 quantiles. }
    \label{fig:fullstein}
\end{figure}

\section{\revision{Simulations with midpoint level uncertainty on $\mathbf{B}$}}
\label{app:miduncertainty}

\revision{Here, we changed the uncertainty level to $\mathbf{\Lambda}=0.1\mathbf{I}_p, 10\mathbf{I}_p$, with two different levels of bias. A general observation is that with bias in the prior having some uncertainty will be helpful in small sample size, but not with more data. }

\subsection*{\revision{Low bias scheme}}
\revision{Here, we set the bias of the prior to $N(0,0.1)$.}

\subsubsection*{KL divergence compared to null experiment}
\begin{figure}[H]
    \centering
    \includegraphics[width = 0.8\linewidth]{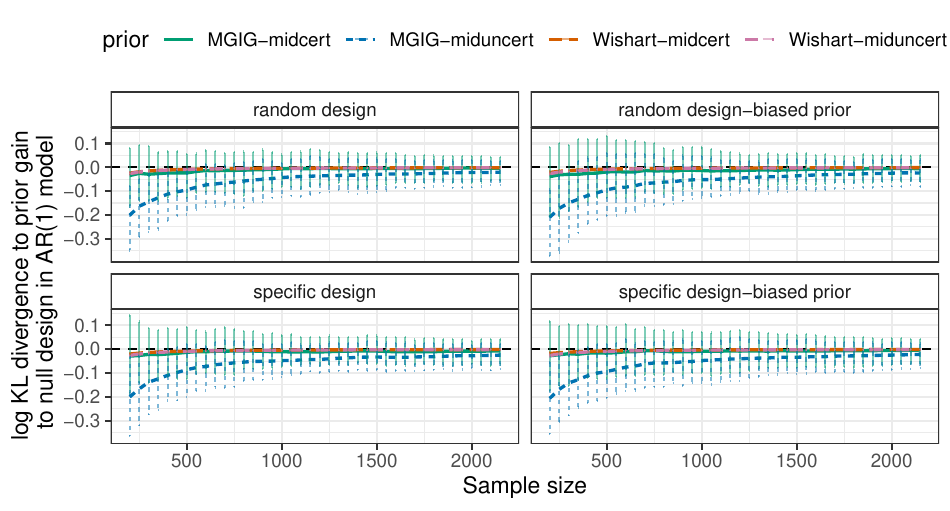}
    \caption{Difference in log KL divergence between prior and posterior comparing random experiment ($\mathbf X \ne \mathbf 0$) and specific experiment (diagonal $\mathbf X$) vs null experiment ($\mathbf X=\mathbf 0$) under AR1 models with 50 responses and 50 predictors, with and without biases on the prior of $\mathbf{B}$. Lines are averages over 100 repeats while error bars are 0.975 and 0.025 quantiles. }
    \label{fig:ar1_kl2}
\end{figure}

\begin{figure}[H]
    \centering
    \includegraphics[width = 0.8\linewidth]{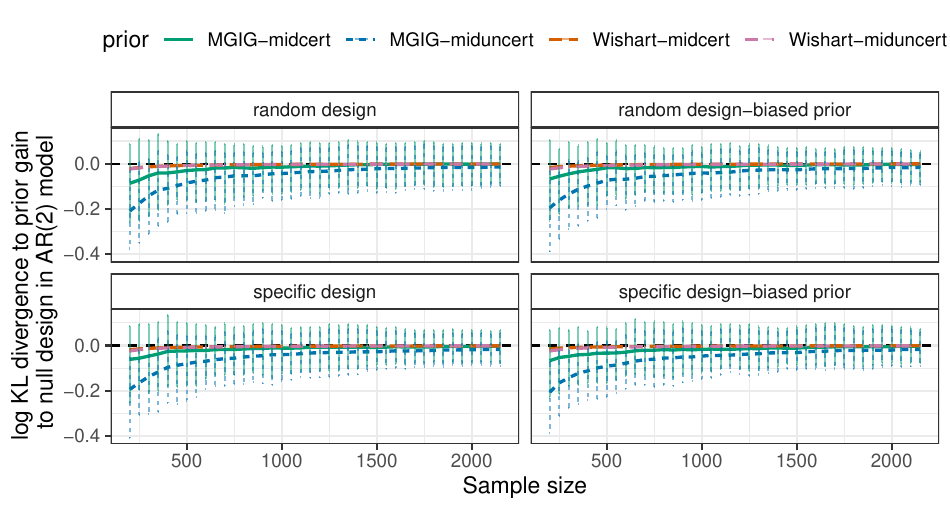}
    \caption{Difference in log KL divergence between prior and posterior comparing random experiment ($\mathbf X \ne \mathbf 0$) and specific experiment (diagonal $\mathbf X$) vs null experiment ($\mathbf X=\mathbf 0$) under AR2 models with 50 responses and 50 predictors, with and without biases on the prior of $\mathbf{B}$. Lines are averages over 100 repeats while error bars are 0.975 and 0.025 quantiles. }
    \label{fig:ar2_kl}
\end{figure}

\begin{figure}[H]
    \centering
    \includegraphics[width = 0.8\linewidth]{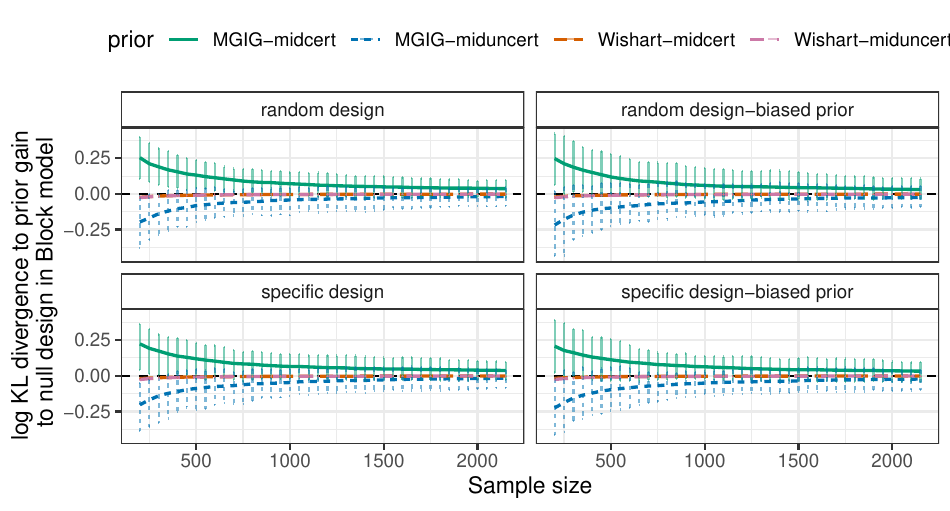}
    \caption{Difference in log KL divergence between prior and posterior comparing random experiment ($\mathbf X \ne \mathbf 0$) and specific experiment (diagonal $\mathbf X$) vs null experiment ($\mathbf X=\mathbf 0$) under block models with 50 responses and 50 predictors, with and without biases on the prior of $\mathbf{B}$. Lines are averages over 100 repeats while error bars are 0.975 and 0.025 quantiles. }
    \label{fig:block_kl}
\end{figure}

\begin{figure}[H]
    \centering
    \includegraphics[width = 0.8\linewidth]{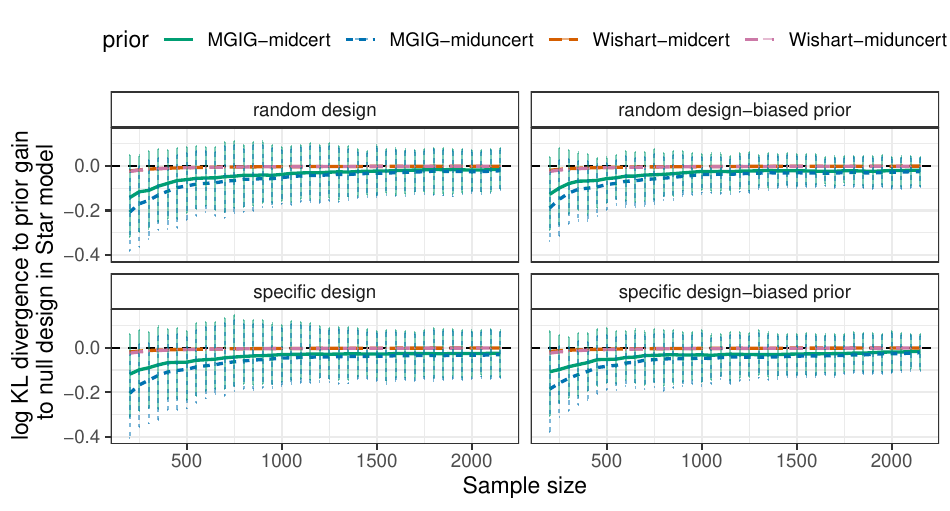}
    \caption{Difference in log KL divergence between prior and posterior comparing random experiment ($\mathbf X \ne \mathbf 0$) and specific experiment (diagonal $\mathbf X$) vs null experiment ($\mathbf X=\mathbf 0$) under Star models with 50 responses and 50 predictors, with and without biases on the prior of $\mathbf{B}$. Lines are averages over 100 repeats while error bars are 0.975 and 0.025 quantiles. }
    \label{fig:block_kl}
\end{figure}

\begin{figure}[H]
    \centering
    \includegraphics[width = 0.8\linewidth]{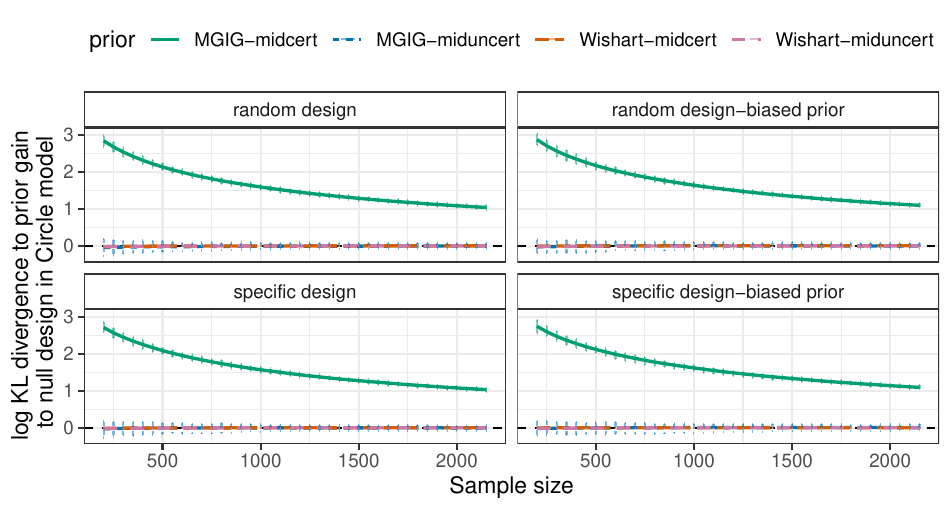}
    \caption{Difference in log KL divergence between prior and posterior comparing random experiment ($\mathbf X \ne \mathbf 0$) and specific experiment (diagonal $\mathbf X$) vs null experiment ($\mathbf X=\mathbf 0$) under circle models with 50 responses and 50 predictors, with and without biases on the prior of $\mathbf{B}$. Lines are averages over 100 repeats while error bars are 0.975 and 0.025 quantiles. }
    \label{fig:circle_kl}
\end{figure}

\begin{figure}[H]
    \centering
    \includegraphics[width = 0.8\linewidth]{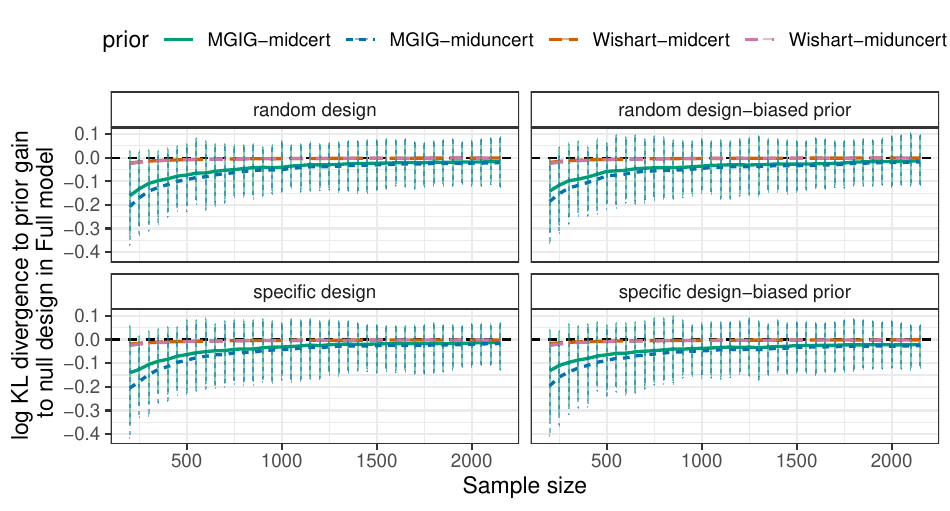}
    \caption{Difference in log KL divergence between prior and posterior comparing random experiment ($\mathbf X \ne \mathbf 0$) and specific experiment (diagonal $\mathbf X$) vs null experiment ($\mathbf X=\mathbf 0$) under Full models with 50 responses and 50 predictors, with and without biases on the prior of $\mathbf{B}$. Lines are averages over 100 repeats while error bars are 0.975 and 0.025 quantiles. }
    \label{fig:full_kl}
\end{figure}

\subsubsection*{Stein's loss compared to null experiment}
\revision{This set of experiments suggests that good prior on $\mathbf{B}$ make experiments helpful in the point estimate of $\Omega$ in other models, but with biased prior this is not the case as more positive Stein's loss change means a worse point estimate compared to null design.}

\begin{figure}[H]
    \centering
    \includegraphics[width = 0.8\linewidth]{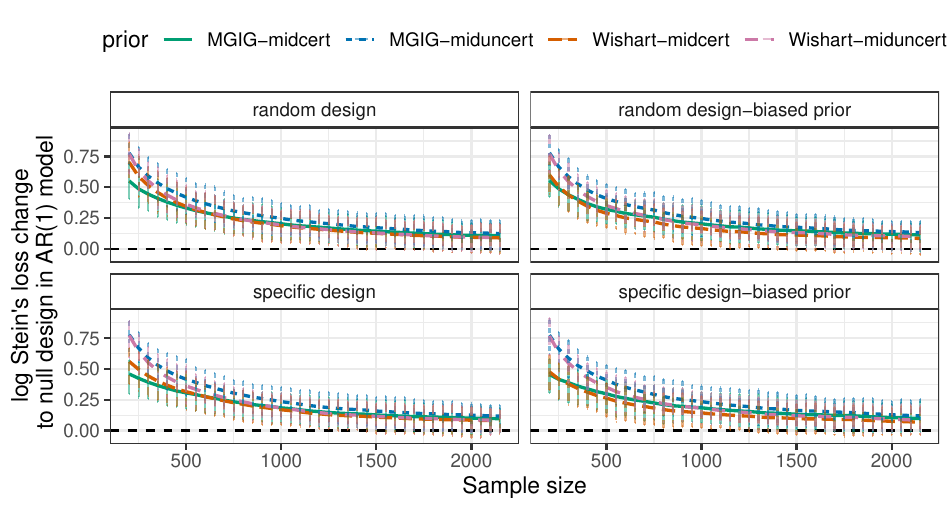}
    \caption{Difference in log Stein's loss of random experiment ($\mathbf X \ne \mathbf 0$) and specific experiment (diagonal $\mathbf X$) vs null experiment ($\mathbf X=\mathbf 0$) under AR1 models with 50 responses and 50 predictors, with and without biases on the prior of $\mathbf{B}$. Lines are averages over 100 repeats while error bars are 0.975 and 0.025 quantiles. }
    \label{fig:ar1stein}
\end{figure}

\begin{figure}[H]
    \centering
    \includegraphics[width = 0.8\linewidth]{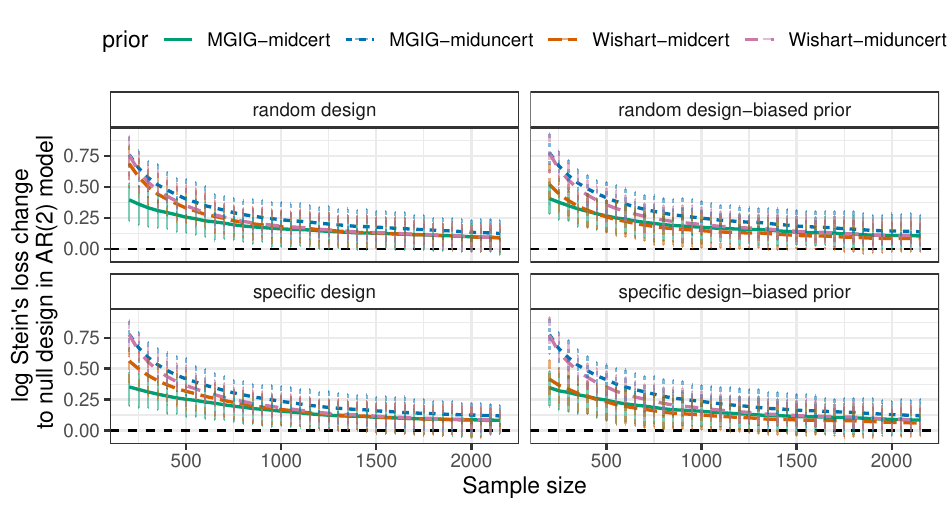}
    \caption{Difference in log Stein's loss of random experiment ($\mathbf X \ne \mathbf 0$) and specific experiment (diagonal $\mathbf X$) vs null experiment ($\mathbf X=\mathbf 0$) under AR2 models with 50 responses and 50 predictors, with and without biases on the prior of $\mathbf{B}$. Lines are averages over 100 repeats while error bars are 0.975 and 0.025 quantiles. }
    \label{fig:ar2stein}
\end{figure}

\begin{figure}[H]
    \centering
    \includegraphics[width = 0.8\linewidth]{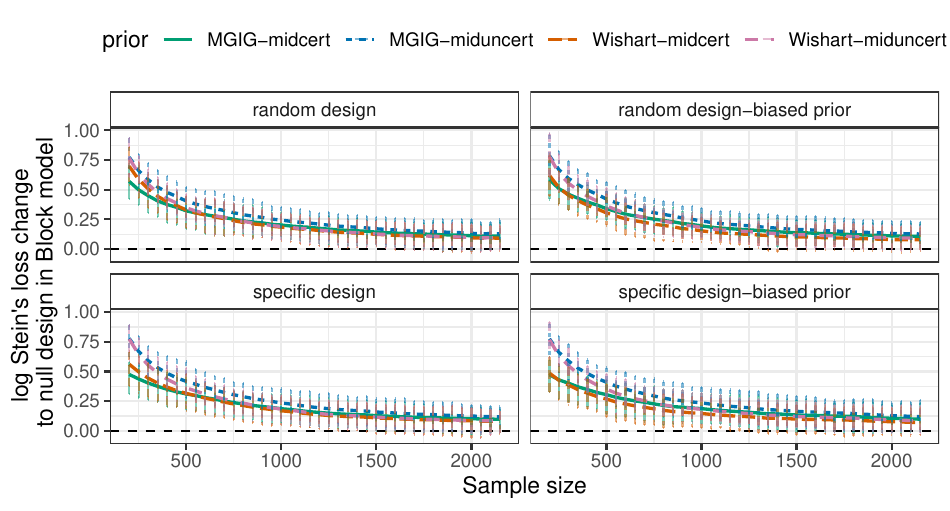}
    \caption{Difference in log Stein's loss of random experiment ($\mathbf X \ne \mathbf 0$) and specific experiment (diagonal $\mathbf X$) vs null experiment ($\mathbf X=\mathbf 0$) under Block models with 50 responses and 50 predictors, with and without biases on the prior of $\mathbf{B}$. Lines are averages over 100 repeats while error bars are 0.975 and 0.025 quantiles. }
    \label{fig:Blockstein}
\end{figure}

\begin{figure}[H]
    \centering
    \includegraphics[width = 0.8\linewidth]{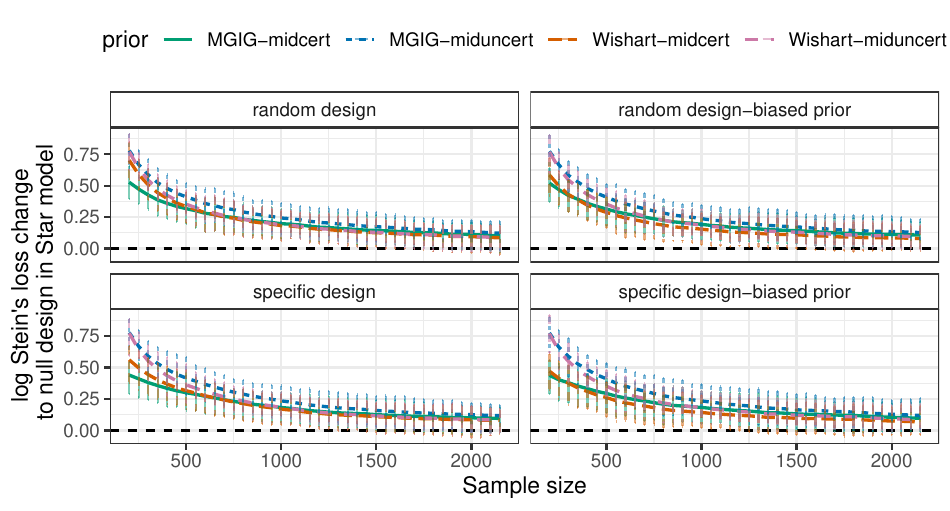}
    \caption{Difference in log Stein's loss of random experiment ($\mathbf X \ne \mathbf 0$) and specific experiment (diagonal $\mathbf X$) vs null experiment ($\mathbf X=\mathbf 0$) under Star models with 50 responses and 50 predictors, with and without biases on the prior of $\mathbf{B}$. Lines are averages over 100 repeats while error bars are 0.975 and 0.025 quantiles. }
    \label{fig:Blockstein}
\end{figure}

\begin{figure}[H]
    \centering
    \includegraphics[width = 0.8\linewidth]{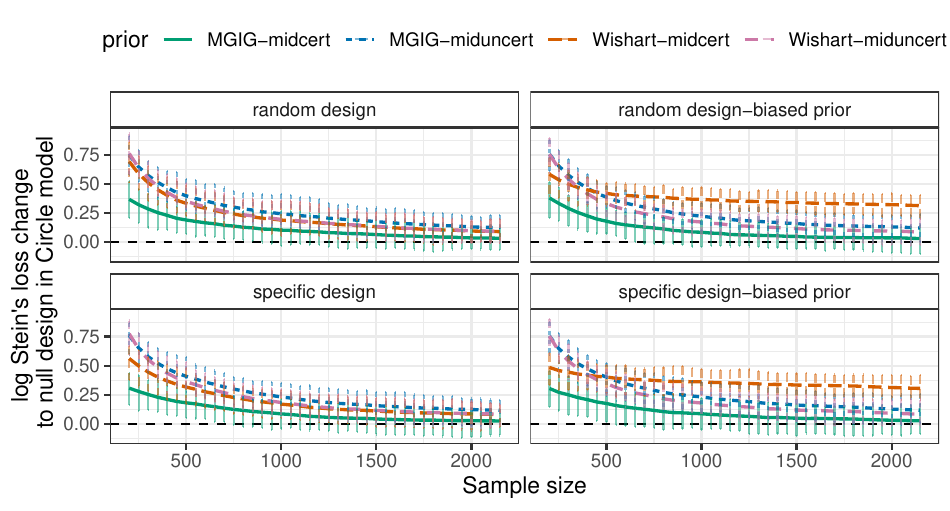}
    \caption{Difference in log Stein's loss of random experiment ($\mathbf X \ne \mathbf 0$) and specific experiment (diagonal $\mathbf X$) vs null experiment ($\mathbf X=\mathbf 0$) under Circle models with 50 responses and 50 predictors, with and without biases on the prior of $\mathbf{B}$. Lines are averages over 100 repeats while error bars are 0.975 and 0.025 quantiles. }
    \label{fig:circlestein}
\end{figure}

\begin{figure}[H]
    \centering
    \includegraphics[width = 0.8\linewidth]{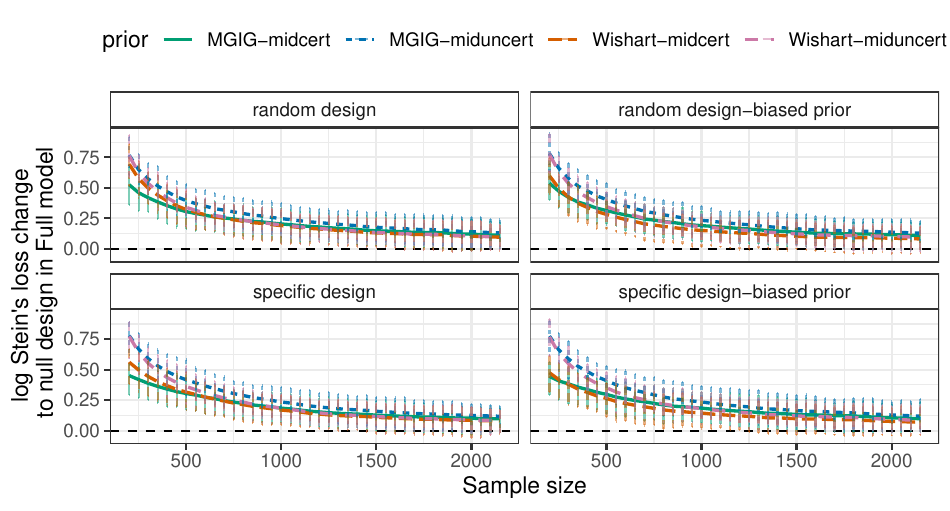}
    \caption{Difference in log Stein's loss of random experiment ($\mathbf X \ne \mathbf 0$) and specific experiment (diagonal $\mathbf X$) vs null experiment ($\mathbf X=\mathbf 0$) under Full models with 50 responses and 50 predictors, with and without biases on the prior of $\mathbf{B}$. Lines are averages over 100 repeats while error bars are 0.975 and 0.025 quantiles. }
    \label{fig:fullstein}
\end{figure}

\subsection*{\revision{High bias scheme}}
\revision{Here, we set the bias of the prior to $N(0,1)$.}

\subsubsection*{KL divergence compared to null experiment}
\begin{figure}[H]
    \centering
    \includegraphics[width = 0.8\linewidth]{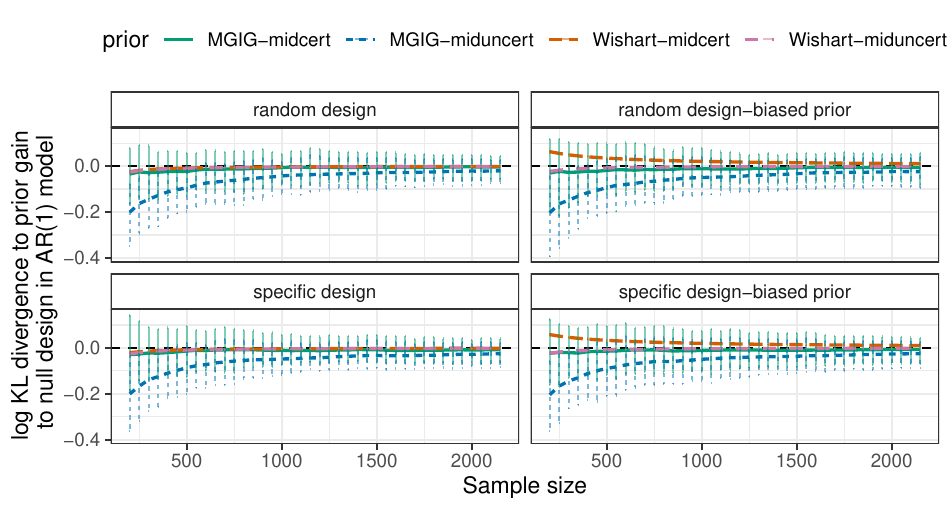}
    \caption{Difference in log KL divergence between prior and posterior comparing random experiment ($\mathbf X \ne \mathbf 0$) and specific experiment (diagonal $\mathbf X$) vs null experiment ($\mathbf X=\mathbf 0$) under AR1 models with 50 responses and 50 predictors, with and without biases on the prior of $\mathbf{B}$. Lines are averages over 100 repeats while error bars are 0.975 and 0.025 quantiles. }
    \label{fig:ar1_kl2}
\end{figure}

\begin{figure}[H]
    \centering
    \includegraphics[width = 0.8\linewidth]{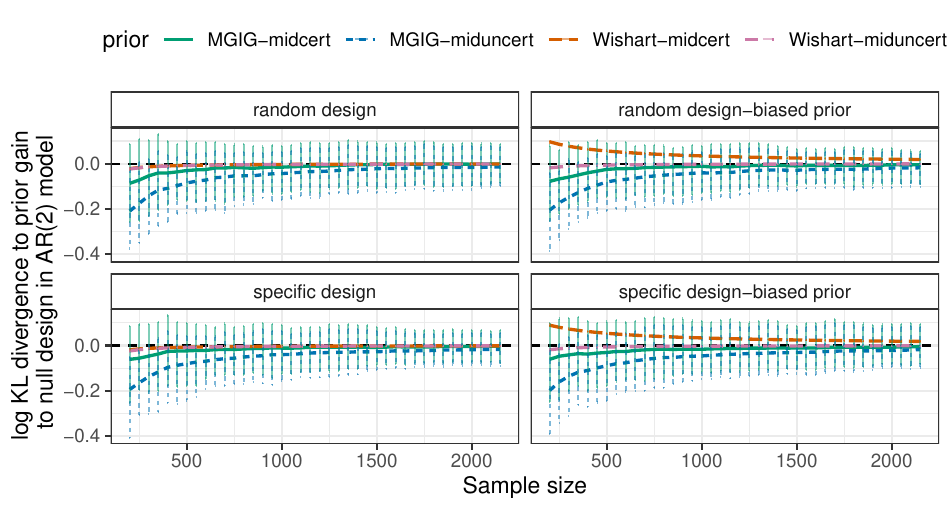}
    \caption{Difference in log KL divergence between prior and posterior comparing random experiment ($\mathbf X \ne \mathbf 0$) and specific experiment (diagonal $\mathbf X$) vs null experiment ($\mathbf X=\mathbf 0$) under AR2 models with 50 responses and 50 predictors, with and without biases on the prior of $\mathbf{B}$. Lines are averages over 100 repeats while error bars are 0.975 and 0.025 quantiles. }
    \label{fig:ar2_kl}
\end{figure}

\begin{figure}[H]
    \centering
    \includegraphics[width = 0.8\linewidth]{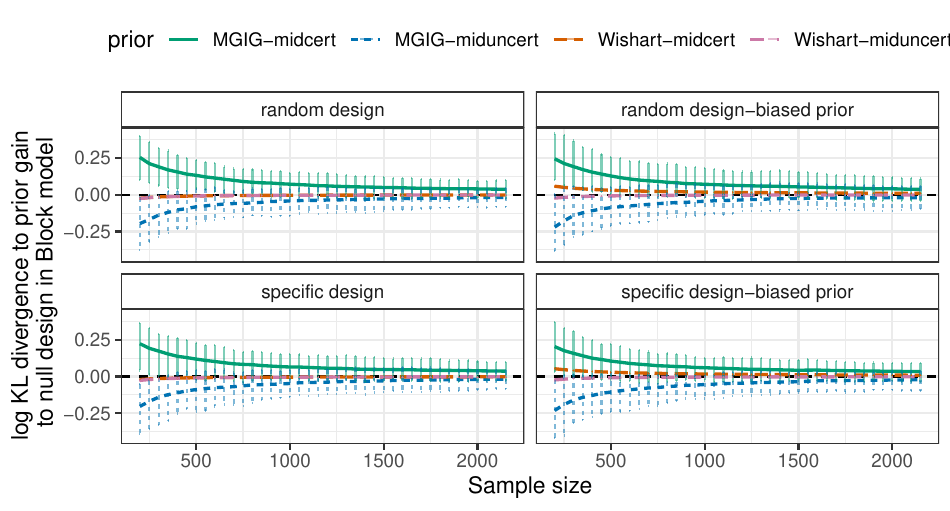}
    \caption{Difference in log KL divergence between prior and posterior comparing random experiment ($\mathbf X \ne \mathbf 0$) and specific experiment (diagonal $\mathbf X$) vs null experiment ($\mathbf X=\mathbf 0$) under block models with 50 responses and 50 predictors, with and without biases on the prior of $\mathbf{B}$. Lines are averages over 100 repeats while error bars are 0.975 and 0.025 quantiles. }
    \label{fig:block_kl}
\end{figure}

\begin{figure}[H]
    \centering
    \includegraphics[width = 0.8\linewidth]{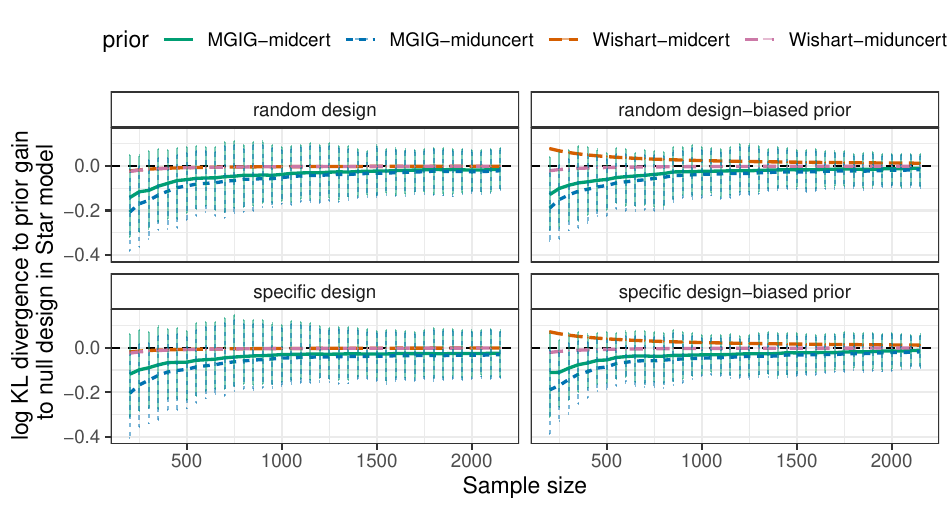}
    \caption{Difference in log KL divergence between prior and posterior comparing random experiment ($\mathbf X \ne \mathbf 0$) and specific experiment (diagonal $\mathbf X$) vs null experiment ($\mathbf X=\mathbf 0$) under Star models with 50 responses and 50 predictors, with and without biases on the prior of $\mathbf{B}$. Lines are averages over 100 repeats while error bars are 0.975 and 0.025 quantiles. }
    \label{fig:block_kl}
\end{figure}

\begin{figure}[H]
    \centering
    \includegraphics[width = 0.8\linewidth]{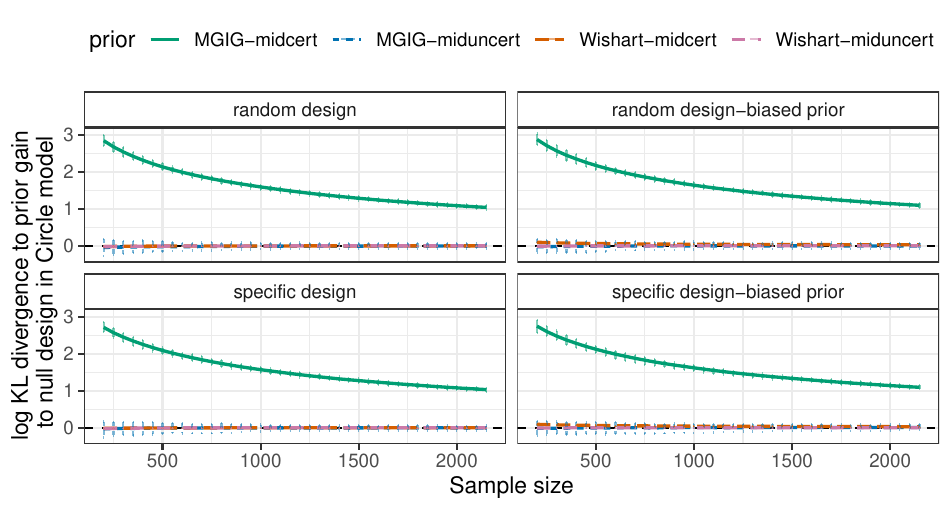}
    \caption{Difference in log KL divergence between prior and posterior comparing random experiment ($\mathbf X \ne \mathbf 0$) and specific experiment (diagonal $\mathbf X$) vs null experiment ($\mathbf X=\mathbf 0$) under circle models with 50 responses and 50 predictors, with and without biases on the prior of $\mathbf{B}$. Lines are averages over 100 repeats while error bars are 0.975 and 0.025 quantiles. }
    \label{fig:circle_kl}
\end{figure}

\begin{figure}[H]
    \centering
    \includegraphics[width = 0.8\linewidth]{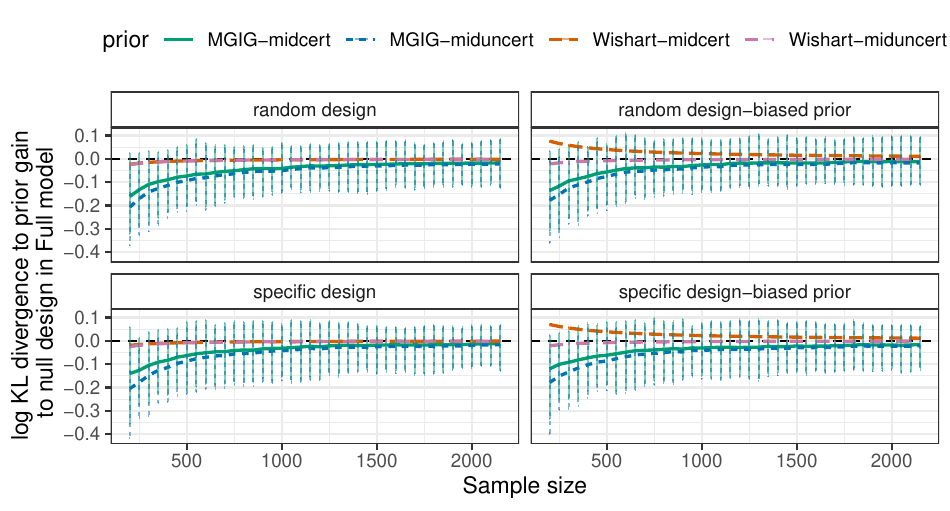}
    \caption{Difference in log KL divergence between prior and posterior comparing random experiment ($\mathbf X \ne \mathbf 0$) and specific experiment (diagonal $\mathbf X$) vs null experiment ($\mathbf X=\mathbf 0$) under Full models with 50 responses and 50 predictors, with and without biases on the prior of $\mathbf{B}$. Lines are averages over 100 repeats while error bars are 0.975 and 0.025 quantiles. }
    \label{fig:full_kl}
\end{figure}

\subsubsection*{Stein's loss compared to null experiment}
\revision{This set of experiments suggests that good prior on $\mathbf{B}$ make experiments helpful in the point estimate of $\Omega$ in other models, but with biased prior this is not the case as more positive Stein's loss change means a worse point estimate compared to null design.} 

\begin{figure}[H]
    \centering
    \includegraphics[width = 0.8\linewidth]{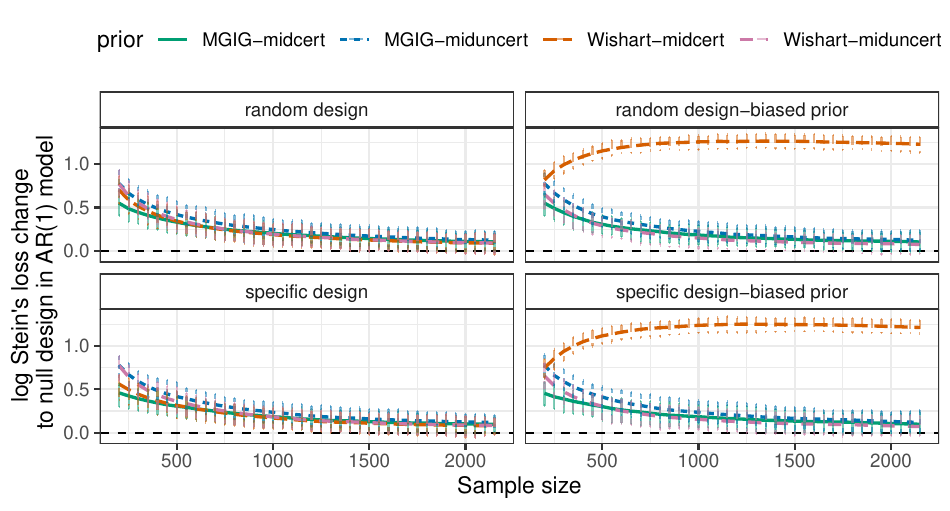}
    \caption{Difference in log Stein's loss of random experiment ($\mathbf X \ne \mathbf 0$) and specific experiment (diagonal $\mathbf X$) vs null experiment ($\mathbf X=\mathbf 0$) under AR1 models with 50 responses and 50 predictors, with and without biases on the prior of $\mathbf{B}$. Lines are averages over 100 repeats while error bars are 0.975 and 0.025 quantiles. }
    \label{fig:ar1stein}
\end{figure}

\begin{figure}[H]
    \centering
    \includegraphics[width = 0.8\linewidth]{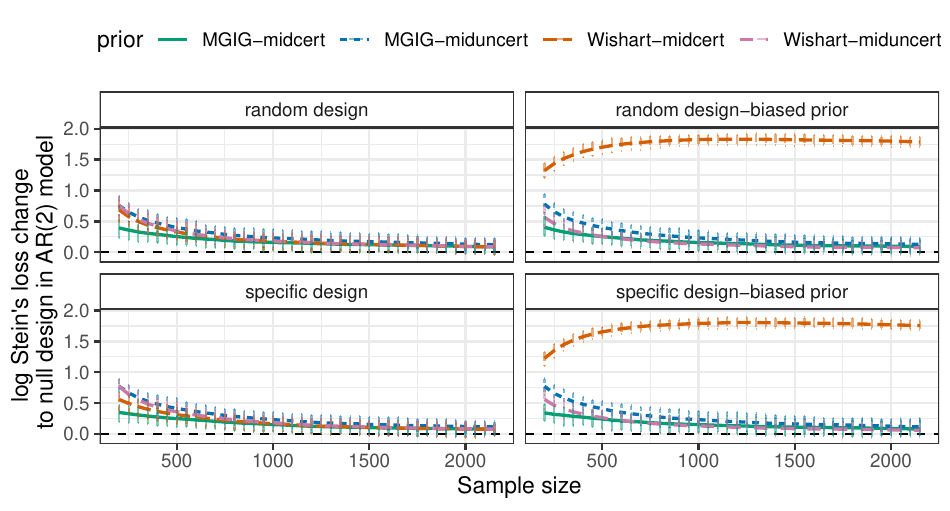}
    \caption{Difference in log Stein's loss of random experiment ($\mathbf X \ne \mathbf 0$) and specific experiment (diagonal $\mathbf X$) vs null experiment ($\mathbf X=\mathbf 0$) under AR2 models with 50 responses and 50 predictors, with and without biases on the prior of $\mathbf{B}$. Lines are averages over 100 repeats while error bars are 0.975 and 0.025 quantiles. }
    \label{fig:ar2stein}
\end{figure}

\begin{figure}[H]
    \centering
    \includegraphics[width = 0.8\linewidth]{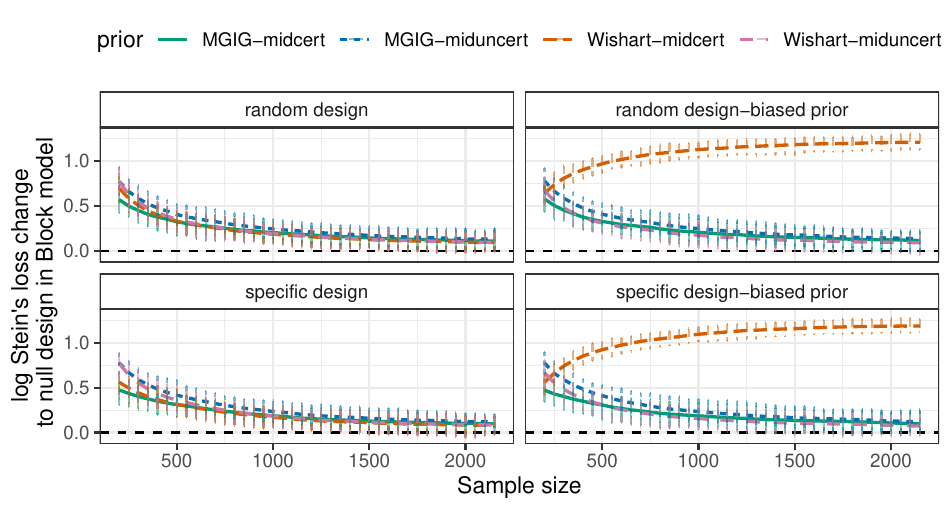}
    \caption{Difference in log Stein's loss of random experiment ($\mathbf X \ne \mathbf 0$) and specific experiment (diagonal $\mathbf X$) vs null experiment ($\mathbf X=\mathbf 0$) under Block models with 50 responses and 50 predictors, with and without biases on the prior of $\mathbf{B}$. Lines are averages over 100 repeats while error bars are 0.975 and 0.025 quantiles. }
    \label{fig:Blockstein}
\end{figure}

\begin{figure}[H]
    \centering
    \includegraphics[width = 0.8\linewidth]{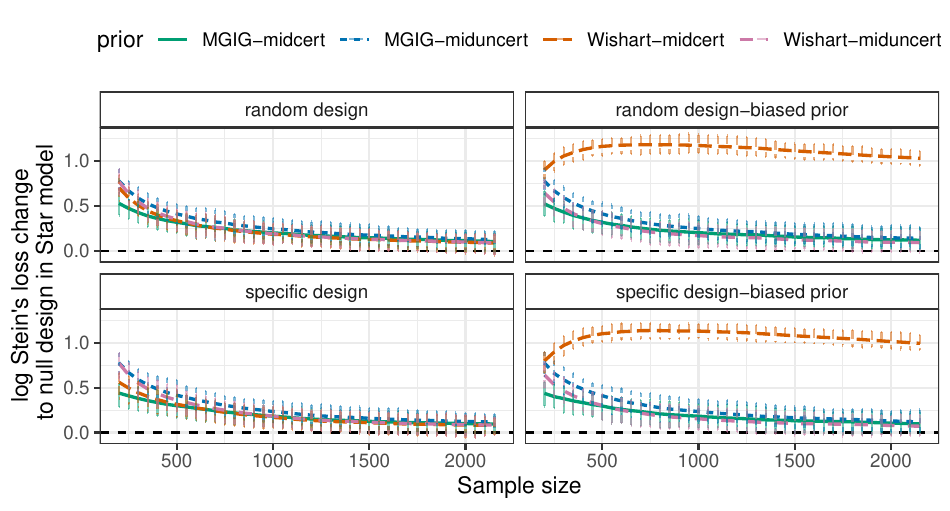}
    \caption{Difference in log Stein's loss of random experiment ($\mathbf X \ne \mathbf 0$) and specific experiment (diagonal $\mathbf X$) vs null experiment ($\mathbf X=\mathbf 0$) under Star models with 50 responses and 50 predictors, with and without biases on the prior of $\mathbf{B}$. Lines are averages over 100 repeats while error bars are 0.975 and 0.025 quantiles. }
    \label{fig:Blockstein}
\end{figure}

\begin{figure}[H]
    \centering
    \includegraphics[width = 0.8\linewidth]{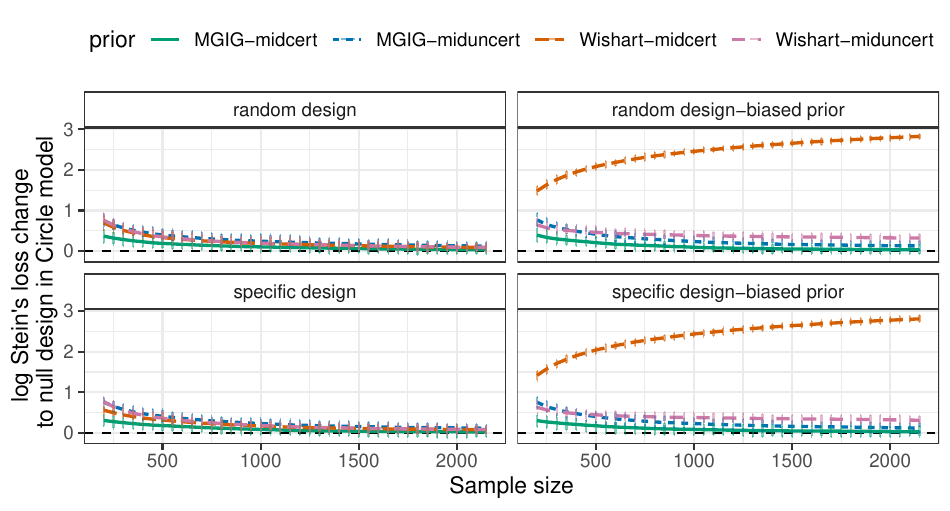}
    \caption{Difference in log Stein's loss of random experiment ($\mathbf X \ne \mathbf 0$) and specific experiment (diagonal $\mathbf X$) vs null experiment ($\mathbf X=\mathbf 0$) under Circle models with 50 responses and 50 predictors, with and without biases on the prior of $\mathbf{B}$. Lines are averages over 100 repeats while error bars are 0.975 and 0.025 quantiles. }
    \label{fig:circlestein}
\end{figure}

\begin{figure}[H]
    \centering
    \includegraphics[width = 0.8\linewidth]{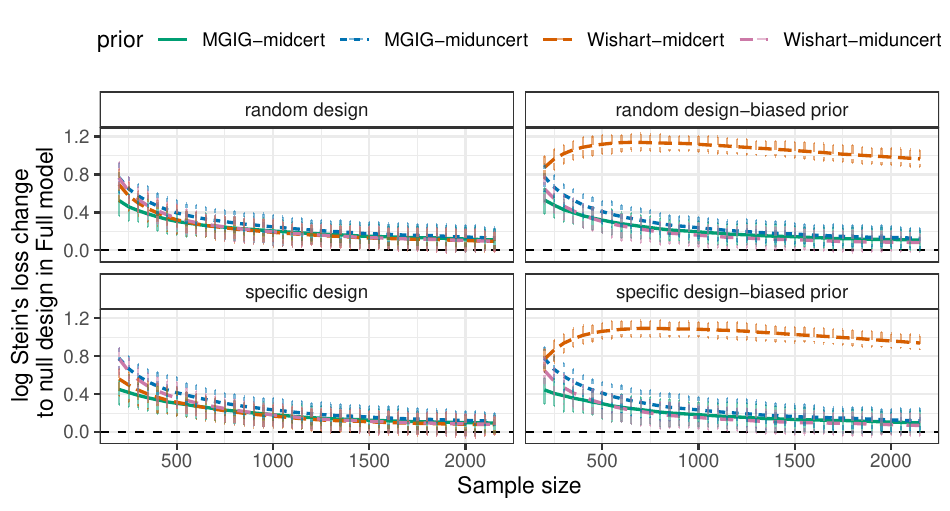}
    \caption{Difference in log Stein's loss of random experiment ($\mathbf X \ne \mathbf 0$) and specific experiment (diagonal $\mathbf X$) vs null experiment ($\mathbf X=\mathbf 0$) under Full models with 50 responses and 50 predictors, with and without biases on the prior of $\mathbf{B}$. Lines are averages over 100 repeats while error bars are 0.975 and 0.025 quantiles. }
    \label{fig:fullstein}
\end{figure}

\section{\revision{Simulations with violations to distributional assumptions}}

\revision{Here, we investigate the performance of the inference under violations to the distributional assumptions. 
We repeat some of the experiments to have data generated from a multivariate Laplace distribution and multivariate Student's t distribution with two degrees of freedom while still fit a Gaussian chain graph model to infer the precision matrix. The posterior is based on Gaussian likelihood.}

\subsection*{Laplace tail}

\revision{Inference performance is quite robust to Laplace tails.}

\subsubsection*{KL divergence compared to null experiments}
\begin{figure}[H]
    \centering
    \includegraphics[width = 0.8\linewidth]{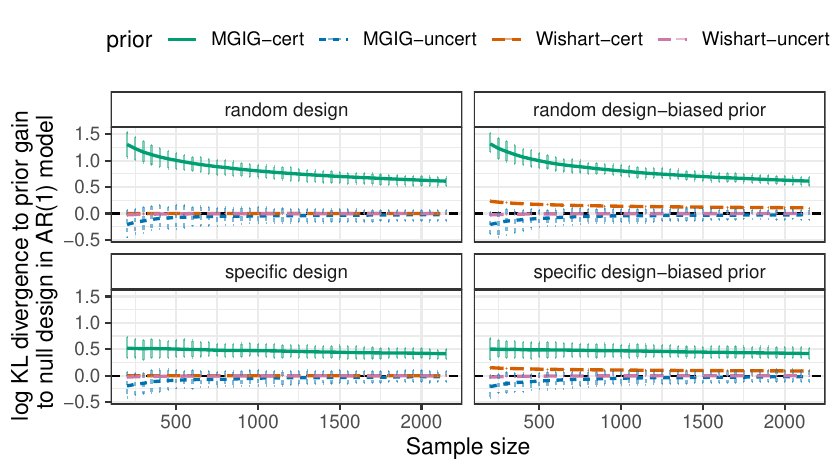}
    \caption{Difference in log KL divergence between prior and posterior comparing random experiment ($\mathbf X \ne \mathbf 0$) and specific experiment (diagonal $\mathbf X$) vs null experiment ($\mathbf X=\mathbf 0$) under AR1 models with 50 responses and 50 predictors, with and without biases on the prior of $\mathbf{B}$. Lines are averages over 100 repeats while error bars are 0.975 and 0.025 quantiles. }
    \label{fig:ar1_kl2}
\end{figure}

\begin{figure}[H]
    \centering
    \includegraphics[width = 0.8\linewidth]{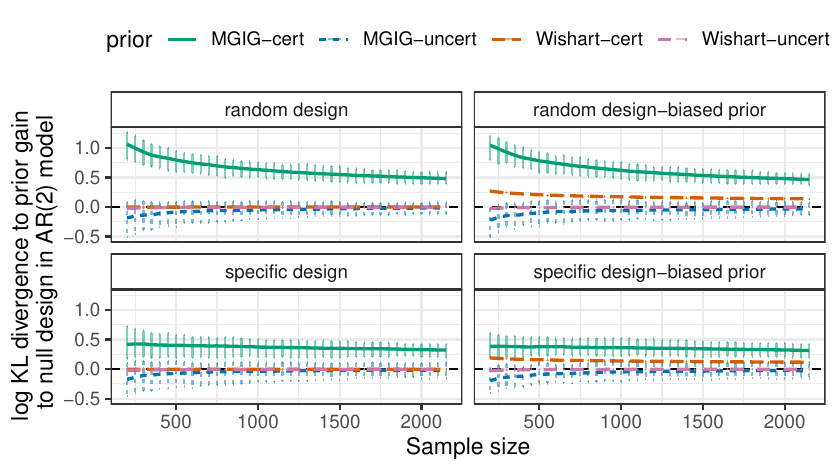}
    \caption{Difference in log KL divergence between prior and posterior comparing random experiment ($\mathbf X \ne \mathbf 0$) and specific experiment (diagonal $\mathbf X$) vs null experiment ($\mathbf X=\mathbf 0$) under AR2 models with 50 responses and 50 predictors, with and without biases on the prior of $\mathbf{B}$. Lines are averages over 100 repeats while error bars are 0.975 and 0.025 quantiles. }
    \label{fig:ar2_kl}
\end{figure}

\begin{figure}[H]
    \centering
    \includegraphics[width = 0.8\linewidth]{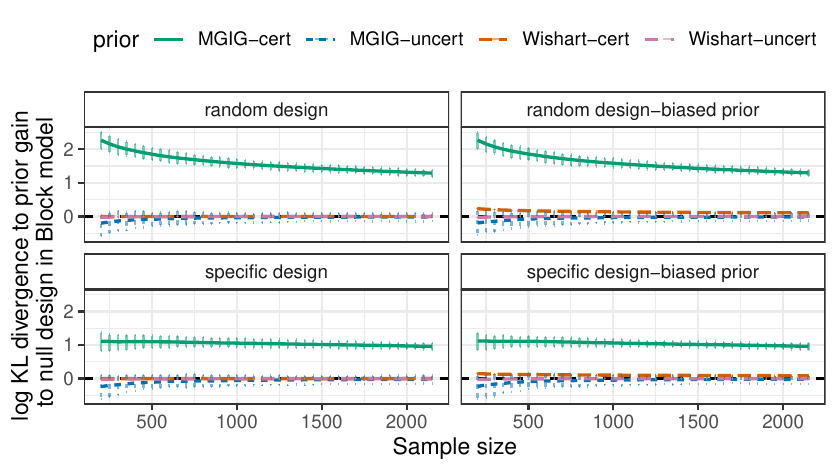}
    \caption{Difference in log KL divergence between prior and posterior comparing random experiment ($\mathbf X \ne \mathbf 0$) and specific experiment (diagonal $\mathbf X$) vs null experiment ($\mathbf X=\mathbf 0$) under block models with 50 responses and 50 predictors, with and without biases on the prior of $\mathbf{B}$. Lines are averages over 100 repeats while error bars are 0.975 and 0.025 quantiles. }
    \label{fig:block_kl}
\end{figure}

\begin{figure}[H]
    \centering
    \includegraphics[width = 0.8\linewidth]{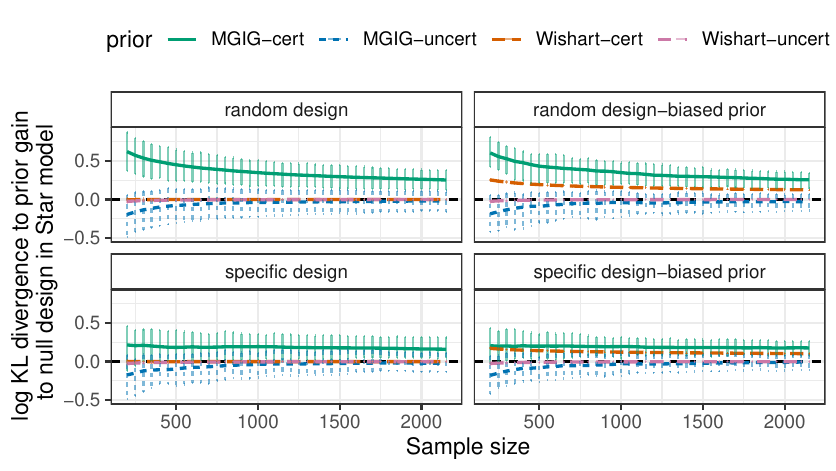}
    \caption{Difference in log KL divergence between prior and posterior comparing random experiment ($\mathbf X \ne \mathbf 0$) and specific experiment (diagonal $\mathbf X$) vs null experiment ($\mathbf X=\mathbf 0$) under Star models with 50 responses and 50 predictors, with and without biases on the prior of $\mathbf{B}$. Lines are averages over 100 repeats while error bars are 0.975 and 0.025 quantiles. }
    \label{fig:block_kl}
\end{figure}

\begin{figure}[H]
    \centering
    \includegraphics[width = 0.8\linewidth]{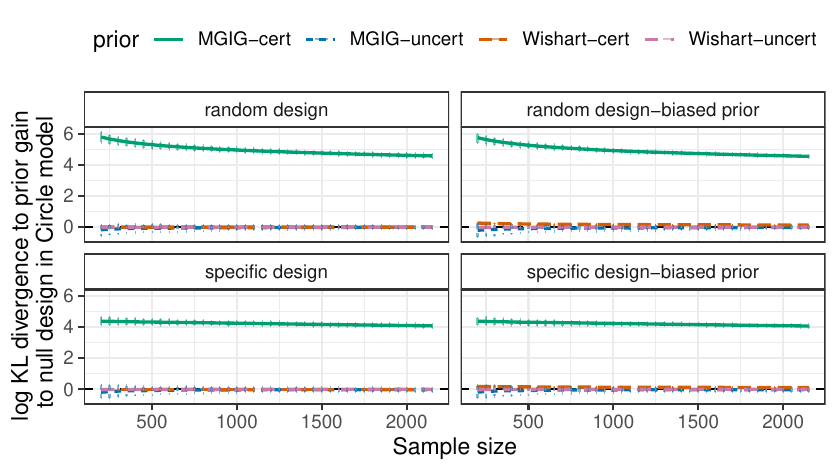}
    \caption{Difference in log KL divergence between prior and posterior comparing random experiment ($\mathbf X \ne \mathbf 0$) and specific experiment (diagonal $\mathbf X$) vs null experiment ($\mathbf X=\mathbf 0$) under circle models with 50 responses and 50 predictors, with and without biases on the prior of $\mathbf{B}$. Lines are averages over 100 repeats while error bars are 0.975 and 0.025 quantiles. }
    \label{fig:circle_kl}
\end{figure}

\begin{figure}[H]
    \centering
    \includegraphics[width = 0.8\linewidth]{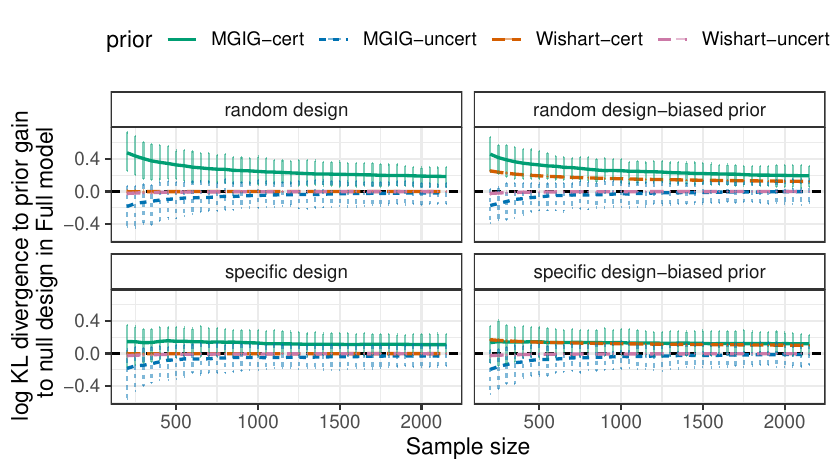}
    \caption{Difference in log KL divergence between prior and posterior comparing random experiment ($\mathbf X \ne \mathbf 0$) and specific experiment (diagonal $\mathbf X$) vs null experiment ($\mathbf X=\mathbf 0$) under Full models with 50 responses and 50 predictors, with and without biases on the prior of $\mathbf{B}$. Lines are averages over 100 repeats while error bars are 0.975 and 0.025 quantiles. }
    \label{fig:full_kl}
\end{figure}

\subsubsection*{Stein's loss compared to null experiments}

\begin{figure}[H]
    \centering
    \includegraphics[width = 0.8\linewidth]{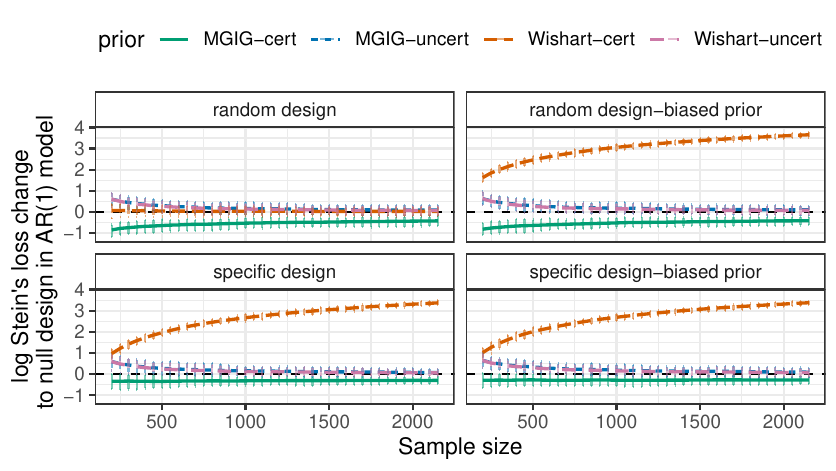}
    \caption{Difference in log Stein's loss of random experiment ($\mathbf X \ne \mathbf 0$) and specific experiment (diagonal $\mathbf X$) vs null experiment ($\mathbf X=\mathbf 0$) under AR1 models with 50 responses and 50 predictors, with and without biases on the prior of $\mathbf{B}$. Lines are averages over 100 repeats while error bars are 0.975 and 0.025 quantiles. }
    \label{fig:ar1stein}
\end{figure}

\begin{figure}[H]
    \centering
    \includegraphics[width = 0.8\linewidth]{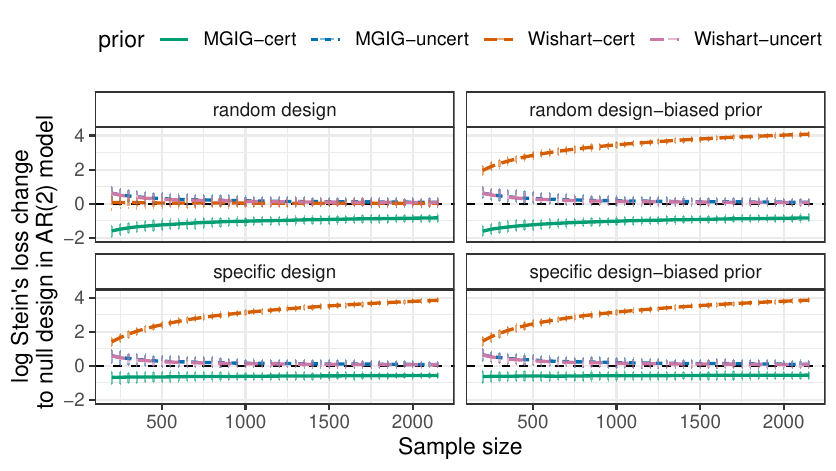}
    \caption{Difference in log Stein's loss of random experiment ($\mathbf X \ne \mathbf 0$) and specific experiment (diagonal $\mathbf X$) vs null experiment ($\mathbf X=\mathbf 0$) under AR2 models with 50 responses and 50 predictors, with and without biases on the prior of $\mathbf{B}$. Lines are averages over 100 repeats while error bars are 0.975 and 0.025 quantiles. }
    \label{fig:ar2stein}
\end{figure}

\begin{figure}[H]
    \centering
    \includegraphics[width = 0.8\linewidth]{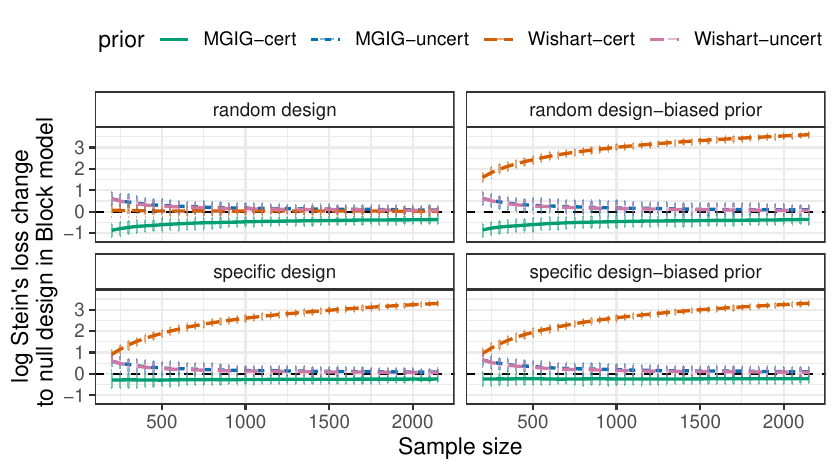}
    \caption{Difference in log Stein's loss of random experiment ($\mathbf X \ne \mathbf 0$) and specific experiment (diagonal $\mathbf X$) vs null experiment ($\mathbf X=\mathbf 0$) under Block models with 50 responses and 50 predictors, with and without biases on the prior of $\mathbf{B}$. Lines are averages over 100 repeats while error bars are 0.975 and 0.025 quantiles. }
    \label{fig:Blockstein}
\end{figure}

\begin{figure}[H]
    \centering
    \includegraphics[width = 0.8\linewidth]{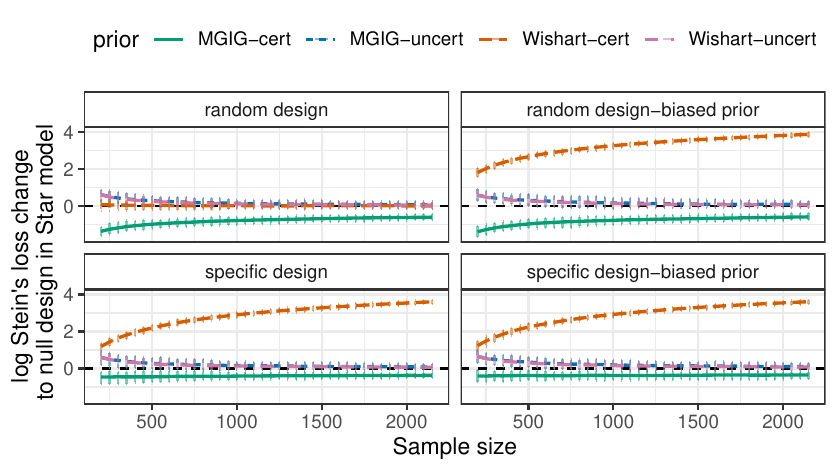}
    \caption{Difference in log Stein's loss of random experiment ($\mathbf X \ne \mathbf 0$) and specific experiment (diagonal $\mathbf X$) vs null experiment ($\mathbf X=\mathbf 0$) under Star models with 50 responses and 50 predictors, with and without biases on the prior of $\mathbf{B}$. Lines are averages over 100 repeats while error bars are 0.975 and 0.025 quantiles. }
    \label{fig:Blockstein}
\end{figure}

\begin{figure}[H]
    \centering
    \includegraphics[width = 0.8\linewidth]{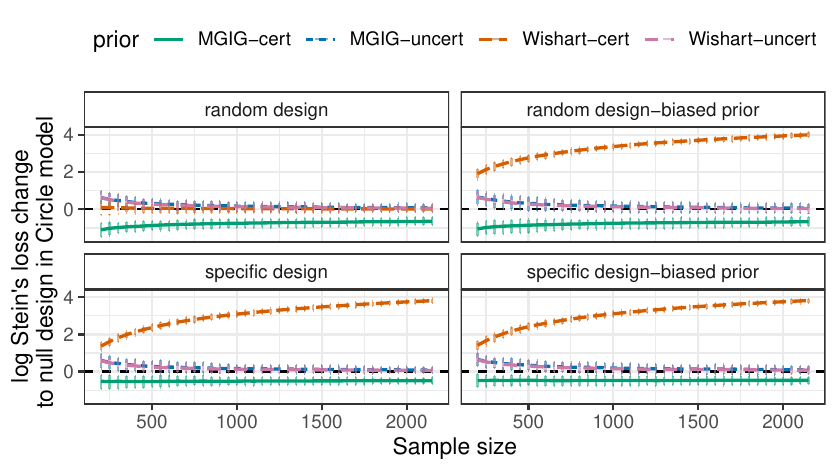}
    \caption{Difference in log Stein's loss of random experiment ($\mathbf X \ne \mathbf 0$) and specific experiment (diagonal $\mathbf X$) vs null experiment ($\mathbf X=\mathbf 0$) under Circle models with 50 responses and 50 predictors, with and without biases on the prior of $\mathbf{B}$. Lines are averages over 100 repeats while error bars are 0.975 and 0.025 quantiles. }
    \label{fig:circlestein}
\end{figure}

\begin{figure}[H]
    \centering
    \includegraphics[width = 0.8\linewidth]{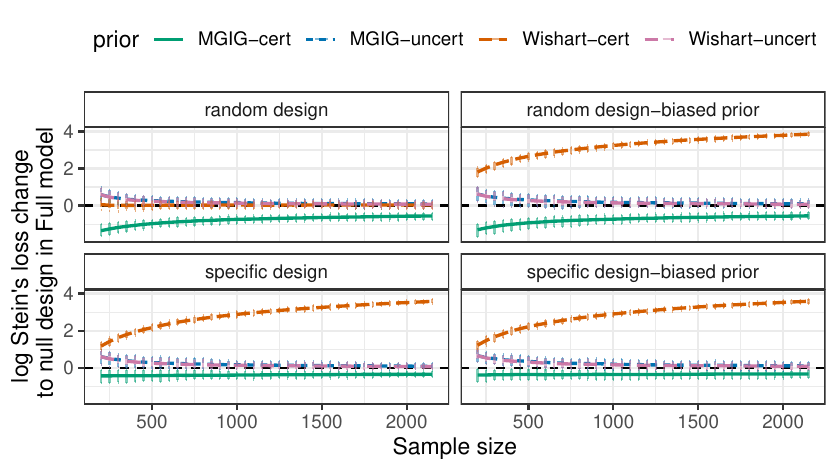}
    \caption{Difference in log Stein's loss of random experiment ($\mathbf X \ne \mathbf 0$) and specific experiment (diagonal $\mathbf X$) vs null experiment ($\mathbf X=\mathbf 0$) under Full models with 50 responses and 50 predictors, with and without biases on the prior of $\mathbf{B}$. Lines are averages over 100 repeats while error bars are 0.975 and 0.025 quantiles. }
    \label{fig:fullstein}
\end{figure}

\subsection*{Student's t-distribution tails}

\revision{The conclusions are different with t-distribution tails. Results for point estimates still showed that a good prior is needed for experiment to outperform a null experiments, but KL divergence gain is different from Gaussian chain graph. This is expected given that we have model misspecification.}

\subsubsection*{KL divergence compared to null experiments}
\begin{figure}[H]
    \centering
    \includegraphics[width = 0.8\linewidth]{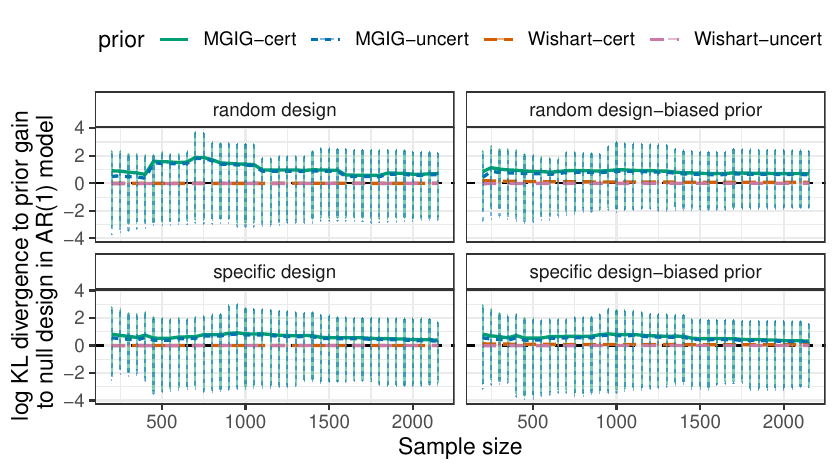}
    \caption{Difference in log KL divergence between prior and posterior comparing random experiment ($\mathbf X \ne \mathbf 0$) and specific experiment (diagonal $\mathbf X$) vs null experiment ($\mathbf X=\mathbf 0$) under AR1 models with 50 responses and 50 predictors, with and without biases on the prior of $\mathbf{B}$. Lines are averages over 100 repeats while error bars are 0.975 and 0.025 quantiles. }
    \label{fig:ar1_kl2}
\end{figure}

\begin{figure}[H]
    \centering
    \includegraphics[width = 0.8\linewidth]{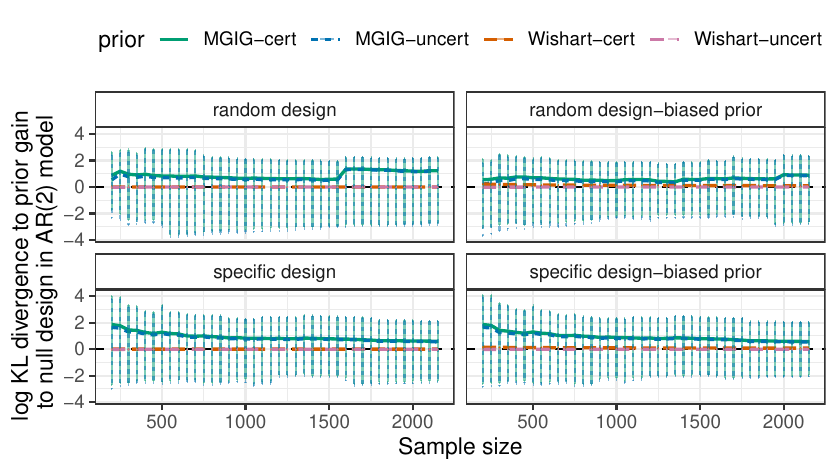}
    \caption{Difference in log KL divergence between prior and posterior comparing random experiment ($\mathbf X \ne \mathbf 0$) and specific experiment (diagonal $\mathbf X$) vs null experiment ($\mathbf X=\mathbf 0$) under AR2 models with 50 responses and 50 predictors, with and without biases on the prior of $\mathbf{B}$. Lines are averages over 100 repeats while error bars are 0.975 and 0.025 quantiles. }
    \label{fig:ar2_kl}
\end{figure}

\begin{figure}[H]
    \centering
    \includegraphics[width = 0.8\linewidth]{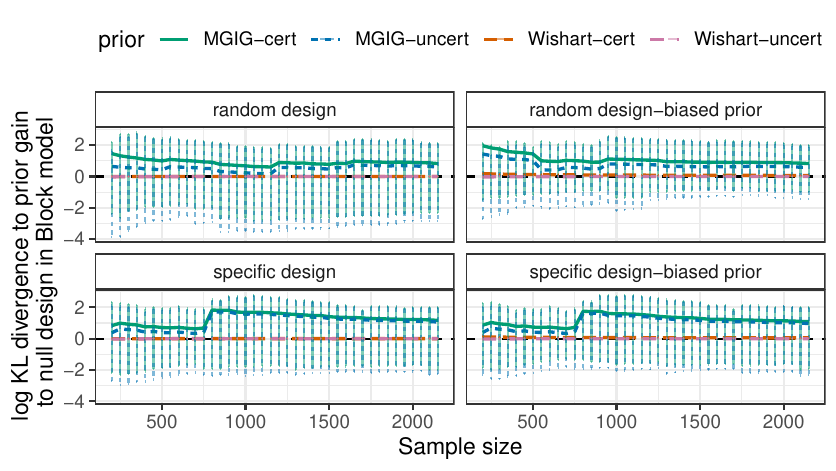}
    \caption{Difference in log KL divergence between prior and posterior comparing random experiment ($\mathbf X \ne \mathbf 0$) and specific experiment (diagonal $\mathbf X$) vs null experiment ($\mathbf X=\mathbf 0$) under block models with 50 responses and 50 predictors, with and without biases on the prior of $\mathbf{B}$. Lines are averages over 100 repeats while error bars are 0.975 and 0.025 quantiles. }
    \label{fig:block_kl}
\end{figure}

\begin{figure}[H]
    \centering
    \includegraphics[width = 0.8\linewidth]{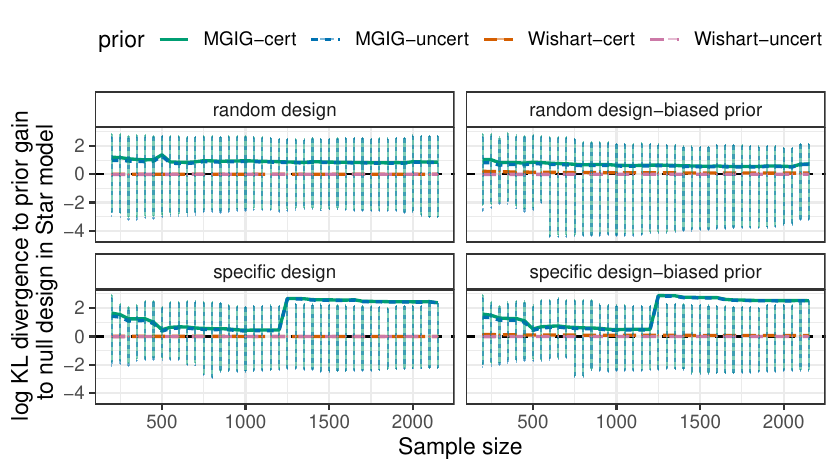}
    \caption{Difference in log KL divergence between prior and posterior comparing random experiment ($\mathbf X \ne \mathbf 0$) and specific experiment (diagonal $\mathbf X$) vs null experiment ($\mathbf X=\mathbf 0$) under Star models with 50 responses and 50 predictors, with and without biases on the prior of $\mathbf{B}$. Lines are averages over 100 repeats while error bars are 0.975 and 0.025 quantiles. }
    \label{fig:block_kl}
\end{figure}

\begin{figure}[H]
    \centering
    \includegraphics[width = 0.8\linewidth]{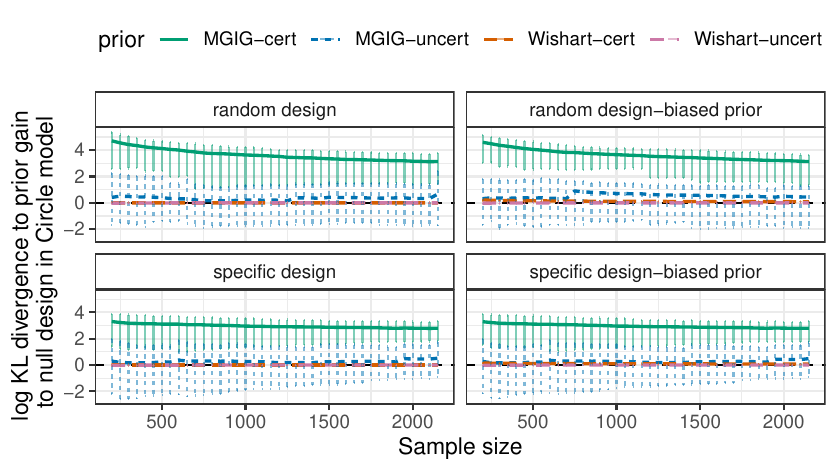}
    \caption{Difference in log KL divergence between prior and posterior comparing random experiment ($\mathbf X \ne \mathbf 0$) and specific experiment (diagonal $\mathbf X$) vs null experiment ($\mathbf X=\mathbf 0$) under circle models with 50 responses and 50 predictors, with and without biases on the prior of $\mathbf{B}$. Lines are averages over 100 repeats while error bars are 0.975 and 0.025 quantiles. }
    \label{fig:circle_kl}
\end{figure}

\begin{figure}[H]
    \centering
    \includegraphics[width = 0.8\linewidth]{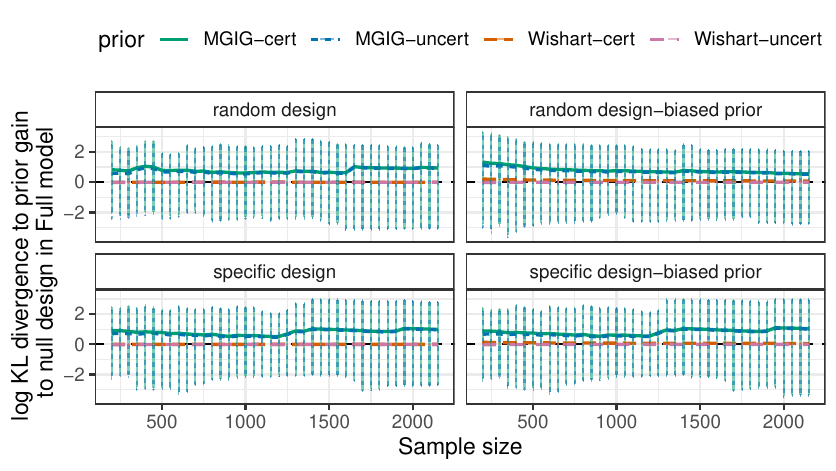}
    \caption{Difference in log KL divergence between prior and posterior comparing random experiment ($\mathbf X \ne \mathbf 0$) and specific experiment (diagonal $\mathbf X$) vs null experiment ($\mathbf X=\mathbf 0$) under Full models with 50 responses and 50 predictors, with and without biases on the prior of $\mathbf{B}$. Lines are averages over 100 repeats while error bars are 0.975 and 0.025 quantiles. }
    \label{fig:full_kl}
\end{figure}

\subsubsection*{Stein's loss compared to null experiments}

\begin{figure}[H]
    \centering
    \includegraphics[width = 0.8\linewidth]{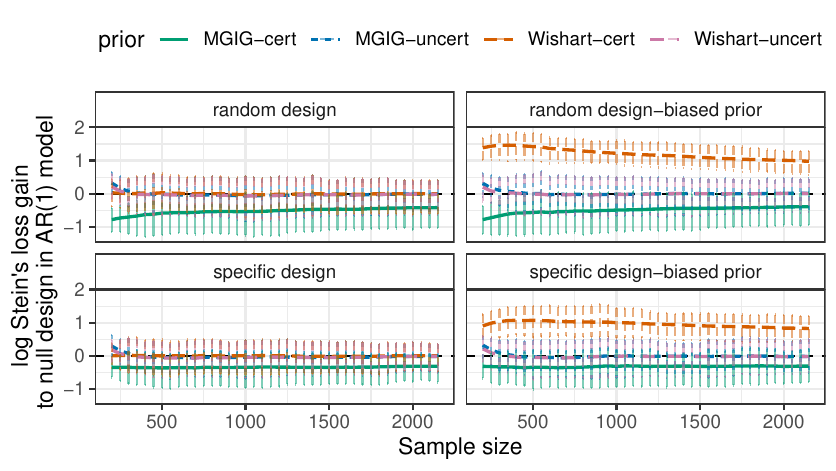}
    \caption{Difference in log Stein's loss of random experiment ($\mathbf X \ne \mathbf 0$) and specific experiment (diagonal $\mathbf X$) vs null experiment ($\mathbf X=\mathbf 0$) under AR1 models with 50 responses and 50 predictors, with and without biases on the prior of $\mathbf{B}$. Lines are averages over 100 repeats while error bars are 0.975 and 0.025 quantiles. }
    \label{fig:ar1stein}
\end{figure}

\begin{figure}[H]
    \centering
    \includegraphics[width = 0.8\linewidth]{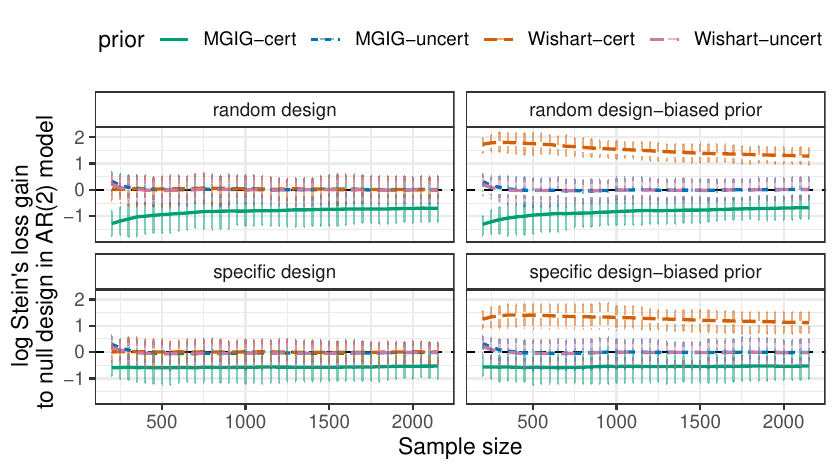}
    \caption{Difference in log Stein's loss of random experiment ($\mathbf X \ne \mathbf 0$) and specific experiment (diagonal $\mathbf X$) vs null experiment ($\mathbf X=\mathbf 0$) under AR2 models with 50 responses and 50 predictors, with and without biases on the prior of $\mathbf{B}$. Lines are averages over 100 repeats while error bars are 0.975 and 0.025 quantiles. }
    \label{fig:ar2stein}
\end{figure}

\begin{figure}[H]
    \centering
    \includegraphics[width = 0.8\linewidth]{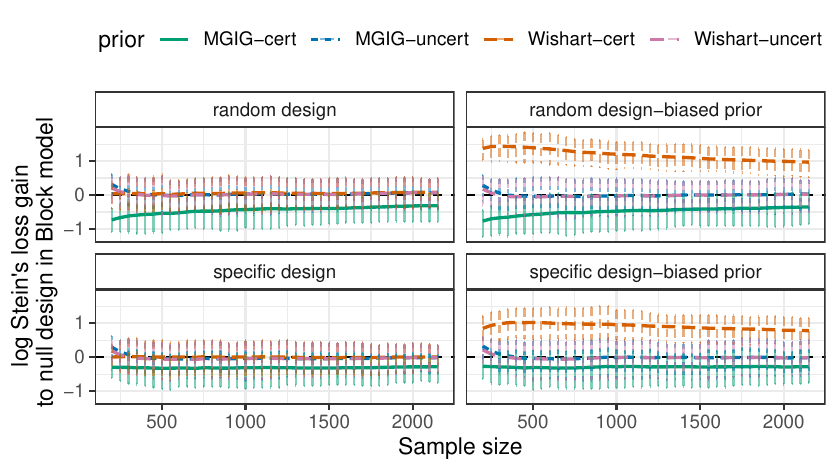}
    \caption{Difference in log Stein's loss of random experiment ($\mathbf X \ne \mathbf 0$) and specific experiment (diagonal $\mathbf X$) vs null experiment ($\mathbf X=\mathbf 0$) under Block models with 50 responses and 50 predictors, with and without biases on the prior of $\mathbf{B}$. Lines are averages over 100 repeats while error bars are 0.975 and 0.025 quantiles. }
    \label{fig:Blockstein}
\end{figure}

\begin{figure}[H]
    \centering
    \includegraphics[width = 0.8\linewidth]{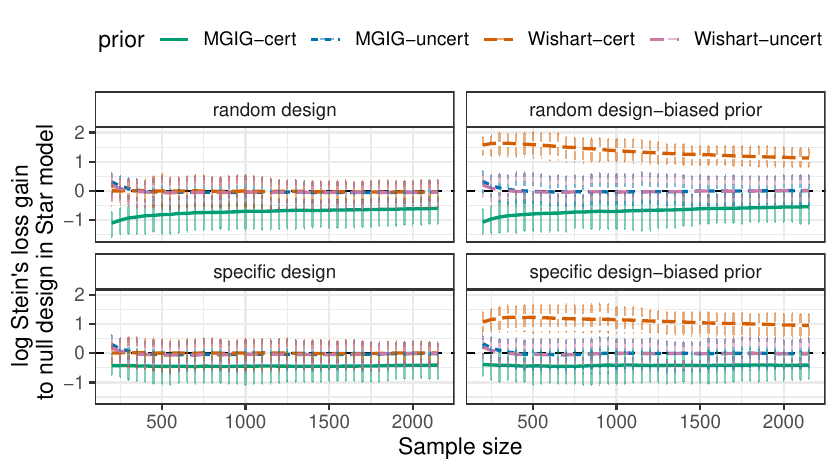}
    \caption{Difference in log Stein's loss of random experiment ($\mathbf X \ne \mathbf 0$) and specific experiment (diagonal $\mathbf X$) vs null experiment ($\mathbf X=\mathbf 0$) under Star models with 50 responses and 50 predictors, with and without biases on the prior of $\mathbf{B}$. Lines are averages over 100 repeats while error bars are 0.975 and 0.025 quantiles. }
    \label{fig:Blockstein}
\end{figure}

\begin{figure}[H]
    \centering
    \includegraphics[width = 0.8\linewidth]{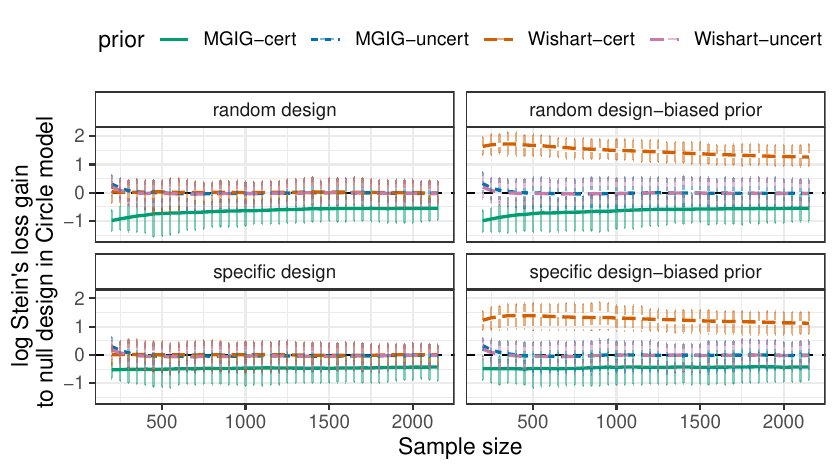}
    \caption{Difference in log Stein's loss of random experiment ($\mathbf X \ne \mathbf 0$) and specific experiment (diagonal $\mathbf X$) vs null experiment ($\mathbf X=\mathbf 0$) under Circle models with 50 responses and 50 predictors, with and without biases on the prior of $\mathbf{B}$. Lines are averages over 100 repeats while error bars are 0.975 and 0.025 quantiles. }
    \label{fig:circlestein}
\end{figure}

\begin{figure}[H]
    \centering
    \includegraphics[width = 0.8\linewidth]{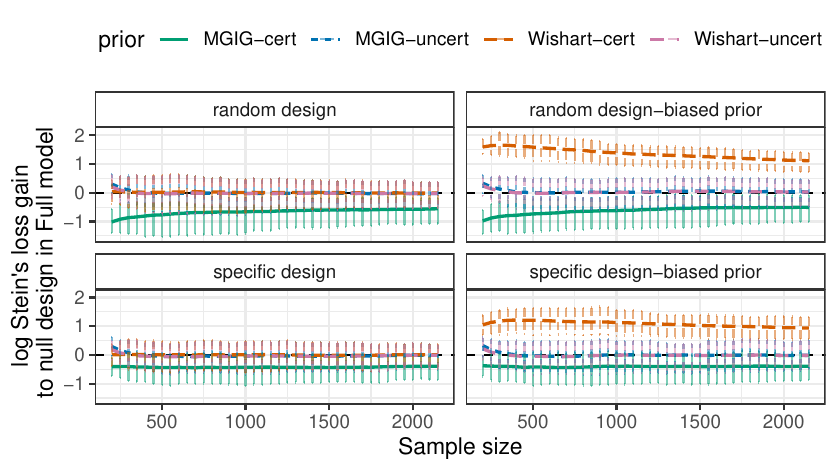}
    \caption{Difference in log Stein's loss of random experiment ($\mathbf X \ne \mathbf 0$) and specific experiment (diagonal $\mathbf X$) vs null experiment ($\mathbf X=\mathbf 0$) under Full models with 50 responses and 50 predictors, with and without biases on the prior of $\mathbf{B}$. Lines are averages over 100 repeats while error bars are 0.975 and 0.025 quantiles. }
    \label{fig:fullstein}
\end{figure}

\end{document}